\documentclass[acmsmall]{acmart}

\usepackage{tikz}
\usepackage{hyperref}
\usetikzlibrary{arrows}
\usepackage[utf8]{inputenc}
\usepackage{amsmath}
\usepackage{amsthm}
\usepackage{mathpartir} 
\usepackage{mathtools}
\usepackage{centernot} 
\usepackage{graphicx}
\usepackage{subcaption}
\usepackage{xcolor} 
\usepackage{soul}
\usepackage{float}
\usepackage[font={small,it}]{caption}
\usepackage{import}
\usepackage{calc}  
\usepackage{enumitem}  
\usepackage{environ}
\usepackage{declarations}
\usepackage{stmaryrd}
\usepackage[T1]{fontenc} 
\usepackage{xspace}
\usepackage{listings}
\usepackage{adjustbox}
\usepackage{minted}
\usepackage{array}
\usepackage{wrapfig}
\usepackage[most]{tcolorbox}
\usepackage[utf8]{inputenc}

\tcbset{
  mytheorem/.style={
    enhanced,
    colback=red!4!white,
    colframe=red!50!white,
    boxrule=0.6pt,
    arc=2mm,
    outer arc=2mm,
    fonttitle=\bfseries,
    coltitle=black,
  }
}
\tcbset{
  mylemmaobox/.style={
    enhanced,
    colback=blue!3!white,   
    colframe=blue!30!white, 
    boxrule=0.6pt,
    arc=2mm,
    outer arc=2mm,
    fonttitle=\bfseries,
    coltitle=black,
  }
}

\definecolor{lightgray}{rgb}{0.95, 0.95, 0.95}

\usepackage{tikz}
\usetikzlibrary{arrows} 
\usepackage{pstricks-add, pst-pdf}%
\usepackage{pgfplots}
\pgfplotsset{compat=1.18}
\usepackage{colortbl}
\usepackage{tabulary}
\usepackage{etoolbox}
\definecolor{codegreen}{rgb}{0,0.6,0}
\definecolor{codegray}{rgb}{0.5,0.5,0.5}
\definecolor{codepurple}{rgb}{0.58,0,0.82}
\definecolor{backcolour}{rgb}{1, 1, 0.96}
\usepackage{color}
\lstdefinestyle{mystyle}{
  backgroundcolor=\color{backcolour},
  commentstyle=\color{codegreen},
  numberstyle=\tiny\color{codegray},
  keywordstyle=\color{magenta},
  stringstyle=\color{codepurple},
  basicstyle=\ttfamily\scriptsize,
  breakatwhitespace=false,
  breaklines=true,
  captionpos=b,
  keepspaces=true,
  numbers=left,
  numbersep=5pt,
  showspaces=false,
  showstringspaces=false,
  showtabs=false,
  tabsize=2
}

\lstdefinelanguage{Rust}{
  morekeywords={pub, struct, enum, fn, impl, let, mut, ref, match, if, else, for, in, while, loop, return, break, continue, Self, super, crate, mod, use, as, where, trait},
  sensitive=true,
  morecomment=[l]{//},
  morecomment=[s]{/*}{*/},
  morestring=[b]{"},
}
\lstdefinestyle{ruststyle}{
  language=Rust,
  commentstyle=\color{codegreen},
  basicstyle=\ttfamily\scriptsize,
  numbers=left,
  numberstyle=\tiny\color{gray},
  stepnumber=1,
  numbersep=8pt,
  showstringspaces=false,
  tabsize=4,
  breaklines=true,
}

\lstdefinestyle{codecase}{ %
    language=Haskell,
    backgroundcolor=\color{pink!20},
    basicstyle=\small\ttfamily,
    keywordstyle=\small\bfseries,
    numberstyle=\small\ttfamily,
    escapeinside={\%*}{*)},
    keywordstyle=\color{blue}
}
\lstdefinestyle{codecaseour}{ %
    backgroundcolor=\color{pink!20},
    basicstyle=\small\ttfamily,
    keywordstyle=\small\bfseries,
    numberstyle=\small\ttfamily,
    keywordstyle=\color{blue}
}
\AtBeginDocument{%
  }

\setcopyright{cc}
\setcctype{by}
\acmDOI{10.1145/3839492}
\acmYear{2026}
\acmJournal{PACMPL}
\acmVolume{10}
\acmNumber{OOPSLA2}
\acmArticle{360}
\acmMonth{10}
\acmSubmissionID{oopslab26main-p841-p}
\received{2026-03-17}
\received[accepted]{2026-08-06}

\begin{document}

\title{Sound Enforcement of Dynamic Release Information Flow Policy}

\author{Jeffrey Ching}
\orcid{0009-0008-3363-2676}
\affiliation{%
  \institution{Duke University}
  \city{Durham}
  \country{USA}
}
\email{jeffrey.ching@duke.edu}

\author{Danfeng Zhang}
\correspondingauthor
\orcid{0000-0003-1942-6872}
\affiliation{%
  \institution{Duke University}
  \city{Durham}
  \country{USA}
}
\email{danfeng.zhang@duke.edu}

\begin{abstract}
Information flow analysis is the de facto method of assessing confidentiality and integrity issues. However, the widespread adoption of information flow analysis in real-world systems is still lacking, partly due to a fundamental gap between theory and practice: the dynamic nature of security concerns in real-world systems goes beyond the scope of existing techniques that assume a static policy (i.e., data secrecy does not change). 
 
Recognizing the fundamental gap, a substantial amount of research has studied various aspects of it (e.g., enabling declassification, endorsement, and revocation policies). A recent work takes a step further by formalizing a promising end-to-end policy called dynamic release that unifies prior formalizations by allowing information flow restrictions to downgrade and upgrade in arbitrary ways. However, how to soundly enforce the powerful dynamic release policy is still an open question.

In this paper, we present the first type system that enforces dynamic release policy and formally prove its soundness. More specifically, we (1) formalize a core language that enables dynamic release policy, (2) develop a type system that checks dynamic release policy, (3) develop new proof techniques and formally prove that the type system enforces dynamic release policy, and (4) implement a prototype of the type system as an extension to the Rust language, along with case studies on a conference reviewing system and Civitas.
\end{abstract}

\ccsdesc[500]{Security and privacy~Formal security models}
\ccsdesc[500]{Theory of computation~Program analysis}

\keywords{Information flow; Dynamic policy; Type system}

\maketitle

\section{Introduction}
Many security concerns in computer systems can be understood in terms of information flows. For example, it is crucial to ensure that private and untrusted data, as well as any data derived from them, never influence unintended channels (e.g., a public channel) within a computer system. For decades, noninterference~\cite{goguen1982} has been the de facto security policy for specifying and checking information flow requirements. However, a fundamental assumption of noninterference is that the confidentiality and integrity of information \emph{do not} change over time, whereas security concerns in real-world applications are rarely static, making it infeasible to safeguard them via noninterference. This critical gap prevents widespread adoption of information flow analysis in real-world systems.

Therefore, to better describe information flow restrictions in real-world applications, researchers have introduced policies that allow sensitivity to change over time, known as dynamic policies. For example, declassification policy~\cite{myers1997} allows the release of sensitive information to the public under certain conditions. Other dynamic policies have been explored, such as gradual release~\cite{askarov2007} and its extended version called tight gradual release~\cite{askarov2009}. Erasure policy~\cite{Chong2008}, on the other hand, requires public information to become more sensitive or be completely removed. It is required in applications on untrustworthy storage~\cite{askarov2015, waye2017} and secure voting systems~\cite{clarkson2008civitas}.

While prior work studies \emph{instances} of dynamic information flow policies, the field is overloaded with inconsistent terminology, formalisms, and sometimes contradictory semantics~\cite{Li2022}.
These obstacles make it difficult even for veteran researchers to follow, let alone compare, different policies. For instance, existing definitions of dynamic policies can have vastly different syntax for defining policies~\cite{sabelfeld2009, sabelfeld2005}, differing natures of security conditions, and even disagreements on whether or not a program is secure based on different variants of the same policy~\cite{broberg2015}.
Moreover, many policies lack robust enforcement mechanisms, and only a limited number, notably the erasure policy~\cite{Chong2008}, have been implemented as an extension to Jif~\cite{myers2001}.  

A recent work, Dynamic Release~\cite{Li2022}, provides a promising solution: it generalizes various kinds of dynamic policies for the first time via a novel formalization.
Beyond expressing various dynamic policies, the original work also allows apples-to-apples comparison between different dynamic policies and creates new insights on existing definitions. However, as a pioneering work, it comes without any enforcement mechanism.
In this paper, we tackle the following goal: \emph{to develop a sound and general enforcement mechanism for various kinds of dynamic policies.} While dynamic release policy provides a general platform for reasoning about various dynamic policies, we need to tackle several technical challenges for its enforcement. First, we need a metalanguage with its security specification, along with its operational semantics. A challenge in language design is to strike a balance between expressiveness and simplicity. To do so, we carefully design some language features tailored for dynamic release policy, including (1) a distinguished set of \emph{security events} along with distinguished commands to set and unset those events during program execution, (2) controlled output commands in the form of $\outcmd{\level}{\expr}{\sevent{s_0}\dots\sevent{s_k}}$, which reveal the value of $\expr$ to security level $\level$ \emph{only if} all security events $\sevent{s_0}\dots\sevent{s_k}$ are met, and (3) a novel \cod{relabel} command that performs permissive yet sound conversion between various kinds of labels. 

Second, we must enforce the specified dynamic release policy on the program. To achieve this, we develop a novel type system that primarily controls information flow at compile time. One challenge is to achieve a balance between soundness and permissiveness, given the fact that the meaning of each label can change throughout program execution. To tackle the challenge, we define two sets of static semantics that we show soundly control information flows according to dynamic release semantics. In particular, the ``flows-to'' relation specifies which label is allowed to flow to another, including flows between dynamic labels and static levels, and even flows between different kinds of policies (e.g., declassification and endorsement). The ``release-to'' relation checks if we can safely release information to a security level at a given program point. 
Thanks to the simplicity of the type system, we also implement it as an extension to the Rust language.

Third, we formally prove that the type system enforces the specified dynamic release policy for all program executions. For noninterference with a static policy, various proof techniques (e.g.,~\cite{li2017arxiv,flowcaml,volpano1996}) have been developed. However, these techniques are not applicable for dynamic release for two major reasons. First, standard noninterference proofs are built on an \emph{indistinguishability} relation on memory states: two memory states are said to be indistinguishable when they agree on the public set of variables. As the set of public variables stays the same throughout program execution, a standard proof then proceeds by showing that the indistinguishability relation is an invariance throughout program execution. However, due to the dynamic nature of dynamic release policy, the set of public variables changes over time, making it impossible to establish such an invariance. To address the challenge, we introduce a novel indistinguishability relation parameterized on a security event trace (Definition~\ref{def:indistinguishibility}). Moreover, dynamic release is built on the notion of \emph{consistency} (Definition~\ref{def:consistency_appendix}), which is absent in standard noninterference. To accommodate that, we introduce unique conditions in several lemmas and theorems to reflect when intentional leakage is allowed.
 
In summary, we make the following contributions.
\begin{itemize}
    \item We present a metalanguage together with its security specification for dynamic release policies (Section~\ref{sec:Language}). We further formalize the semantics of dynamic security labels, as well as the semantics of several new constructs tailored to dynamic release policies.
    
    \item We present a type system that employs two sets of static inference rules to reason about and verify information flows within a program (Section~\ref{sec:static_rules}).
    
    \item We develop novel proof techniques (Section~\ref{sec:endtoendproof}) to establish the soundness of our type system. In particular, we prove that any well-typed program satisfies its specified dynamic release policy across all possible executions.
    
    \item We implement a prototype of our language and its type system as an extension to the Rust language. Moreover, we incorporate all the examples presented in the paper, and port Civitas~\cite{clarkson2008civitas, juels2005coercion}, originally implemented in Jif~\cite{myers2001}, as well as a conference management system developed in Lifty~\cite{polikarpova2020}, to our Rust-based implementation (Section~\ref{sec:case_studies}). We demonstrate that these applications, which embody realistic security requirements, can be successfully verified using our enforcement mechanism, while incurring only minimal runtime overhead.
\end{itemize}

\section{Background and Overview}
\label{sec:background}

\mypara{Noninterference.}
Noninterference~\cite{goguen1982} is the de facto policy enforced by information flow control~\cite{Sabeleld2003}. Intuitively, noninterference requires that for any two runs of a program such that only secret inputs change, the program always produces the same publicly observable outputs (i.e., confidential inputs do not affect public outputs). However, noninterference is well-recognized~\cite{Sabeleld2003,cecchetti2021, guarnieri2019,hritcu2013} to be too strict for real-world applications. One fundamental gap there is that noninterference assumes that the secrecy and integrity of information are \emph{static} (i.e., they do not change). But in reality, the dynamic nature of security concerns demands information flow policies that can accommodate evolving secrecy and integrity constraints.

\mypara{Dynamic Policy}
\emph{Dynamic} policies allow the sensitivity and/or integrity of information to change throughout its lifetime. There are mainly three kinds of such dynamic policies in the literature:
\newcommand{\hlpink}[1]{\colorbox{kwpink}{#1}}
\newcommand{\kwstyleA}[1]{%
  \begingroup
    \setlength{\fboxsep}{0.2pt}
    \colorbox{lightgray!30}{\bfseries #1}%
  \endgroup
}

\newcommand{\hlblueRed}[1]{%
  \begingroup
    \setlength{\fboxsep}{0.2pt}
    \colorbox{lightgray!30}{\bfseries\textcolor{red}{#1}}%
  \endgroup
}
\newcommand{\hlblueRedtwo}[1]{%
  \begingroup
    \setlength{\fboxsep}{0.2pt}
    \colorbox{lightgray!30}{\textcolor{red}{#1}}%
  \endgroup
}

\newcommand{\kwstyleB}[1]{\bfseries #1} 
\lstdefinestyle{highlightkeywords}{
  language=C,
  basicstyle=\ttfamily\small,
  columns=fullflexible,
  keepspaces=true,
  keywordstyle=,
  deletekeywords={while, if, else},
  morekeywords=[1]{output, eventon, eventoff, relabel, if,else},
  keywordstyle=[1]{\kwstyleB},              
  morekeywords=[2]{while},
  keywordstyle=[2]{\kwstyleB},
  escapeinside={(*@}{@*)},
  moredelim=**[is][\redline]{!r!}{!r!},
}

\definecolor{kwblue}{rgb}{0.85,0.92,1}
\definecolor{kwpink}{rgb}{1,0.9,0.95}
\lstset{
  basicstyle=\small\ttfamily,
  numberstyle=\ttfamily\scriptsize,
  tabsize=4,
  captionpos=b,
  style=highlightkeywords,
  deletekeywords={while},
  breaklines=false,
  breakautoindent=false,
  postbreak=\space,
  breakindent=5pt,
  aboveskip=3pt,
  belowskip=3pt,
  belowcaptionskip=0pt,
  morecomment=[l]{//},
  mathescape=true
}

\newcommand{\boxwidth}{0.42}
\newcommand{\leftmarginvalue}{25}
\newsavebox\DecSec
\begin{lrbox}{\DecSec}
\begin{minipage}{\boxwidth \textwidth}
\begin{lstlisting}[numbers=left,xleftmargin=\leftmarginvalue 
pt,framexleftmargin=15pt,mathescape,basicstyle=\footnotesize]
(*@\hlblueblacktwo{//wbid: $\Low$}@*)
(*@\hlblueblacktwo{//bid, bid1, bid2: $\sevent{\neg release}?\High \rightarrow_t\Low$}@*)
(*@\hlblueblack{eventoff}@*)(*@\hlblueblacktwo{($\sevent{release}$);}@*)
bid = bid1 > bid2 ? bid1 : bid2;
(*@\textcolor{red}{wbid = bid; //insecure}@*)$\label{line:bid1}$
(*@\tred{output}@*)(*@\tredtwo{($\Low$, wbid, using $\sevent{release}$); //insecure}@*)
(*@\hlblueblack{eventon}@*)(*@\hlblueblacktwo{($\sevent{release}$);}@*)
wbid=(*@\hlblueblack{relabel}@*)
(bid, (*@\hlblueblacktwo{$\sevent{\neg release}?\High \rightarrow_t\Low\;\text{to} \; \Low \; \text{using} \;\sevent{release})$;}@*) $\label{line:bid2}$
output($\Low$, wbid, using $\sevent{release}$);
\end{lstlisting}
\end{minipage}
\end{lrbox}

\newsavebox\EraSec
\begin{lrbox}{\EraSec}
\begin{minipage}{\boxwidth \textwidth}
\begin{lstlisting}[numbers=left,xleftmargin=\leftmarginvalue  
pt,framexleftmargin=15pt,mathescape,basicstyle=\footnotesize]
(*@\hlblueblack{//store: $\cod{M}$}@*)
(*@\hlblueblack{//copy, credit\_card:$(\sevent{\neg trans}?\cod{M}\rightarrow_t\top)$}@*)
(*@\hlblueblack{eventoff}@*)(*@\hlblueblacktwo{($\sevent{trans}$);}@*)
copy = credit_card; $\label{line:credit0}$
// use copy to finish transaction
output($\cod{M}$, copy, using $\sevent{\absent{trans}}$);
(*@\hlblueblack{eventon}@*)(*@\hlblueblacktwo{($\sevent{trans}$);}@*)
(*@\textcolor{red}{store = credit\_card; //insecure}@*)$\label{line:credit2}$
(*@\tred{output}@*)(*@\tredtwo{($\cod{M}$, copy, using $\sevent{trans}$); //insecure}@*) $\label{line:credit1}$
\end{lstlisting}
\end{minipage}
\end{lrbox}

\newsavebox\RevSec
\begin{lrbox}{\RevSec}
\begin{minipage}{\boxwidth\textwidth}
\begin{lstlisting}[numbers=left,xleftmargin=\leftmarginvalue  
pt,framexleftmargin=\leftmarginvalue pt,mathescape,basicstyle=\footnotesize]
(*@\hlblueblack{//Alice: $\Low$}@*)
(*@\hlblueblack{//notes, book: $(\sevent{\neg return}?\Low\rightarrow_p\High)$}@*)
(*@\hlblueblack{eventoff}@*)(*@\hlblueblacktwo{($\sevent{return}$);}@*) 
notes= chap(book);
output($\Low$,notes, using $\sevent{\absent{return}}$);
(*@\hlblueblack{eventon}@*)(*@\hlblueblacktwo{($\sevent{return}$);}@*)
(*@\tredtwo{Alice = relabel      //insecure}@*)
 (*@\tredtwo{(book, return?$\Low\rightarrow_p \High\;\text{to}\; \High\; \text{using}\; \sevent{return}); \label{line:book2}$}@*)
(*@\tred{output}@*)(*@\tredtwo{($\Low$, book, using $\sevent{return}$); //insecure}@*)  $\label{line:book3}$
output($\Low$, notes, using $\sevent{o_{notes}@\Low}$);
\end{lstlisting}
\end{minipage}
\end{lrbox}


\newsavebox\BiCred
\begin{lrbox}{\BiCred}
\begin{minipage}{\boxwidth \textwidth}
\begin{lstlisting}[numbers=left,xleftmargin=\leftmarginvalue  
pt,framexleftmargin=\leftmarginvalue pt,mathescape,basicstyle=\footnotesize]
$\color{violet}{//book: (\sevent{return}?\Low\leftrightarrow_p\High)[Per]}$
$\color{violet}{//Alice: \Low, notes: \Low}$
while (start == true)
 if ($\sevent{return}$ = false)
  (*@\hlblueblack{eventoff}@*)(*@\hlblueblacktwo{($\sevent{return}$);}@*)
  notes= chapter(book); $\label{line:bibook1}$
  output($\Low$,notes); $\label{line:bibook2}$;
 else
  (*@\hlblueblack{eventon}@*)(*@\hlblueblacktwo{($\sevent{return}$);}@*)
  Alice = notes $\label{line:bibook3}$;
  (*@\textcolor{red}{Alice=}@*)(*@\tred{relabel}@*)
  (*@\tredtwo{(book, $\Low\leftrightarrow_p\High, \High,using\; \sevent{return}) \label{line:bibook4};$}@*)
  output($\Low$,notes);
\end{lstlisting}
\end{minipage}
\end{lrbox}

\begin{figure*}[b]
\centering
\small
\resizebox{\textwidth}{!}{%
\begin{tabular}{|c|c|c|}
\hline
\cellcolor{yellow!2}\usebox\DecSec 
& \cellcolor{yellow!2}\usebox\EraSec
& \cellcolor{yellow!2}\usebox\RevSec
\\ 
(i). Bidding Game
& (ii). Credit Card  
& (iii). Library System 
\\
\hline
\end{tabular}
} 
\caption{Examples of the dynamic policies in the order of declassification, erasure and revocation. The shaded lines use novel language features that we introduce in this work. Insecure code is marked in red. Predicates $\sevent{\absent{trans}}$ and $\sevent{\absent{return}}$ indicate that the corresponding security events $\sevent{trans}$ and $\sevent{return}$ are never set to true.}
\Description{}
\label{fig:dynamicpolapp}
\end{figure*}
\begin{itemize}
  \item Declassification/Endorsement~\cite{askarov2007,banerjee2008, cecchetti2021}: a declassification (resp. endorsement) policy allows the sensitivity (resp. integrity) of a piece of information to be downgraded.  For example, consider a bidding game illustrated in Figure~\ref{fig:dynamicpolapp}-i, which can release the highest bid to the public only if the bidding is over, indicated by setting \cod{release} at Line 7. The bid cannot be released before the bidding concludes, making both Line 5 and Line 6 insecure. However, after the previously sensitive bids are \emph{declassified} to the public at Line 7, the highest bid can be released to the public as shown in Line 10.

  \item  Erasure~\cite{Chong2005,Chong2008, askarov2015}: an erasure policy allows the sensitivity of a piece of information (and its derivatives) to be upgraded. One important use case of the erasure policy is to upgrade sensitivity so that no memory can hold the information anymore (effectively enforces it to be removed from the whole system). For example, consider an online payment system illustrated in Figure~\ref{fig:dynamicpolapp}-ii, which takes credit card information from customers to complete the purchase as in Lines 4--6. However, it also needs to upgrade the sensitivity of \cod{credit\_card} and \cod{copy} after Line 7 to avoid information leakage after the transaction is complete. Hence, Lines 8 and 9 are considered insecure as they leak the credit card after the transaction is complete.
  
  \item Delegation and Revocation~\cite{ferraiolo1995,myers2000}: delegation and revocation are usually expressed in a role-based system to add and remove access permission for each role respectively. For example, consider a library system illustrated in Figure~\ref{fig:dynamicpolapp}-iii, where Alice is allowed to access the book and take notes during the borrowing period. However, once the book is returned, she cannot access the book anymore as shown in Lines 7--9. While revocation also upgrades the sensitivity of information, it differs from erasure that Alice can still access information released during the borrowing period (e.g., her notes) as shown in Line 10.
\end{itemize}

\mypara{Dynamic Release Policy}  
To relieve system developers from having to understand different kinds of dynamic policies as well as the nuances among policies that appear similar on the surface but differ in important ways, a recent work by Li and Zhang~\cite{Li2022} presents a promising end-to-end policy called \emph{dynamic release} that unifies prior formalizations by allowing information flow restrictions to downgrade and upgrade in arbitrary ways. Furthermore, dynamic release is formalized in a framework that allows apples-to-apples comparisons across existing dynamic policies.

More specifically, a dynamic release policy is specified by a \emph{dynamic security label} in the form of $\cnd{cnd}? \lab_1 \diamond \lab_2$ (see  the full syntax in Section~\ref{sec:semantics_labels}), where the condition $\cnd{cnd}$ specifies when data sensitivity (or integrity) changes, and $\diamond$ specifies both the direction and kind of change. Each policy can either be \emph{transient} or \emph{persistent}, where intuitively a transient policy (resp. a persistent policy) disallows (resp. allows) information that was leaked in the past to be revealed afterwards if information flow restrictions upgrade. In this paper, we use $\rightarrow_t$ and $\leftrightarrow_t$(resp. $\rightarrow_p$ and $\leftrightarrow_p$) to denote possible directions for transient (resp. persistent) policies. We simply use arrows without subscripts when the kind of policy is irrelevant.

With dynamic release policy, we can concisely and precisely specify the information security requirements for each of the examples in Figure~\ref{fig:dynamicpolapp} as dynamic security labels (annotated in gray):
\begin{itemize}
    \item 
    In Figure~\ref{fig:dynamicpolapp}.i, we can specify the dynamic security label on bids as $\neg\sevent{release}?\High \rightarrow_t\Low$, meaning that they are initially $\High$ and hence, cannot be released during the bidding phase. However, once the bidding has ended (signaled by the condition $\sevent{release}$), the sensitivity of wbid can be downgraded to $\Low$, making it publicly available. Moreover, according to the policy, the commands from Lines 8--10 are secure as bids are already declassified at those points. But the commands on Lines 5--6 releasing the bids before the bidding period ends are insecure.
    
    \item 
    In Figure~\ref{fig:dynamicpolapp}.ii, we can specify the dynamic security label on the credit card and its copy as $\neg\sevent{trans}?~\cod{M}\rightarrow_t \top$, meaning that before the transaction is complete, the merchant can use both information at level $\cod{M}$. But once the transaction is complete, their sensitivity becomes the highest level $\top$, to ensure that they cannot be used in the rest of the system. Therefore, the commands on Lines 4--5 that take credit card information from customers to complete the purchase are secure according to the policy, while the commands on Lines 8--9, which leak credit card information after the transaction is complete, are insecure.
    
    \item 
    In Figure~\ref{fig:dynamicpolapp}.iii, we can specify the dynamic security label on \cod{book} and \cod{notes} as $\neg\sevent{return}?\Low \rightarrow_p \High$ meaning that information cannot be released after $\sevent{return}$ is set, unless the same information was released before the change (as this is a persistent policy with $\rightarrow_p$). Therefore, the command at Line 10 is secure according to the policy, since the notes were released to Alice before the book is returned. But releasing the book after it is returned on Lines 7--9 is insecure, as it was not (completely) released before the book is returned.
    
\end{itemize}
Note that Figures~\ref{fig:dynamicpolapp}.ii and~\ref{fig:dynamicpolapp}.iii both employ an upgrade policy, but differ in a subtle yet important way. In the former, information released prior to the sensitivity upgrade cannot be revealed after the upgrade, whereas in the latter such release remains permitted. This distinction also arises in other policies: for instance, gradual release~\cite{askarov2007} is a persistent policy, while cryptographic erasure~\cite{askarov2015}, along with other erasure-based approaches, follows a transient policy semantics.

\mypara{Sound Enforcement Overview} As illustrated in Figure~\ref{fig:dynamicpolapp}, dynamic release policy provides a unified framework to specify and reason about a broad range of dynamic policies. However, prior work~\cite{Li2022} formulates security as an end-to-end semantic property quantified over \emph{all possible pairs of program executions} (Section~\ref{sec:dynamic_release}). Hence, it remains an open question of how to statically determine whether a given program (e.g., the ones in Figure~\ref{fig:dynamicpolapp}) complies with its specified policy.

The goal of this work is to develop a programming language that enables programmers to (i) implement programs, (ii) specify desired dynamic release policies, and (iii) verify that the source code satisfies its claimed dynamic release policy.
So far, we have introduced several key components of the language. First, the core is an imperative language with standard constructs, including assignments, conditionals, loops, and output operations. Second, dynamic release policies are specified via annotations on variables at the beginning of the program. Third, policy changes are governed by \emph{security events}, which can be toggled on and off through the distinguished functions $\cod{eventon}$ and $\cod{eventoff}$. Finally, the language features $\cod{relabel}$ and $\cod{output}$ commands that enable permissive checking of information flows across different dynamic release labels.

To illustrate how to check the dynamic release policy in a static way, we revisit the code in Figure~\ref{fig:dynamicpolapp}.i. The bidding game has a dynamic policy $\neg\sevent{release}?\High \rightarrow_t \Low$. To check if the code satisfies this policy, we first need a sound mechanism to directly compare dynamic labels at each program point, such as Lines 8--9. This is achieved via a ``flows to'' relation $\flowsto$ on labels, which essentially defines a security lattice on dynamic release labels. Compared with (static) lattices that are typically used in information flow systems, the security lattices for dynamic release policies are more complicated as the interpretation of a label changes over time. Hence, the lattice over the ``flows to'' relation is parameterized by some \emph{facts} of security events that must hold at each program point. For example, on Lines 8--9, the label of \cod{bid}, namely $\neg\sevent{release}?\High \rightarrow_t\Low$, can safely flow to $\Low$ whenever $\sevent{release}$ has been $\true$ in the event trace (intuitively, since it has been declassified). We formalize the ``flow to'' relation on labels in Section~\ref{sec:static_rules} and prove its soundness: for any two labels $S\vdash \lab_1$, $\lab_2$ and security event trace $\strace$, $\lab_1 \flowsto_\strace \lab_2$ if and only if $\lab_2$ is \emph{always} at least as restrictive as that of $\lab_1$ on any extension of $\strace$ (Theorem~\ref{theorem:setsound}), assuming that $\strace$ meets all security events in $S$. A type system built on the relation can reject insecure assignments, such as Line 5, where the relation cannot be statically verified. 

Moreover, the command $\outcmd{\level}{\expr}{\sevent{s_0}\dots\sevent{s_k}}$ reveals the value of $\expr$ to the security level $\level$. To check if the policy on $\expr$ allows its value to be released at that point, such as Line 6 and Line 10 in Figure~\ref{fig:dynamicpolapp}.i, we formalize the ``release to'' relation on labels in Section~\ref{sec:release_to} and prove its soundness: for any label $\lab$, level $\level$ and a set of security events $S$ such that $S \vdash \lab \releaseto \level$, information with label $\lab$ can be released on a channel at level $\level$ under any execution trace that satisfies $S$ (Theorem~\ref{theorem:setsound_no_extend}). With the relation, we develop a type system that accepts the secure release in Line 10, but also rejects the insecure release in Line 6. One subtlety of checking an output command is that a persistent policy allows information revealed in the past to be released again regardless of the current policy. To do so, the security event trace also tracks a history of released information (Section~\ref{sec:commandsemantics}). 

Finally, we formally prove that the type system enforces the end-to-end dynamic release policy: for any well-typed program with annotated dynamic release labels, the knowledge gained from observing each program output does not exceed the “allowance” prescribed by the policy.
A central technical challenge in this proof arises from the dynamic nature of release labels, which causes both the security lattice and the release-to relation to evolve during execution. Consequently, existing proof techniques (e.g.,~\cite{li2017arxiv}) cannot be applied directly. We address these challenges and present a proof sketch in Section~\ref{sec:endtoendproof}, with the full formal proof provided in~\cite{full_proof}.

\section{Language}
\label{sec:Language}
We first present a simple imperative language with its security specification, along with the operational semantics and the semantics on dynamic release labels.

\begin{figure}
\centering
\[
\begin{array}{r@{\hskip 0.5em}l}
\textbf{Variables} & \x, y, z \\
\textbf{Expressions} & \expr ::= \x \mid n \mid \expr~\text{op}~\expr \\
\textbf{Levels} & \level \in \Lattice \\
\textbf{Sec. Events} (\mathbb{S}) &  \sevent{s} \mid \neg \sevent{s} \mid \sevent{o}_x{@\level} \\ 
\textbf{Commands} &
  \begin{aligned}[t]
  \command & ::= \Skip \mid \assign{\x}{\expr} \mid \while{\expr}{\command} \mid \\
        & \ifcmd{\expr}{\command_1}{\command_2}\mid \command_1;\command_2 \mid \cod{eventon}(\sevent{s})\mid\\
        & \cod{eventoff}(\sevent{s})\mid \outcmd{\level}{\expr}{\sevent{s_0}\dots\sevent{s_k}}\mid \\
        & \outcmd{\level}{\expr}{\sevent{o}_x, \sevent{s_1}\dots\sevent{s_k}} \mid\\
        & x:=\relabel{\expr}{\lab_f}{\level_t}{\sevent{s_0}\dots\sevent{s_k}}
  \end{aligned} \\
\textbf{Sec. Labels} (\mathbb{B}) & \lab,\labtwo ::= \level \mid \cnd{cnd}?~\lab~\diamond~\labtwo \\
\textbf{Conditions} & \cnd{cnd} ::= \sevent{s} \mid \sevent{\absent {s}} \mid \cnd{cnd} \land \cnd{cnd} \mid \cnd{cnd} \lor \cnd{cnd} \mid \neg \cnd{cnd} \\
\textbf{Mut. Dir.} & \diamond ::= \rightarrow_t \mid  \leftrightarrow_t\mid \rightarrow_p \mid \leftrightarrow_p \\
\end{array}\]
\vspace{-2ex}
\caption{Language syntax with security specification.}
\label{fig:while-syntax}
\end{figure}

\subsection{Syntax}
\label{sec:syntax}
The language in Figure~\ref{fig:while-syntax} provides standard features such as variables and expressions, along with commands such as branches, loops, sequential commands, and assignments. The remaining features are designed for security purposes: 
\begin{itemize}
    \item We assume a predefined lattice $\Lattice$ consists of a set of static security levels. Here, $\Lattice$ is an arbitrary security lattice, which can represent confidentiality levels, integrity levels, or a combination of both. For simplicity, though, we assume a confidentiality lattice throughout the paper. We further assume at least two levels in most examples: $\Low$ for public and $\High$ for secret, while the language supports an arbitrary lattice. The assumption is general enough to support both Denning-style lattice~\cite{denning1976} (e.g., with $\Low \sqsubset \High $), and principal/role-based model~\cite{Arden:2015csf, ferraiolo1995, myers2000} (e.g., with $\{\cod{alice,bob}\} \sqsubset \{\cod{alice}\}$, stating that information accessible to both Alice and Bob is less restrictive than information accessible to Alice only).

    \item  In the simplest case, a security \emph{label} $\lab\in \mathbb{B}$ is a security \emph{level} taken from the security lattice. In general, a security label can be \emph{dynamic}, as specified in the form of $\cnd{cnd}? \lab \diamond \labtwo$ where the condition $\cnd{cnd}$ controls whether $\lab$ or $\labtwo$ takes effect, and the mutation $\diamond$ controls the direction and kind of the transition. In particular, there are three possible directions of mutation, $\rightarrow$, $\leftarrow$ and $\leftrightarrow$. $\lab \rightarrow \labtwo$ (resp. $\lab \leftarrow \labtwo$) specifies a one-time transition from $\lab$ to $\labtwo$ (resp. from $\labtwo$ to $\lab$) when its corresponding condition $\cnd{cnd}$ evaluates to $\false$ (resp. $\true$) the first time. Since $\cnd{cnd}?\lab \leftarrow \labtwo$ is semantically identical to $\cnd{\neg cnd}?\labtwo \rightarrow \lab$, our syntax omits the $\leftarrow$ form and treats it as syntactic sugar for the latter. Direction $\leftrightarrow$ means that the transition can happen in both directions depending on the value of $\cnd{cnd}$. Moreover, there are two types of policies: transient and persistent, specified with subscripts $t$ and $p$ respectively. The main difference between them is that a persistent policy always allows to reveal information that has been revealed in the past, while a transient policy does not~\cite{Li2022}.
\end{itemize}

To strike a balance between expressiveness and enforcement, we carefully designed a few language features in Figure~\ref{fig:while-syntax}, making it substantially different from the one presented in the dynamic release paper~\cite{Li2022}:

\begin{itemize}
\item Rather than allowing the condition $\cnd{cnd}$ to be an arbitrary program expression, we restrict it to a set of distinguished, predefined \emph{security events} $\sevent{s} \in \SEvents$. These security events can be viewed as special Boolean variables, distinct from ordinary program variables such as $\x, y, z$, and may only be modified via the distinguished functions \cod{eventon} and \cod{eventoff}. This restriction enables tractable reasoning about event status within a type system.

Besides user-defined security events above, the language also supports predicates over them, such as standard logical operations on security events and $\sevent{\absent {s}}$, which holds when the security event $\sevent{s}$ has never been set to true. Moreover, to enable reasoning about whether a variable’s value has been released in the past, we introduce a new class of security events, denoted by $\sevent{o}_\x{@\level}$, indicating that the variable $\x$ has previously been released at security level $\level$.

\item Rather than associating each output channel with a dynamic release label, we restrict output commands to fixed security levels, written as $\outcmd{\level}{\expr}{\sevent{s_0}\dots\sevent{s_k}}$. As is standard, we assume that only observers at or above level $\level$ can access the value of $\expr$ produced by this command.
Furthermore, the expression $\expr$ is revealed only if all security events in the sequence $\sevent{s_0}, \dots, \sevent{s_k}$ are enabled at the time the command is executed. Otherwise, the output command has no effect (i.e., it behaves as a nop). Consistent with prior works on dynamic policies~\cite{Chong2008, askarov2007, Li2022}, we assume that information is only released with output commands in the language; assignments to variables do not reveal any information directly. 

\item We introduce a $\cod{relabel}$ command that can also ``peek'' at the current status of security events and utilize more permissive information flow reasoning rules when possible. In particular, $x:=\relabelusing{\expr}{\lab}{\level}{\sevent{s_0 \dots s_k}}$ updates $x$ to the value of $\expr$ if the events $\sevent{s_0 \dots s_k}$ are all enabled at the time the command is executed. Otherwise, the command has no effect (i.e., it behaves as a nop). Note that the command is also annotated with a security label $\lab$, and a security level $\level$; the idea is that since information only flows from $\expr$ to $x$ when all security events $\sevent{s_0 \dots s_k}$ are enabled, a type system can rely on this assumption to reason more permissively about whether a flow from $\expr$ to $x$ is secure under this assumption. 
\end{itemize}

For example, consider Lines 8--9 in Figure~\ref{fig:dynamicpolapp}.i,  where $\sevent{release}$ is a security event that turns $\true$ after bidding is complete. At runtime, we know that $\sevent{release}$ holds, and hence, the relabel command writes the value of \cod{bid} to \cod{wbid}. Statically, the type system relies on the \cod{using} annotations to determine which security events hold at each program point, which in this example is $\sevent{release}$. With the assumption, the dynamic label on \cod{wbid}, which is ${\neg\sevent{release}? \High \rightarrow_t\Low}$, can be leaked to $\Low$ as the information is already declassified whenever $\sevent{release}$ holds. 
Hence, the novel design of  $\cod{relabel}$ command allows a static type system to soundly prove security without deep knowledge of the program state at each program point. We formalize and enforce such reasoning in our type system (Section~\ref{sec:typesystem}).

One caveat is that the design also relies on a programmer, or an advanced program analysis, to properly annotate the clause after \cod{using} to preserve the original semantics of a program. We leave this as orthogonal future work, as our primary goal of this work is to develop a core language with sound enforcement of dynamic release policies, assuming that an oracle is provided.

\subsection{Semantics of Commands}
\label{sec:commandsemantics}

\begin{figure*}
\begin{mathpar}
\inferrule[OS-eventon]
{}
{\configThree{\Mem}{\cod{eventon}(\sevent{s})}{\strace} \To 
\configThree{\Mem}{\Skip}{\strace\cdot\sevent{s}}}
\quad
\inferrule[OS-eventoff]
{}
{\configThree{\Mem}{\cod{eventoff}(\sevent{s})}{\strace} \To 
\configThree{\Mem}{\Skip}{\strace\cdot \neg\sevent{s} }}
\and
\inferrule[OS-Output-Succ]
{\configTwo{\Mem}{e}\evalto v \quad
\inferrule{\forall i \in 0...k.~\\\\
\cod{eval}(\sevent{s_i},\sigma)=\true}{} \quad
\strace'=
\text{
$\begin{cases}
\strace\cdot \sevent{o}_\x\sevent{@\level}, & \text{$\expr$ is a variable $\x$} \\
\strace, & \text{otherwise}
\end{cases}$
}}
{\configThree{\Mem}{\outcmd{\level}{\expr}{\sevent{s_0}\dots\sevent{s_k}}}{\strace} \xrightarrow[]{\configThree{\level}{v}{\strace}}
\configThree{\Mem}{\Skip}{\strace'} }
\and
\inferrule[OS-Output-fail]
{\configTwo{\Mem}{\expr}\evalto v \\ 
\inferrule{\exists i \in 0...k.~\\\\
\cod{eval}(\sevent{s_i},\sigma) = \false}{}
}
{\configThree{\Mem}{\outcmd{\level}{\expr}{\sevent{s_0}\dots\sevent{s_k}}}{\strace} \xrightarrow[]{}
\configThree{\Mem}{\Skip}{\strace}}
\and
\inferrule[OS-Relabel-Succ] 
    {\configTwo{\Mem}{\expr} \Downarrow n \\ \forall i \in 0...k.~\cod{eval}(\sevent{s_i},\sigma)=\true }
    { \configThree{\Mem}{\x:=\relabel{\expr}{\lab_f}{\level_t}{\sevent{s_0}\dots\sevent{s_k}}}{\strace}\rightarrow \configThree{\Mem[\x:=n]}{\Skip}{\strace}}
\and
\inferrule[OS-Relabel-Fail] 
    {\exists i \in 0...k.~\cod{eval}(\sevent{s_i},\sigma)=\false}
    { \configThree{\Mem}{x:=\relabel{\expr}{\lab_f}{\level_t}{\sevent{s_0}\dots\sevent{s_k}}}{\strace}\rightarrow \configThree{\Mem}{\Skip}{\strace}}
\end{mathpar}
\caption{Small-Step Semantics of Selected Commands, which $\cod{eval}$ checks the value of $\sevent{s}$ on trace $\strace$.}
\Description{}
\label{fig:key_semantics}
\end{figure*}

We model a program state as memory, written as $\Mem$: a mapping from program variables to their corresponding values. For an expression $\expr$, its big-step semantics is written as $\configs{\Mem,\expr} \Downarrow n$, which means that $\expr$ evaluates to value $n$ under memory $\Mem$. For commands, we use small-step semantics since the dynamic nature of information flow policy requires step-by-step details during a program execution. The small-step semantics is written as:$$ \configThree{\Mem}{\command}{\strace}\xrightarrow{\configThree{\level}{v}{\strace}} \configThree{\Mem'}{\command'}{\strace'} $$where $\strace$ and $\strace'$ are sequences of triggered security events (e.g., $\sevent{s_1}\cdot \neg \sevent{s_2}\cdot \sevent{o}_\x\sevent{@\Low}$).

Each step emits a (potentially empty) output event $\configThree{\level}{v}{\strace}$, which indicates the output of value $v$ to an output channel at level $\level$ under the security event trace $\sigma$. We assume that all security events are initialized. By default, the condition is set to true in right-arrow and bi-directional policies. Only output commands reveal information to entities and, hence, generate output events, as formalized in Rule~\ruleref{OS-Output-Succ} in Figure~\ref{fig:key_semantics}. Moreover, changes to security events are tracked in the security event trace, by rules \ruleref{OS-EventOn}, \ruleref{OS-EventOff} and \ruleref{OS-Output-Succ}. 

It is worth noting that both the \cod{output} and \cod{relabel} commands are syntactically associated with \cod{using} clauses. These \cod{using} annotations specify sufficient conditions on security events under which the corresponding commands are considered secure. For example, the output command at Line 6 in Figure~\ref{fig:dynamicpolapp}.i can reveal \cod{wbid} to the public \emph{only} when $\sevent{release}$ is set.
Accordingly, when these \cod{using} clauses are satisfied, the corresponding commands behave identically to their standard semantics. Otherwise (namely, when any security event appearing after the \cod{using} keyword evaluates to $\false$) the command instead behaves as a \cod{nop}. For example, the output command at Line 6 in Figure~\ref{fig:dynamicpolapp}.i produces no outputs due to Rule~\ruleref{OS-Output-Fail}, as $\sevent{release}$ is not set at that point.

As we elaborate in Section~\ref{sec:typesystem}, a type system can statically verify that these programmer-provided annotations are sufficient conditions for end-to-end security. Hence, the \cod{using} annotations affect the observable behavior of the original program only in two cases: either the annotations are overly restrictive, preventing an otherwise secure command from executing, or the corresponding commands in the original program would violate the security policy (e.g., Line 6 in Figure~\ref{fig:dynamicpolapp}.i) and must therefore be suppressed.

Finally, as introduced informally earlier, the $\cod{relabel}$ command has the same semantics as a normal assignment $\assign{x}{\expr}$ when all the security events' values are as expected; otherwise, the command is the same as a nop instruction. The semantics of other commands are standard and are included in~\cite{full_proof}.

One program execution under initial state $\configThree{\Mem_0}{c_0}{\strace_0}$ produces an \emph{execution trace} $\trace$ as follows:
{\small
\[
\configThree{\Mem_0} {\command_0}{\strace_0} 
\xrightarrow{\configThree{\level_0}{v_0}{\strace_0}} \configThree{\Mem_1} {\command_1}{\strace_1}  \cdots 
\xrightarrow{\configThree{\level_{n-1}}{v_{n-1}}{\strace_{n-1}}}  \configThree{\Mem_n}{\command_n}{\strace_n}
\]
}

We use $\configThree{\Mem}{\command}{\strace} \hookrightarrow \trace$ to denote that execution from $\configThree{\Mem}{\command}{\strace}$ produces an execution trace $\trace$. We use $ \trace^{[i]} $ to denote the configuration after the $i$-th evaluation step in $\traceout{\tau}$, $\traceout{\tau^{[:i]}}$ to denote a sub-trace from the initial point to the i-th evaluation step, and $ \traceout{\tau^{[i:j]}}$ to denote a sub-trace between $i$-th and $j$-th evaluation steps. We use $ \len{\traceout{\tau}} $ to denote the number of evaluation steps in the trace. Since $\strace$ is a sequence of events, we use similar annotation for its subsequences. 

\subsection{Semantics of Security Labels}
\label{sec:semantics_labels}

\begin{figure}
\centering
\begin{equation*}
\auxfuncold{\level}{\strace} = \level
\end{equation*}
\begin{equation*}
\auxfuncold{\cnd{cnd}? \lab\rightarrow \labtwo}{\strace} =
\begin{cases}
    \auxfuncold{\lab}{\strace} & \cod{first}(\cnd{cnd},\false,\strace) = -1\\
    \auxfuncold{\labtwo}{\strace^{[i:]}} & i = \cod{first}(\cnd{cnd},\false,\strace)
\end{cases}
\end{equation*}
\begin{equation*}
\auxfuncold{\cnd{cnd}? \lab\leftrightarrow \labtwo}{\strace} =
\begin{cases}
    \auxfuncold{\lab}{\strace^{[i+1:]}} & i = \cod{last}(\cnd{cnd},\false,\strace) \not= \len{\strace}\\
    \auxfuncold{\labtwo}{\strace^{[i+1:]}} & i = \cod{last}(\cnd{cnd},\true,\strace) \not= \len{\strace}
\end{cases}
\end{equation*}
\vspace{-2ex}
\caption{Label semantics.}
\label{fig:label-semantics}
\label{transient_policy_label_semantics}
\end{figure}

To interpret dynamic release labels, we treat each label $\cnd{cnd}?\lab\diamond \labtwo$ as an ``expression'' and evaluate its value (i.e., a security level from a predefined lattice in our setting) at the end of a trace $\traceout{\strace}$ as shown in Figure~\ref{fig:label-semantics}. In the definition, taken from~\cite{Li2022},
the helper function \cod{first} returns the first index of $\strace$ such that $\cnd{cnd}$ evaluates to the second parameter, or -1 if no such index is found; and \cod{last} returns the last index of $\strace$ such that $\cnd{cnd}$ evaluates to the second  parameter, or -1 if no such index is found. As for $\leftrightarrow$, the labels may transit an arbitrary number of times whenever the value of $\cnd{cnd}$ changes, which is mainly used in the revocation dynamic policy.

Intuitively, $\auxfuncold{\lab}{\strace}$ computes the \emph{least} restrictive security level that $\lab$ can flow to at the end of $\strace$. More specifically, for a static level $\level$, its interpretation is simply $\level$. For dynamic labels like $\cnd{cnd}?\lab \rightarrow \labtwo$, a one-time sensitivity change from $\lab$ to $\labtwo$ occurs when $\cnd{cnd}$ first evaluates to $\false$. Let $i$ be the first index in $\strace$ such that $\cnd{cnd}$ evaluates to $\false$. $\auxfuncold{\cnd{cnd}?\lab \rightarrow \labtwo}{\strace}$ reduces to $\auxfuncold{\lab}{\strace}$ if no such $i$ exists; otherwise, it reduces to $\auxfuncold{\labtwo}{\strace^{[i:]}}$. Bi-directional labels are interpreted based on the last configuration of $\strace$, where $i$ is the last index at which $\cnd{cnd}$ evaluates to $\false$. Therefore, $i \not= \len{\traceout{\strace}}$ implies that $\cnd{cnd}$ evaluates to $\true$ at the end of $\strace$, and the label reduces to $\lab$. The label $\lab$ is then evaluated under $\strace^{[i+1]}$, ensuring that any potentially nested conditions are handled properly.

The interpretation of the labels above specifies when and where information can be released to: information with label $\lab$ can be released to entities at or above level $\auxfuncold{\lab}{\strace}$ at the end of $\strace$. However, it is not suitable for information flow enforcement, since it does not allow direct comparison of dynamic labels at each program point. For example, label $\neg\sevent{release}?\High\rightarrow_t \Low$ (declassification policy) and label $\High$ both evaluate to $\High$ when the security event $\sevent{release}$ evaluates to $\false$ under $\strace$. However, information flow from $\High$ to $\neg\sevent{release}?\High\rightarrow_t \Low$ is obviously insecure as if the flow happens, the former is effectively declassified once $\sevent{release}$ becomes $\true$.

To address the challenge above, we observe that for a secure flow from label $\lab$ to $\labtwo$ under an event trace $\strace$, we need to examine all possible \emph{extensions} of $\strace$ so that $\lab$ is always bounded by $\lab'$ in the future (i.e., $\forall \strace \preceq \strace'.~ \auxfuncold{\lab}{\strace'} \LEQ  \auxfuncold{\lab'}{\strace'}$).  
Consequently, we write $\lab \sqsubseteq_\sigma \lab'$ to denote that $\lab'$ is at least as restrictive as $\lab$ at $\sigma$: that is, for all future extensions of $\sigma$, wherever information with label $\lab'$ can be released, information with label $\lab$ can also be released:
\begin{definition}[Partial Ordering on Dynamic Release Labels]
\label{def:partial_dynamic}
\[\forall \lab, \lab', \strace.~\lab \LEQ_\strace \lab' \iff \forall \strace'.~(\strace \preceq \strace'\Rightarrow \auxfuncold{\lab}{\strace'} \LEQ  \auxfuncold{\lab'}{\strace'})\]
\end{definition}

Note that with dynamic release labels, the ordering on labels is parameterized on a security event trace as the relation changes during program execution. This contrasts with standard noninterference, where the security lattice remains fixed throughout program execution. For example, consider two security event traces $\strace_1$ and $\strace_2$ ($\strace_1\preceq\strace_2$), where $\sevent{release}$ is always $\false$ in $\strace_1$, but it turned $\true$ in $\strace_2$. It is easy to check that ($\neg\sevent{\sevent{release}}?\High\rightarrow \Low \not\LEQ_{\strace_1} \Low$) and ($\neg\sevent{\sevent{release}}?\High\rightarrow \Low \LEQ_{\strace_2} \Low$), meaning that the declassification label cannot flow to $\Low$ under $\strace_1$ but it can flow to $\Low$ under $\strace_2$. On the other hand, ($\High \not\LEQ_{\strace_1} \neg\sevent{\sevent{release}}?\High\rightarrow \Low$) and ($\High \not\LEQ_{\strace_2} \neg\sevent{\sevent{release}}?\High\rightarrow \Low$), meaning that $\High$ can never flow to information with the declassification label.

\mypara{Persistent vs. Transient Label.} Our partial ordering definition on dynamic release labels (Definition~\ref{def:partial_dynamic}) does not distinguish between persistent and transient labels. This might be surprising as intuition might suggest that a persistent policy is less restrictive than its counterpart of a transient policy (i.e., $\sevent{s}?\lab\rightarrow_t \labtwo$ is strictly more restrictive than $\sevent{s}? \lab\rightarrow_p \labtwo$). However, we found this intuition to be incorrect, as the difference between persistent and transient labels only shows up in the end-to-end security condition of dynamic release, rather than where information can be released to at each program point. We will elaborate on their differences in the type system (Section~\ref{sec:typesystem}).

\subsection{Dynamic Release Policy}
\label{sec:dynamic_release}
Dynamic release~\cite{Li2022} is the first end-to-end information flow policy that generalizes multiple kinds of  dynamic information flow policies. Its precise definition is a bit involved and we only present the essentials in this section to provide sufficient background information to understand how to soundly enforce it, which is the main contribution of this paper. We follow the notations in the original paper, other than a few minor changes for readability.

Informally, a program $\command$ satisfies dynamic release policy iff for any output $t$ produced by $\command$, the \emph{knowledge} gained about any variable $\x$ by observing $t$ is bounded by what's \emph{allowed} by the policy of $\x$ at the program point that outputs $t$. 

To formalize this, we first define ``indistinguishable'' memory set w.r.t. some memory $\Mem$ and a set of variables $X$ as the set of memories that are equivalent to $\Mem$ on the values of all variables in $X$.
\begin{definition}[Memory Closure] Given a memory $\Mem$ and a set of
variables $X$, the memory closure of $\Mem$ on $X$ is the set of memories that agree with $\Mem$ on all variables in $X$:
\label{def:mem-eq}
\begin{align*}
\closure{\Mem}_{X} \triangleq \{\Mem' \mid \forall x\in X.~\Mem(x)=\Mem'(x)\}
\end{align*}
\end{definition}
Given a memory $m$ and a security label $\lab$, we write $\llbracket \Mem \rrbracket_{\neq \lab}$ as the set of memories that agree with $\Mem$ on all variables except those whose security label is $\lab$. That is, $\llbracket \Mem \rrbracket_{\neq \lab} \triangleq \llbracket \Mem \rrbracket_{\{x \mid \Gamma(x) \neq \lab\}}$. Intuitively, only the values of variables with label $\lab$ may differ between $\Mem$ and any $\Mem' \in \llbracket \Mem \rrbracket_{\neq p}$.

In addition to memory, we need to utilize a filter function on output traces, in order to precisely capture what an attacker can learn from a trace: the filter function returns $\false$ for irrelevant outputs on a trace. With any trace filter $f$, we define a trace projection as follows:
\begin{definition}[Projection of Trace Filter] 
\label{def:proj}
\begin{equation*}
\proj{\trace}_f \triangleq \configs{ \configs{l,v,\strace} \in 
    \trace~\mid~f(l,v,\strace)}
\end{equation*}
\end{definition}

To formalize the attacker's observation, we define a trace projection up to level $\level$, $\proj{\trace}_\level \triangleq \proj{\trace}_{\lambda l, v, \strace.~l \leq \level} $ that takes in the attacker's level $\level$ and returns a sub-trace of output events that are visible to the attacker.

\mypara{Attacker's Knowledge from an Output Trace}
Consider an attacker at security level $\fixedlevel$. The attacker is able to observe all output events $\configThree{\level'}{v}{\strace}$ such that $\level'\LEQ \fixedlevel$. Hence, two traces $\trace_1$ and $\trace_2$ are \emph{indistinguishable} to the attacker, written as $\trace_1\sim_\fixedlevel \trace_2$, if they agree on all outputs at or below attacker's level $\fixedlevel$. More formally,
\begin{definition}[Indistinguishability] 
\label{def:Indistinguishability_appendix}

\begin{align*}
\sim_\level\ \triangleq \{(\trace_1,\trace_2)~|~\proj{\trace_1}_\level \preceq
\proj{\trace_2}_\level\}
\end{align*}
\end{definition}

Therefore, the ``knowledge'' gained by observing an output trace $\tau$ produced by a program $\command$ is defined as all initial memory states that can produce an indistinguishable trace from $\trace$:
\begin{definition}[Knowledge Gained from Output Trace] 
\label{def:Know_gain_indis}
\begin{align*}
k_1(c, \trace, \level) \triangleq \{ \Mem~\mid~
\configThree{\Mem}{\command}{\emptyset} \hookrightarrow \trace' \AND \trace\sim_\level
\trace'\}
\end{align*}
\end{definition}

Intuitively, the knowledge set consists of all initial memory states that remain possible given the portion of the execution trace observable to the attacker. Moreover, the smaller the knowledge set, the more information the attacker has learned from the program's execution.

By definition, the knowledge set monotonically decreases as additional events are produced along the execution trace, making it well suited for formalizing declassification policies~\cite{askarov2007}. In contrast, it is inadequate for erasure policies, which require reasoning about increases in uncertainty, or equivalently, expanding the knowledge set. Consider the credit-card example in Figure~\ref{fig:dynamicpolapp}-ii. After the first output, the attacker learns the credit card number, so the knowledge set contains only those initial memories with the actual credit card number. After Line 7, however, the credit card number should no longer be inferable. Accordingly, the knowledge set should again include all possible credit card numbers, reflecting the restored uncertainty required by the erasure policy.
To accommodate both kinds of dynamic policies within a unified framework, Dynamic Release reasons about the knowledge gained from each \emph{individual} output event, rather than from the entire execution trace. We introduce this notion next.

\mypara{Attacker's Knowledge from the Last Event} 
While the attacker's security level remains fixed throughout program execution, the meaning of a dynamic label may change over time.
To precisely define what can be revealed at each program point, a major challenge solved by~\cite{Li2022}, we begin with one key observation: events produced by two program executions are not, in general, aligned according to their positions in the corresponding execution traces. The reason is that due to downgrading, the publicly observable outputs of different executions may have different lengths (Observation 2 in~\cite{Li2022}). To address this issue, we introduce the \emph{secret projection} of an execution trace, parameterized by an attacker at security level $\level$ and a secret labeled $\lab$. The secret projection retains only the \emph{effective outputs}, namely, those outputs produced while information labeled $\lab$ is considered secret from an attacker at level $\level$.

\begin{definition}[Secret Projection of Trace]
\label{def:secretproj_modified_appendix}
Given a policy $\lab$ and an attacker at level $\level$, the secret projection of a trace is the subtrace containing outputs visible at level $\level$ that occur while $\lab$ remains secret from an attacker at level $\level$:
\begin{align*}
{\proj{\trace}_{\lab, \level}~\triangleq \proj{\trace}_{\lambda  l, v,\strace.~l\LEQ \level~\AND ~\auxfuncold{\lab}{\strace}\not\LEQ\level}} 
\end{align*}
\end{definition}

Using effective outputs, we introduce a notion of trace consistency. Two execution traces, $\trace_1$ and $\trace_2$, are said to be \emph{consistent} with respect to an attacker at security level $\level$ and a secret labeled $\lab$ if their effective outputs (1) have the same length, and (2) end with the same output event.

\begin{definition}[Consistency] 
\label{def:consistency_appendix}
Two output sequences $\trace_1$ and $\trace_2$ are consistent w.r.t. a policy $\lab$ and an attack level $\level$, written as $\trace_1 \equiv_{\lab,\level} \trace_2$ if 
\begin{align*}
n=\len{\proj{\trace_1}_{\lab,\level}}=\len{\proj{\trace_2}_{\lab,\level}} \land
\proj{\trace_1}_{\lab,\level}^{[n]}=\proj{\trace_2}_{\lab,\level}^{[n]}
\end{align*}
\end{definition}

With the above definitions, we can build the definition of the knowledge gained by the attacker at level $\level$ by observing the last event. 

\begin{definition}[Attacker's Knowledge Gained from the Last Event] 
\label{def:Know_gain_last}
For an attacker at level $\level$, the attacker's knowledge w.r.t. information with policy $\lab$, after observing the last event of an output sequence traces program $c$, is the set of all initial memories that produce an output sequence that is indistinguishable to some consistent counterpart of the trace $\trace$: 
\begin{align*}
\small
k_2(\command, \trace, \level, \lab) = \bigcup_{\exists. \Mem', \strace, j.~\configThree{\Mem'}{\command}{\strace} \termout 
\trace' \AND \trace'^{[:j]} \equiv_{\lab,\level}  \trace } k_1(\command, \trace'^{[:j]}, \level )
\end{align*}
\end{definition}

Finally, we define dynamic release security. For every output event produced during program execution, the attacker's knowledge gained from observing that output (Definition~\ref{def:Know_gain_last}) must be bounded by the policy allowance at the corresponding program point. For a transient policy, the policy allowance is simply $\llbracket m \rrbracket_{\neq \lab}$, where only the values of variables labeled $\lab$ are treated as secret. For a persistent policy, the allowance is instead defined as $\llbracket m \rrbracket_{\neq \lab} \cap k_1(c, \tau^{[:i-1]}, L)$, thereby permitting the attacker to retain any knowledge acquired through previous outputs while still preventing any additional information beyond what the policy allows from being learned.

\begin{definition}[Dynamic Release]
\label{def:newpolicy}
\begin{multline*}
\forall \Mem, \level\subseteq \mathbb{\level}, \lab \in  \mathbb{B}, \strace,
\trace.~\configThree{\Mem}{c}{\strace} \termout \trace  
\Rightarrow 
\forall 1\leq i\leq \len{\trace}.\quad \\
k_2(c,\trace^{[:i]}, \level, \lab) 
\supseteq 
\begin{cases} 
\closure{\Mem}_{\neq \lab}, & \text{transient} \\
\closure{\Mem}_{\neq \lab} \cap k_1(c, \trace^{[:{i-1}]},\level), & 
\text{persistent}
\end{cases}
\end{multline*}
\end{definition}

\section{Type System}
\label{sec:static_rules}
To enforce dynamic release policy, we develop a static type system to control information flow at compile time. One challenge of doing so is to achieve a balance between soundness and permissiveness, given the fact that the meaning of each label can change throughout program execution.
To address this challenge, we first introduce two static relations, each parameterized by $S$, a set of statically computed facts about security events that must hold at each program point. The relations approximate when information can securely flow between dynamic labels, and be released respectively at compile time. Built on them, we define a type system that enforces the end-to-end dynamic release policy on the whole program. 

\subsection{Oracle on Security Events}
Since the interpretation of dynamic release labels relies on a security event trace $\strace$, the type system inevitably requires knowledge of the existence or absence of security events up to some program point. Moreover, for a persistent policy, dynamic release also depends on the outputs in the past (Definition~\ref{def:newpolicy}).
To decouple the reasoning of security events and outputs from information flow enforcement, we follow a simple design where the type system only uses the facts established by the $\cod{using}$ keyword in the $\cod{output}$ and $\cod{relabel}$ commands. For example, a set of facts $S=\{\sevent{s_1},\neg \sevent{s_1}, \neg\sevent{s_2}, \sevent{\absent{s_3}}, \sevent{o}_\x\sevent{@\Low}\}$ states that up to the program point of interest, $\sevent{s_1}$ must have been both $\true$ and $\false$, $\sevent{s_2}$ must have been $\false$ (and may be $\true$), $\sevent{\absent{s_3}}$ has never been $\true$, and the value of $x$ has been released to level $\level$ already. These facts are checked by commands such as $\outcmd{\Low}{x}{\sevent{s_1},\neg \sevent{s_1}, \neg\sevent{s_2}, \sevent{\absent{s_3}}, \sevent{o}_\x\sevent{@\Low}}$ when any condition evaluates to false, the command behaves like a nop command (Section~\ref{sec:commandsemantics}). In general, we can also estimate $S$ statically, such as using dataflow analysis~\cite{alfred2007compilers}, which we leave as future work. Hereafter, we use $S\vDash \sevent{s_2}$ to indicate that $\sevent{s_2}$ must have been $\true$ given $S$, and $S\vDash \absent{\sevent{s_2}}$ to indicate that $\sevent{s_2}$ must never have been $\true$ given $S$.

\subsection{Static Rules for Flows-to Relation}
\label{sec:flows_to}
\begin{figure}
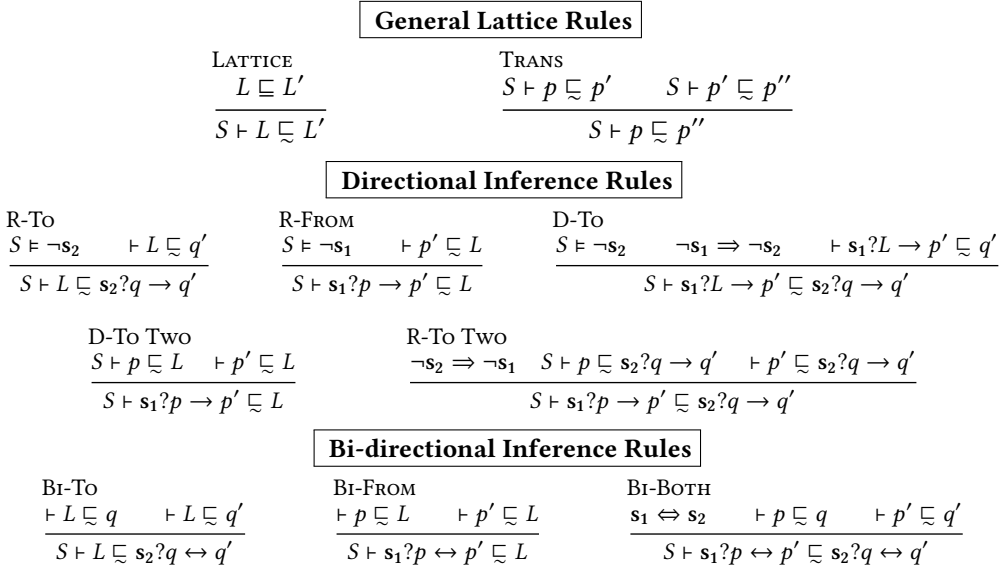

\framebox{
\textbf{General Lattice Rules}}
{
\begin{mathpar}
\mprset{flushleft} 
\inferrule[Lattice] 
  { \level\LEQ \level' }
  {S\vdash \level \flowsto\level' }
\and
\inferrule[Trans]
  {S\vdash \lab\flowsto \lab' \\ S\vdash \lab'\flowsto \lab''}
  {S\vdash \lab\flowsto \lab''} 
\end{mathpar}
}

\framebox{
\textbf{Directional Inference Rules}}
{\small
\begin{mathpar}
\mprset{flushleft}
\inferrule[R-To] 
  { S\vDash \neg \sevent{s_2} \\  \vdash \level\flowsto \labtwo'}
  { S\vdash \level\flowsto \sevent{s_2}?\labtwo\rightarrow \labtwo'}
\and
\inferrule[R-From] 
  { S\vDash \neg \sevent{s_1} \\  \vdash \lab'\flowsto \level}
  { S\vdash \sevent{s_1}?\lab\rightarrow\lab'\flowsto \level}
\and
\inferrule[D-To] 
  {S\vDash \neg\sevent{s_2} \\\neg\sevent{s_1} \Rightarrow \neg\sevent{s_2}  \\  
  \vdash \sevent{s_1}?\level \rightarrow \lab' \flowsto \labtwo'}
  { S\vdash \sevent{s_1}?\level\rightarrow \lab' \flowsto \sevent{s_2}?\labtwo\rightarrow \labtwo'}
\and
\inferrule[D-To Two] 
{ S\vdash \lab\flowsto \level \quad \vdash \lab'\flowsto \level}
  {S\vdash \sevent{s_1}?\lab\rightarrow \lab'\flowsto \level}
\and
\inferrule[R-To Two] 
{\neg\sevent{s_2} \Rightarrow \neg\sevent{s_1}\quad S \vdash \lab\flowsto \sevent{s_2}?\labtwo\rightarrow \labtwo' \quad \vdash \lab'\flowsto \sevent{s_2}?\labtwo\rightarrow \labtwo'}
  {S\vdash \sevent{s_1}?\lab\rightarrow \lab'\flowsto \sevent{s_2}?\labtwo\rightarrow \labtwo'}
\end{mathpar}
}
\framebox{
\textbf{Bi-directional Inference Rules}}
{\small
\begin{mathpar}
\mprset{flushleft}
\inferrule[Bi-To] 
  { \vdash \level \flowsto \labtwo \\ \vdash\level \flowsto \labtwo'}
  {S\vdash  \level \flowsto \sevent{s_2}? \labtwo \leftrightarrow \labtwo'}
\and
\inferrule[Bi-From] 
  { \vdash \lab\flowsto \level  \\ \vdash \lab'\flowsto \level}
  {S\vdash \sevent{s_1}?\lab \leftrightarrow \lab'\flowsto  \level}
\and 
\inferrule[Bi-Both] 
  { \sevent{s_1} \Leftrightarrow \sevent{s_2} \\ \vdash \lab\flowsto \labtwo \\ \vdash \lab'\flowsto \labtwo'}
  {S\vdash \sevent{s_1}?\lab \leftrightarrow \lab'\flowsto  \sevent{s_2}? \labtwo \leftrightarrow \labtwo'}
\end{mathpar}
}
\caption{Selected Static Rules for ``Flows-to'' Relation, which $S\vDash \sevent{s_1}$ denotes that $\sevent{s_1}$ must have been $\true$ given $S$.}
\Description{}
\label{fig:inference_rules}
\end{figure}

We first define a \emph{flows-to} relation on dynamic labels, parameterized by a set of security events $S$, such as $\{\sevent{s_1},\neg\sevent{s_2}\}$. Intuitively, the flows-to relation $S \vdash \lab \flowsto \labtwo$ holds if label $\lab$ is less restrictive than, or equivalent to $\labtwo$ for any event trace $\traceout{\strace}$ that satisfies $S$ and its extensions.
A selection of rules for the flows-to relation is summarized in Figure~\ref{fig:inference_rules}; all rules are included in~\cite{full_proof}.

\mypara{Base Cases}
Rules~\ruleref{Lattice} and~\ruleref{Trans} are general rules based on Denning-style security lattice and the rule of transitivity. These rules are the same as a typical static information flow system.

\mypara{Directional Cases}
Rule~\ruleref{R-To} compares a static level $\level$ against a dynamic label $\sevent{s_2}?\labtwo\rightarrow \labtwo'$ under the assumption that $\neg \sevent{s_2}$ evaluates to $\true$ (i.e., $\sevent{s_2}$ evaluates to $\false$). Hence, we can simplify the dynamic label to $\labtwo'$, as the transition to $\labtwo'$ must have happened. Rule~\ruleref{R-From} is symmetric to Rule~\ruleref{R-To}.
Note that in both rules, we intentionally omit the event facts $S$ before $\vdash$ from the assumptions to ensure soundness. The reason is that, with nested conditions, events in $S$ may be triggered before the top-level condition is evaluated. Consequently, propagating them to the more deeply nested sub-label may be unsound.

The general case when both ends of flows-to relation involve dynamic labels is a bit more complicated. To illustrate the subtlety, we consider a naive but \emph{incorrect} extension to Rule~\ruleref{R-To}
\begin{mathpar}
\small
\inferrule[] 
  { S\vDash \neg \sevent{s_2} \\  \vdash \sevent{s_1}?\lab\rightarrow \lab' \flowsto \labtwo'}
  { S\vdash \sevent{s_1}?\lab\rightarrow \lab' \flowsto \sevent{s_2}?\labtwo\rightarrow \labtwo'}
\end{mathpar}
At first glance, this rule is legit as it follows the same argument as \ruleref{R-To}: it only reduces $\sevent{s_2}?\labtwo\rightarrow \labtwo'$ when $\sevent{s_2}$ turns $\false$. However, this rule is unsound when $\labtwo'$ has nested conditions. To see why, consider the following counterexample: 
\[\sevent{s_1}?\Low \rightarrow \High \flowsto \sevent{s_2}?\Low \rightarrow (\sevent{s_1}?\Low\rightarrow\High)\] 
and a trace $\strace$ where  $\sevent{s_1}$ is set to $\false$ and then to $\true$ in sequence. After that, $\sevent{s_2}$ turns $\false$, as illustrated below. 

\begin{center}
    \includegraphics[width=0.5\columnwidth]{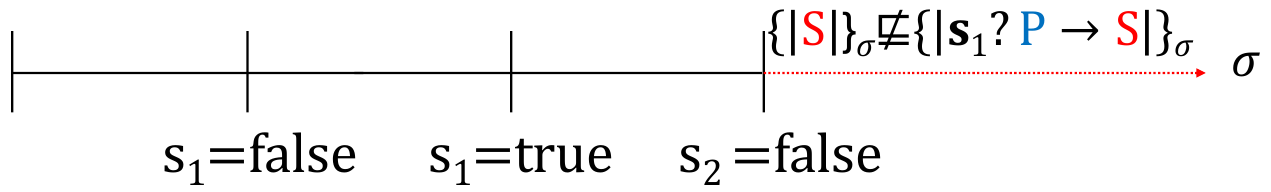}
\end{center}

At the end of $\strace$, $\neg \sevent{s_2}$ holds. Hence, we can instantiate $S$ with $\neg \sevent{s_2}$ in the proposed rule, and derive
\begin{mathpar}
\small
\inferrule[] 
  {  \neg \sevent{s_2} \vDash \neg \sevent{s_2} \\  \vdash \sevent{s_1}?\Low \rightarrow \High \flowsto \sevent{s_1}?\Low\rightarrow\High}
  { \neg \sevent{s_2} \vdash \sevent{s_1}?\Low \rightarrow \High \flowsto \sevent{s_2}? \Low \rightarrow (\sevent{s_1}?\Low\rightarrow\High)}
\end{mathpar}
where the assumption that $ \vdash \sevent{s_1}?\Low \rightarrow \High \flowsto \sevent{s_1}?\Low\rightarrow\High$ holds trivially as both sides are identical.

However, we next show that the derived flow-to relation is unsound. By the semantics, $\auxfuncold{\sevent{s_1}? \Low \rightarrow \High}{\strace}$ reduces to 
$\High$
as $\sevent{s_1}$ turned to $\false$ on the trace. Moreover, $\auxfuncold{\sevent{s_2}? \Low \rightarrow (\sevent{s_1}? \Low \rightarrow \High)}{\strace}$ reduces to $\auxfuncold{\sevent{s_1}? \Low \rightarrow \High}{\strace^{[i:]}}$ where $i$ is the first index when $\sevent{s_2}$ turned to $\false$. Note that even though $\sevent{s_1}$ \emph{was} $\false$ on the whole trace $\strace$, the nested condition $\sevent{s_1}$ is evaluated on the sub-trace \emph{after} the top-level condition $\sevent{s_2}$ is triggered (Figure~\ref{fig:label-semantics}). On that sub-trace ${\strace^{[i:]}}$, $\sevent{s_1}$ evaluates to $\true$. Hence, $\auxfuncold{\sevent{s_1}? \Low \rightarrow \High}{{\strace^{[i:]}}}$ further reduces to $\Low$. Hence, $\auxfuncold{\sevent{s_1}? \Low \rightarrow \High}{\strace}=\High \not\LEQ \Low = \auxfuncold{\sevent{s_2}? \Low \rightarrow (\sevent{s_1}? \Low \rightarrow \High)}{\strace}$, indicating that 
$\sevent{s_1}? \Low \rightarrow \High$ is not always bounded by $\sevent{s_2}? \Low \rightarrow (\sevent{s_1}? \Low \rightarrow \High)$ at $\strace$ (i.e., $\sevent{s_1}? \Low \rightarrow \High \not\LEQ_\strace \sevent{s_2}? \Low \rightarrow (\sevent{s_1}? \Low \rightarrow \High)$, Definition~\ref{def:partial_dynamic}). A similar issue occurs when $\lab$ contains nested conditions.

To ensure soundness while permitting permissiveness when possible, we observe that the soundness issue above can be avoided if $\lab$ has no nested conditions (i.e., it is a level $\level$), and $\sevent{s_1}$ is never $\false$ before $\sevent{s_2}$ turns $\false$, which can be captured as $(\neg\sevent{s_1} \Rightarrow \neg\sevent{s_2})$, where $\Rightarrow$ is a logical implication. The general case is included in the sound derivation rule~\ruleref{D-To} in Figure~\ref{fig:inference_rules}. The symmetric version Rule~\ruleref{D-From} is included in~\cite{full_proof}.

When the permissive rules~\ruleref{D-To} and~\ruleref{D-From} are not applicable, rule~\ruleref{R-To Two} conservatively checks that $\sevent{s_1}?\lab\rightarrow \lab'$ can flow to some level $\level$. For the general case of $\sevent{s_1}?\lab\rightarrow \lab'  \flowsto \sevent{s_2}?\labtwo\rightarrow \labtwo'$, similar to Rule~\ruleref{D-To}, Rule~\ruleref{D-To Two} assumes that $(\neg\sevent{s_2} \Rightarrow \neg\sevent{s_1})$ to be sound.

\mypara{Bi-Directional Cases}
Finally, we introduce the static rules for bi-directional policies, which align with the design of $\rightarrow$. Bi-directional policies only have the less permissive rules as the condition can change throughout trace extensions, and it is hard to define rules without binding to both ends. Rules~\ruleref{Bi-To} and~\ruleref{Bi-From} correspond to the cases where one side is a plain static level $\level$ in the single-direction cases, checking bounds for both parts of the dynamic policy. Rule~\ruleref{Bi-Both} assumes that the dynamic labels on both sides have equivalent conditions, and then checks the ``flows-to'' relation on the corresponding components.

Next, we establish the soundness of the static inference rules: for any labels $\lab$, $\labtwo$ and a set of security events $S$, $S\vdash \lab\flowsto \labtwo$ implies that $\labtwo$ is more restrictive or equal to $\lab$ on any event trace $\strace$ where $S$ holds, as well as on any extension of $\strace$.

\newtcbtheorem[use counter=theorem]{mytheorem}{Theorem}%
  {mytheorem}{theorem}

\begin{theorem}[Soundness of Flows-To Relation]
\label{theorem:setsound}
For all policies p and q, a trace $\strace$, a set of security events $S={\sevent{s_0}\dots\sevent{s_k}}$, if $S\vdash \lab\flowsto \labtwo$ and $S$ hold on $\strace$, then $\lab \LEQ_\strace \labtwo$.
\end{theorem}

\begin{proofsketch}
We proceed by induction on the static inference rules. We present the proof for Cases~\ruleref{R-To} and~\ruleref{D-To} here; the complete proof is available in~\cite{full_proof}. 

By Definition~\ref{def:partial_dynamic}, it is sufficient to show that for any $\strace'$ such that $\strace\preceq\strace'$, we have $\auxfuncold{\lab}{\strace'}\LEQ\auxfuncold{\labtwo}{\strace'}$. 

\begin{itemize}
    \item \ruleref{R-To}:
    By the assumption, we have $\vdash \level\flowsto \labtwo'$ and $S\vDash \neg\sevent{s_2}$. Since $S$ holds on $\strace$ by assumption, $\strace \preceq \strace'$, and $S\vDash \neg\sevent{s_2}$, there must be some index $i$ on $\strace'$ at which $\sevent{s_2}$ first turns $\false$. According to the semantics of Figure~\ref{fig:label-semantics}, we know that $\auxfuncold{\sevent{s_2}?\labtwo\rightarrow \labtwo'}{\strace'}=\auxfuncold{\labtwo'}{\strace'^{[i:]}}$.
    
    From the induction hypothesis and the assumption that $\vdash \level\flowsto \labtwo'$, we get $\auxfuncold{\level}{\strace'^{[i:]}}\LEQ \auxfuncold{\labtwo'}{\strace'^{[i:]}}$.
    Hence, $\auxfuncold{\level}{\strace'}= \level = \auxfuncold{\level}{\strace'^{[i:]}}\LEQ \auxfuncold{\labtwo'}{\strace'^{[i:]}}=\auxfuncold{\sevent{s_2}?\labtwo\rightarrow \labtwo'}{\strace'}$.
    
\item \ruleref{D-To}: Similar to \ruleref{R-To}, we first derive $\auxfuncold{\sevent{s_2}?\labtwo\rightarrow \labtwo'}{\strace'}=\auxfuncold{\labtwo'}{\strace'^{[i:]}}$ from the assumptions, where $i$ is the index when $\sevent{s_2}$ first turns $\false$.

If $\sevent{s_1}$ has been $\false$ on $\strace'$, let $j$ be the index when it turns false the first time. Since $\neg\sevent{s_1} \Rightarrow \neg\sevent{s_2}$, it is guaranteed that $j\geq i$. According to the semantics of Figure~\ref{fig:label-semantics}, we know that $\auxfuncold{\sevent{s_1}?\level\rightarrow \lab'}{\strace'}=\auxfuncold{\lab'}{\strace'^{[j:]}}=\auxfuncold{\sevent{s_1}?\level\rightarrow \lab'}{\strace'^{[i:]}}$ where the second equation holds as $\sevent{s_1}$ must remain true between $i$ and $j$. Moreover,
from the induction hypothesis and the assumption $\vdash \sevent{s_1}?\level\rightarrow\lab' \flowsto \labtwo'$, we know that 
$\auxfuncold{\sevent{s_1}?\level\rightarrow\lab'}{\strace'^{[i:]}} \LEQ \auxfuncold{\labtwo'}{\strace'^{[i:]}}$. Hence $\auxfuncold{\sevent{s_1}?\level \rightarrow \lab'}{\strace'}=\auxfuncold{\sevent{s_1}?\level\rightarrow\lab'}{\strace'^{[i:]}}\LEQ \auxfuncold{\labtwo'}{\strace'^{[i:]}}=\auxfuncold{\sevent{s_2}?\labtwo\rightarrow \labtwo'}{\strace'}$. 

If $\sevent{s_1}$ has never been $\false$ on $\strace'$, we know that $\auxfuncold{\sevent{s_1}?\level \rightarrow \lab'}{\strace'}=\level=\auxfuncold{\sevent{s_1}?\level \rightarrow \lab'}{\strace'^{[i:]}}$. 
From the induction hypothesis and the assumption $\vdash \sevent{s_1}?\level\rightarrow\lab' \flowsto \labtwo'$, we know that 
$\auxfuncold{\sevent{s_1}?\level\rightarrow\lab'}{\strace'^{[i:]}} \LEQ \auxfuncold{\labtwo'}{\strace'^{[i:]}}$. Hence $\auxfuncold{\sevent{s_1}?\level \rightarrow \lab'}{\strace'}=\auxfuncold{\sevent{s_1}?\level\rightarrow\lab'}{\strace'^{[i:]}}\LEQ \auxfuncold{\labtwo'}{\strace'^{[i:]}}=\auxfuncold{\sevent{s_2}?\labtwo\rightarrow \labtwo'}{\strace'}$. 


\end{itemize}
\end{proofsketch}

\begin{figure}
{\small
\begin{mathpar}
\mprset{flushleft} 
\inferrule[Lattice] 
  { \level\LEQ \level' }
  {S\vdash \level \releaseto\level' }
\quad
\inferrule[Trans] 
  { S\vdash \level \releaseto\level' \quad S\vdash \level' \releaseto\level'' }
  {S\vdash \level \releaseto\level'' }
\quad
\mprset{flushleft}
\inferrule[R-From-P] 
  { S\vDash \neg \sevent{s_1} \quad  \vdash \lab'\releaseto \level}
  { S\vdash \sevent{s_1}?\lab\rightarrow\lab'\releaseto \level}
\quad
\inferrule[R-From-F] 
  { S\vDash \absent{\neg \sevent{s_1}} \quad  S \vdash \lab\releaseto \level}
  { S\vdash \sevent{s_1}?\lab\rightarrow\lab'\releaseto \level} 
\quad
\inferrule[Bi-From] 
  {\vdash \lab\releaseto \level \quad \vdash \lab'\releaseto \level}
  {S\vdash \sevent{s_1}?\lab \leftrightarrow \lab'\releaseto  \level}
\end{mathpar}
}
\caption{Static Rules for ``Release-to'' Relation.}
\Description{}
\label{fig:inference_rules_no_extend}
\end{figure}

\subsection{Static Rules for Release-to Relation}
\label{sec:release_to}
Dynamic release policy controls information release at output commands. More formally, for each output command $\outcmd{\level}{\expr}{\sevent{s_0}\dots\sevent{s_k}}$, information with label $\lab$ can be released when $\auxfuncold{\lab}{\strace} \LEQ \level$ (Figure~\ref{fig:label-semantics}). Accordingly, we define a release-to relation, written as $S \vdash \lab \releaseto \level$, to statically check if we can safely release information with label $\lab$ to output command at level $\level$ under the assumption that $S$ holds when the output occurs. Compared with the flows-to relation $\flowsto$, the release-to relation is simpler as (1) the target (i.e. the right-hand side) is always a security level by design, and (2) we do not need to consider any changes to $\lab$ after the release.

Both rules~\ruleref{R-From-P} and~\ruleref{Bi-From} are similar to their counterparts in the flows-to relation. The key difference is in rule~\ruleref{R-From-F}: it reduces the dynamic label to $\lab$ whenever $S\vDash \absent{\neg \sevent{s_1}}$. This is sound for the release-to relation as by definition (Figure~\ref{fig:label-semantics}), how the interpretation of where information can be released to is only determined by the current event trace. But for the flows-to relation, we have to consider all extensions to the current trace, which may violate the assumption.

\begin{theorem}[Soundness of Release-To Relation]
\label{theorem:setsound_no_extend}
For all policies p, a trace $\strace$, a set of security events $S={\sevent{s_0\dots s_k}}$, if $S\vdash \lab\releaseto\level$ and $S$ holds on $\strace$, then $\auxfuncold{p}{\strace}\LEQ\level$.
\end{theorem}

The proof is similar to the proof of Theorem~\ref{theorem:setsound}. The full proof is included in~\cite{full_proof}.

\subsection{Typing Rules}
\label{sec:typesystem}
\begin{figure*}
{\small
\begin{mathpar}
\mprset{flushleft} 

\inferrule[Skip] 
  {  }
  { pc, \Gamma \vdash \Skip}
\and
\inferrule[Sequence]
  {  pc, \Gamma\vdash c_1 \quad pc, \Gamma \vdash c_2 }
  {pc, \Gamma\vdash c_1; c_2}
\and
\inferrule[If]
  {  \Gamma\vdash e: \pc' \quad  \pc', \Gamma\vdash c_1 \quad \pc',\Gamma \vdash c_2 \quad \vdash \pc\flowsto \pc'}
  { \pc,\Gamma\vdash \ifcmd{e}{c_1}{c_2}} 
  \and
 \inferrule[Assign]
  {\vdash pc \flowsto \Gamma(x)  \\ \Gamma \vdash e:\Gamma(x)}
  {pc, \Gamma\vdash x:=e} 
\and
\inferrule[While]
  {  \Gamma\vdash e: \pc' \quad  \pc', \Gamma\vdash c\quad \vdash \pc\flowsto \pc'}
  { \pc,\Gamma\vdash \while{e}{c}} 
\and
\inferrule[Output]
  { \Gamma\vdash e:\lab \\ \sevent{s_0...s_k}\vdash \lab \releaseto \level \\ \vdash \pc\flowsto \level }
  { \pc, \Gamma \vdash \outcmd{\level}{\expr}{\sevent{s_0...s_k}}  }
\and
\inferrule[Output-Per]
  {x\text{ immutable } \\ \Gamma(x) \text{ is persistent} \\  \vdash\pc\flowsto\level}
  { \pc, \Gamma \vdash \outcmd{\level}{x}{ \sevent{o}_\x \sevent{@\level, s_0...s_k}}  }
\and
\inferrule[Relabel]
  { \Gamma\vdash e:\lab_f\\\vdash \pc \flowsto \Gamma(x)\\ \sevent{s}_0...\sevent{s}_k\vdash \level_t\flowsto \Gamma(x)\\\\ \sevent{s}_0...\sevent{s}_k\vdash \lab_f\flowsto \level_t}
  {\pc, \Gamma\vdash x:=\relabel{e}{\lab_f}{\level_t}{\sevent{s_0...s_k}}}
\and
\inferrule[EventOn($\sevent{s}$)]
  { \vdash\pc\flowsto \bot }
  { \pc, \Gamma \vdash \cod{EventOn}(\sevent{s})  }
\and
\inferrule[EventOff($\sevent{s}$)]
  { \vdash\pc\flowsto \bot }
  { \pc, \Gamma \vdash \cod{EventOff}(\sevent{s}) }
\end{mathpar}
}
\caption{Typing Rules for Commands.}
\Description{}

\label{fig:type_enforcement}
\end{figure*}

Lastly, we present a type system that enforces the end-to-end dynamic release policy on the whole program.
The overall structure of the typing rules is similar to other information flow type systems~\cite{Sabeleld2003, Chong2008,Chong2005}, where each command is checked under two conditions $\pc$ and $\Gamma$. Here, $\pc$ is the standard program counter label used to control implicit information flows~\cite{Sabeleld2003}, and $\Gamma$ is a mapping that assigns each variable to its security label. The typing rules on expressions, in the form of $\Gamma \vdash \expr : \lab$, are completely standard, which is included in~\cite{full_proof}.

Several typing rules for commands are standard. Rule~\ruleref{$\Skip$} denotes that a nop command is always secure. Rule~\ruleref{Sequence} enforces that both commands in a sequence are typed under the same $\pc$ label. Rule~\ruleref{If} guards both branches of a conditional with a program counter $\pc'$, which must be at least as restrictive as the condition's label, in order to control implicit flows. Rule~\ruleref{While} is similar to Rule~\ruleref{If}. Rule~\ruleref{Assign} requires that the policy on $\x$ is at least as restrictive as that on $\expr$ to control explicit information flow from $\expr$ to $x$.

Other rules are novel and carefully designed to check the dynamic release policy. Rule~\ruleref{Relabel} is the first rule that utilizes annotated events $\sevent{s_0\dots s}_k$ to check information flow from $\expr$ to $x$ in a more permissive way than Rule~\ruleref{Assign}. It allows information to be ``relabeled'' when possible.

\mypara{Example 1: Bidding Game (Figure~\ref{fig:dynamicpolapp}.i)} 

\noindent We examine the relabel command on Line 9, which occurs after $\sevent{release}$ is set to $\true$. This line type-checks under rule~\ruleref{Relabel} as follows (we omit less interesting parts of the proof for simplicity):
\begin{mathpar}
\small
\inferrule[]
  {
    \Gamma \vdash \cod{bid} : \neg \sevent{release}? \High \rightarrow_t \Low\\
    \inferrule
      {
        \sevent{release} \vDash \sevent{release} \\
        \vdash \Low \flowsto \Low
      }
      {
        \sevent{release} \vdash \neg\sevent{release}? \High \rightarrow_t \Low \flowsto \Low
      }
  }
  {
    \Low, \Gamma \vdash \cod{wbid} := \relabel{\cod{bid}}{\neg \sevent{release}? \High \rightarrow_t \Low}{\Low}{\sevent{release}}
  }
\end{mathpar}

Recall that at runtime, \cod{relabel} checks if all events after $\cod{using}$ hold; it turns the command into a nop if any assumed event is not correct (Figure~\ref{fig:key_semantics}). Hence, the typing rule can safely use event $\sevent{release}$ to check if a flow from dynamic label $\neg \sevent{release}? \High \rightarrow_t \Low$ to $\Low$ is secure, via the permissive Rule~\ruleref{R-From} in Figure~\ref{fig:inference_rules}.

Rule~\ruleref{Output} is the second typing rule that utilizes annotated events $\sevent{s_0\dots s_k}$ to check if the value of $\expr$ can be released in a permissive way. 

\mypara{Example 2: Credit Card (Figure~\ref{fig:dynamicpolapp}.ii)} 

\noindent Both variables $\cod{copy}$ and $\cod{credit\_card}$ are labeled with dynamic labels $\sevent{trans}?\cod{M}\rightarrow_t\top$. The output in Line 6 can be type checked by the rule~\ruleref{Output}, which allows the value of $\cod{copy}$ to be released as $\absent{\neg\sevent{trans}}$ at the point:
\begin{mathpar}
\footnotesize
\inferrule[]
{
  \Gamma \vdash \cod{copy} : \neg\sevent{trans}? \cod{M} \rightarrow_t \top \\
  \inferrule[]
  {
    \absent{\sevent{trans}} \vDash \absent{\sevent{trans}} \\
    \vdash \cod{M} \releaseto \cod{M}
  }
  {
    \absent{\sevent{trans}} \vdash \neg\sevent{trans}? \cod{M} \rightarrow_t \top \releaseto \cod{M}
  }
}
{
  \cod{M}, \Gamma \vdash \outcmd{\cod{M}}{\cod{copy}}{\absent{\sevent{trans}}}
}

\end{mathpar}
Recall that at runtime, the output command checks if all events after $\cod{using}$ are correct; it turns the command into a nop if any assumed event is not correct (Figure~\ref{fig:key_semantics}). Hence, the typing rule can use the fact that $\sevent{trans}$ has never been $\true$ (i.e., $\absent{\sevent{trans}}$) to check if dynamic label $\neg\sevent{ trans}? \cod{M} \rightarrow_t \top$ can be released to $\cod{M}$, via Rule~\ruleref{R-From-F} in Figure~\ref{fig:inference_rules_no_extend}. 

On the other hand, the same output command at Line 9, after the transaction is completed, is insecure. Although the command can still be type checked, we note that since the transaction is completed, the assumed event $\absent{\sevent{trans}}$ is no longer correct. Hence, nothing will be released at runtime (Rule~\ruleref{OS-Output-Fail} in Figure~\ref{fig:key_semantics}).

Rule~\ruleref{Output-Per} is a more permissive version of Rule~\ruleref{Output}, in the sense that it allows the value of $x$ to be released \emph{again} when $x$ is immutable and $\Gamma(x)$ has a persistent policy. Here, the output event $\sevent{o}_\x$ assumes that $x$ has already been released before the current output command.

\mypara{Example 3: Library System (Figure~\ref{fig:dynamicpolapp}.iii)}

\noindent Before Alice returns the book, the variable \cod{notes} was released in Line 5, which can be type-checked in the same way as in Line 6 of the credit card example described above. According to the semantics of the $\cod{output}$ command, the output event $\sevent{o_{notes}}$ is added to the execution trace at that point.
After Alice returns the book on Line 6, the same value that was released on Line 5 can still be safely released again on Line 10 as the policy on \cod{notes} and \cod{book} is a persistent policy. This can be confirmed by the following proof tree:
\begin{mathpar}
\footnotesize
\inferrule[]
{
  \inferrule[]{\cod{notes} \text{ immutable } \\ \Gamma(\cod{notes}) \text{ is persistent}}{}
}
{
  \Low, \Gamma \vdash \outcmd{\Low}{\cod{notes}}{\sevent{o_{notes}@\Low}}
}

\end{mathpar}
On the other hand, if we replace the same line with $\outcmd{\Low}{\cod{book}}{\sevent{o_{book}}}$, the command is insecure as $\cod{book}$ was never released before it was returned. In this case, the insecure command is still type-checked. But since the assumed event $\sevent{o_{book}}$ is no longer correct, nothing will be released at runtime (Rule~\ruleref{OS-Output-Fail} in Figure~\ref{fig:key_semantics}).

Finally, Rules~\ruleref{eventon($\sevent{s}$)} and~\ruleref{eventoff($\sevent{s}$)} handle distinguished commands that set/unset security events. For both of them, we require that the $\pc$ label be public (i.e., at the lowest level $\Low$ in the lattice). The reason is to avoid label channels, where secret information interferes with the sensitivity level of variables~\cite{Chong2008}.

\section{Soundness Proof}
\label{sec:endtoendproof}
\begin{figure}
    \centering
    \includegraphics[width=0.5\columnwidth]{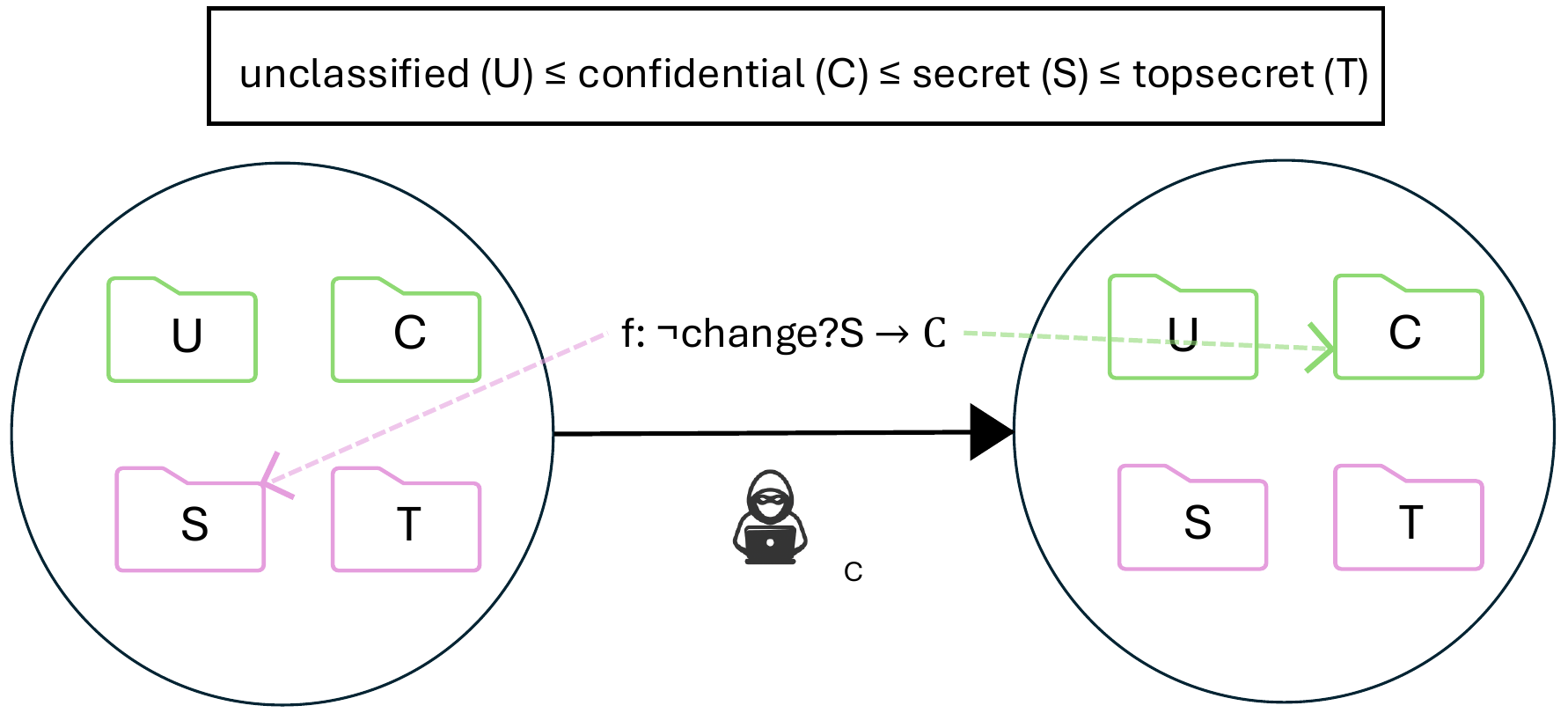}
    \caption{An example in a multi-level military system, where the variable $f$ belongs to levels before and after the security event $\sevent{change}$ is set to $\true$.}
    \Description{}
    \vspace{-2ex}
    \label{fig:theorem3_set}
\end{figure}
The hallmark of information flow security is its ability to maintain end-to-end confidentiality and integrity across all system executions. For noninterference with a static policy, various proof techniques (e.g.,~\cite{li2017arxiv,flowcaml,volpano1996}) have been developed to prove that static type systems can soundly enforce noninterference policy. 

In those standard noninterference proofs, all variables' sensitivity/integrity are assumed to be fixed throughout program execution. Hence, for any attacker at a certain security level $\fixedlevel$, we can statically partition all variables into two sets: public set $\{ \x \mid \Gamma(x)\LEQ \fixedlevel\}$ and secret set $\{ \x \mid \Gamma(x)\not\LEQ \fixedlevel\}$. At a high level, standard noninterference proofs are built on an \emph{indistinguishability relation} on memory states: two memory states $\Mem_1$ and $\Mem_2$ are said to be indistinguishable for an attacker at $\fixedlevel$ when they agree on the public set of variables:
\[\Mem_1 \approx \Mem_2 \triangleq \forall x.~\Gamma(x)\LEQ \fixedlevel \Rightarrow \Mem_1(\x)=\Mem_2(\x) \]
Then, a standard proof proceeds by showing that the indistinguishability relation is an invariance on any two program executions with indistinguishable initial memory states.

To see the invariance more clearly, we consider a military-grade security system that includes four security levels: unclassified (U), confidential (C), secret (S), and top secret (T) in Figure~\ref{fig:theorem3_set}. Here, an attacker with a confidential security clearance can access only unclassified and confidential files. Hence, two memory states are considered indistinguishable if they agree on $\{x \mid \Gamma(x)=\text{U} \lor \Gamma(x)=\text{C}\}$. Note that for a static policy, the set remains unchanged throughout program execution.

However, in the presence of dynamic labels, such a proof strategy no longer applies, as the set of secret variables can change throughout program execution. Suppose that there is a variable $f$ associated with dynamic release label $\neg\sevent{ change}?\text{S} \rightarrow_t \text{C}$.  
A naive extension of the indistinguishability relation to the dynamic setting is as follows:
\[\Mem_1 \approx_\strace \Mem_2 \triangleq \forall \x.~\auxfuncold{\Gamma(\x)}{\strace}\LEQ \text{C} \Rightarrow \Mem_1(\x)=\Mem_2(\x) \]
However, maintaining such an indistinguishability relation is impossible. Let $\strace$ and $\strace'$ be traces right before and after $\sevent{change}$ is set to $\true$ in Figure~\ref{fig:theorem3_set}. Given $\Mem_1 \approx_{\strace} \Mem_2$, $\Mem_1(f)$ and $\Mem_2(f)$ might hold different values since $\auxfuncold{\Gamma(f)}{\strace}\not\LEQ \text{C}$. Hence, we are unable to prove $\Mem_1 \approx_{\strace'} \Mem_2$ which requires that $\Mem_1(f)=\Mem_2(f)$ as $\auxfuncold{\Gamma(f)}{\strace'}\LEQ \text{C}$.

To address the challenge, in addition to the attacker level $\fixedlevel$, we also introduce a distinguished, arbitrary dynamic release label $\fixedlab$ and prove that the desired property of dynamic release (Definition~\ref{def:newpolicy}) holds against any attacker's level $\fixedlevel$ and source's label $\fixedlab \in  \mathbb{B}$. Intuitively, $\fixedlab$ is a distinguished \emph{dynamic} label denoting the source label that is considered ``sensitive'' per security reasoning. With a fixed attacker level $\fixedlevel$ and a fixed source label $\fixedlab$, indistinguishability in memory is defined as
\begin{definition}
\label{def:indistinguishibility}
Two memories $\Mem_1$ and $\Mem_2$ are indistinguishable under an event trace $\strace$ if they agree on all variables such that $\fixedlab \not\LEQ_\strace \Gamma(\x)$:
\[
  {\Mem_1 \approx_\strace \Mem_2} \iff {\forall \x.~\fixedlab \not\LEQ_\strace \Gamma(\x) \Rightarrow \Mem_1(\x)= \Mem_2(\x) }
\]
\end{definition}
Recall that by Definition~\ref{def:partial_dynamic}, $\lab \sqsubseteq_\sigma \lab'$ denotes that $\auxfuncold{\lab}{\strace'} \LEQ  \auxfuncold{\lab'}{\strace'}$ for any extension $\strace'$ of $\strace$.

A key insight here is that we can maintain the invariance since the size of $\{\x \mid \fixedlab \not\LEQ_\strace \Gamma(\x) \}$ can never increase through program execution due to the following fact.
\newtcbtheorem[use counter=lemma]{mylemma}{Lemma}%
  {mylemmaobox}{lemma}

\begin{lemma}[Trace Preservation]
$$ \forall \lab,\lab',\strace.~\lab \LEQ_\strace\lab' \Rightarrow (\forall \strace'.~\strace\preceq\strace'\Rightarrow  \lab\LEQ_{\strace'}\lab')$$
\label{lemma:trace_preservation}
\vspace{-2ex}
\end{lemma}

The second challenge is that given an attacker's level $\fixedlevel$ and the source label $\fixedlab$, information leakage is allowed under some circumstances. Recall that under Definition~\ref{def:Know_gain_last}, the consistency relation on traces ignores all outputs produced when $\auxfuncold{\fixedlab}{\strace}\LEQ \fixedlevel$ (i.e., when information with $\fixedlab$ can be released to the attacker at $\fixedlevel$). Moreover, when $\fixedlab$ is a persistent policy, dynamic release allows information leaked in the past to be revealed again (Definition~\ref{def:newpolicy}). Accordingly, in a few lemmas and theorems in the soundness proof, we added extra conditions under which intentional leakage is allowed. For persistent policy, note that the type system conservatively allows leakage when the value of $\x$ has been revealed previously (by checking $\sevent{o}_\x$ under \cod{using}) and that $\x$ is immutable (Rule~\ruleref{Output-Per} in Figure~\ref{fig:type_enforcement}). Hence, the relaxing condition is that the command $c$ being executed can be type-checked by Rule~\ruleref{Output-Per}.

The formal soundness proof follows a previous small-step noninterference proof~\cite{li2017arxiv}, with notable changes as stated above. In the concrete proof, we introduce an augmented semantics where the ``tainted'' values are wrapped in brackets, written as $[n]$ instead of $n$. Moreover, instructions executed under “tainted” branch conditions (i.e., branch conditions that evaluate to bracketed values) are also wrapped in brackets, written as $[c]$ instead of $c$.
A memory is said to be ``well-formed'' if all bracketed values are only stored in variables that are considered sensitive to an attacker at level $\fixedlevel$.

The full soundness proof, including definition of the augmented language, bracketed values and extended semantics, is provided in~\cite{full_proof}. Next, we highlight important results and a proof sketch. 

We consider any attacker's level $\fixedlevel$ and source label $\fixedlab$ in the proof.
The first important lemma states that indistinguishability in memory (Definition~\ref{def:indistinguishibility}) is preserved in each step, as long as starting memories $\Mem_1$ and $\Mem_2$ are indistinguishable and each is well-formed ($\wellformmem{\strace}{\Mem_1}$ and $\wellformmem{\strace}{\Mem_2}$), $\command_1$ and $\command_2$ both type-check and are equivalent ($\command_1\approx \command_2$), and the relaxing conditions are \emph{not} met (i.e., $\auxfuncold{B}{\strace}\not\LEQ\fixedlevel$ and both $\command_1$ and $\command_2$ type-check without using the relax rule~\ruleref{Output-Per}). Note that the Lemma also allows $\command_2$ to diverge, making the dynamic release assurance termination-insensitive~\cite{askarov2007, askarov2008}.

\begin{lemma} [Unwinding]
\label{lemma:unwinding_dynamicrelease}
\small
\vspace{-1ex}
\begin{multline*}
\forall \Mem_1, \Mem_2, \command_1, \command_2, \tau_1, \tau_2, \strace, x, \pc, \Gamma.~\wellformmem{\strace}{\Mem_1} \land \wellformmem{\strace}{\Mem_2} \land \Mem_1\approx_\strace \Mem_2 \land \\
 (\command_1\neq \outcmd{\level'}{\x}{\sevent{o}_\x, \sevent{s_0\dots s_k}})  \land
 \auxfuncold{\fixedlab}{\strace}\not\LEQ \fixedlevel \land
 \command_1\approx \command_2 \land 
\configThree{\Mem_1}{\command_1}{\strace}\xrightarrow{\tau_1}\configThree{\Mem_1'}{\command_1'}{\strace'}\Rightarrow \\
(\exists \Mem_2',\command_2',\tau_2. \configThree{\Mem_2} {\command_2}{\strace} \xrightarrow{\tau_2}^{*}\configThree{\Mem_2'}{\command_2'}{\strace'}\land \Mem_1'\approx_{\strace'} \Mem_2'\land \command_1'\approx \command_2'\land \proj{\trace_1}_\fixedlevel=\proj{\trace_2}_\fixedlevel) 
\lor(\exists \command.~\command_2=[\command] \text{ and } \command \text{ diverges})
\end{multline*}
\end{lemma}

\begin{proofsketch}
We sketch the proof of the most interesting case with Rule~\ruleref{Output} here.
\begin{itemize}
    \item \ruleref{Output}:
    Based on the assumption that  $\command_1\neq \outcmd{\level'}{\x}{\sevent{o}_\x, \sevent{s_0\dots s_k}}$, we know that only the typing rule~\ruleref{Output} is applicable in this case. 
  As $\command_1\approx \command_2$, it is trivial that $\Skip\approx\Skip$ and $\Mem_1'=\Mem_1\approx_\strace \Mem_2=\Mem_2'$. To prove $\proj{\trace_1}_\fixedlevel=\proj{\trace_2}_\fixedlevel$, there are two cases.
  \begin{itemize}
      \item  $\level'\LEQ\fixedlevel$: We show that $\configs{\Mem,e}\Downarrow v$ for some $n$ without brackets in this case. We proceed by assuming that $\configs{\Mem,\expr}\Downarrow[n]$ and show a contradiction.
      
      From the typing rule, we know that $\Gamma\vdash \expr:\lab$ and $\sevent{s_0}\dots\sevent{s_k}\vdash\lab\releaseto\level'$. From the assumption $\configs{\Mem,\expr}\Downarrow[n]$ and Lemma 4 
      in~\cite{full_proof}, we have $\fixedlab\LEQ_\strace \lab$, which implies $\auxfuncold{\fixedlab}{\strace}\LEQ\auxfuncold{\lab}{\strace}$ by definition. By Theorem~\ref{theorem:setsound_no_extend} and $\sevent{s_0}\dots\sevent{s_k}\vdash\lab\releaseto\level'$, we know $\auxfuncold{\lab}{\strace}\LEQ\level'$ whenever the output semantics is not a nop. Hence, we can derive $\auxfuncold{\fixedlab}{\strace}\LEQ L'\LEQ \fixedlevel$, which contradicts the assumption that $\auxfuncold{\fixedlab}{\strace}\not\LEQ \fixedlevel$ in the theorem statement.

    Hence, it must be true that $\configs{\Mem_1,\expr}\Downarrow n$ for some $n$. By Lemma 6
    in~\cite{full_proof}, $\configs{\Mem_2,\expr}\Downarrow n$ too. Hence, both output traces under $\Mem_1$ and $\Mem_2$ produce the same output event $\configs{\level',n,\strace}$. 
    
      \item  $\level'\not\LEQ\fixedlevel$: The projection traces on $\fixedlevel$ are empty ($\proj{\traceout{\tau_1}}_{\fixedlevel}=\emptyset=\proj{\traceout{\tau_2}}_{\fixedlevel}$) as $\level'\not\LEQ\fixedlevel$.
  \end{itemize}
\end{itemize}
\end{proofsketch}

\begin{figure}
    \centering
    \begin{subfigure}[b]{0.6\textwidth}
        \centering
        \includegraphics[width=\textwidth]{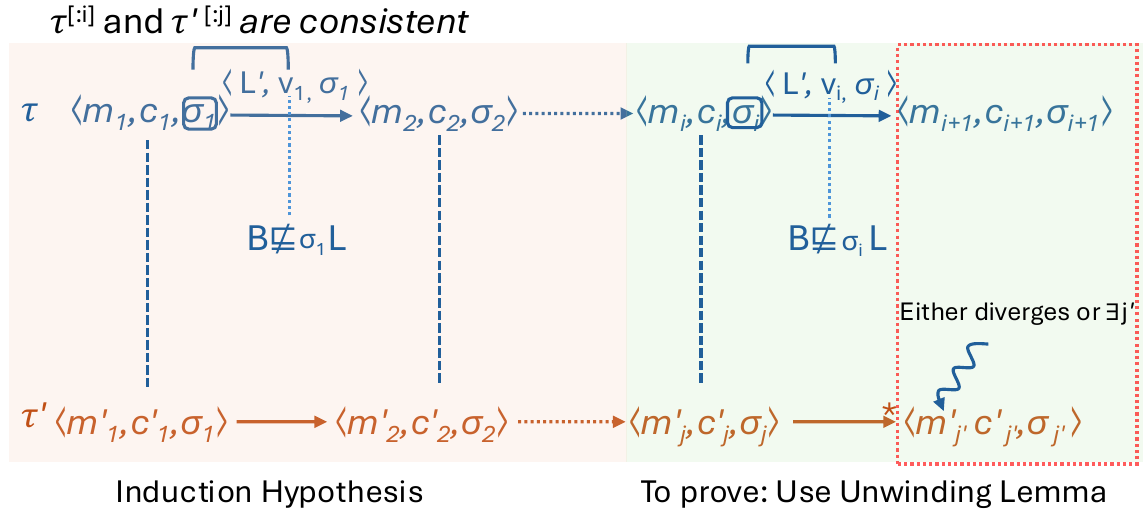}
        \caption{The Figure illustrates the proof concept on transient policy. 
        The orange part is the induction hypothesis and the green part is the part 
        to be proved with the Unwinding Lemma.}
        \label{fig:theorem3_figure_main}
    \end{subfigure}
    \hfill
    \begin{subfigure}[b]{0.6\textwidth}
        \centering
        \includegraphics[width=\textwidth]{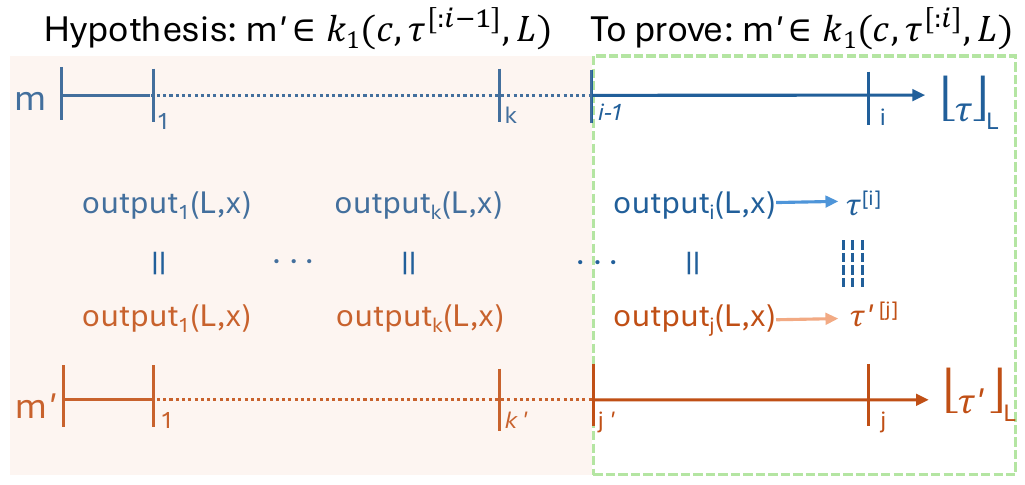}
        \caption{This Figure shows that $\Mem'\in k_1(c,\trace^{[:i],\level})$, 
        which indicates the variable $\x$ has been released. 
        This difference causes the proof not to be concluded with Unwinding Lemma. 
        We only show part of the output syntax due to space.}
        \label{fig:theorem3_figure_per_main}
    \end{subfigure}
    \caption{Illustrations on the proof of Theorem~\ref{theorem:typesound}. 
    Conceptual proof on (a) transient policy, 
    (b) persistent policy.}
    \label{fig:theorem3_combined}
\end{figure}

Based on the unwinding lemma, we next show the main soundness result.
\begin{theorem}[End-to-end Soundness]
\label{theorem:typesound}
For all command $\command$ and typing environment $\Gamma$, if $\pc, \Gamma \vdash \command$ for some program counter $\pc$, then the program $c$ satisfies dynamic release policy, which means
{\footnotesize
\begin{multline*}
\forall \Mem, \strace, \trace.~\configThree{\Mem}{\command}{\strace} \termout \trace  
\Rightarrow 
\forall 1\leq i\leq \len{\trace }.\quad
k_2(c,\trace^{[:i]}, \fixedlevel, \fixedlab) 
\supseteq 
\begin{cases} 
\closure{\Mem}_{\neq \fixedlab}, & \text{transient} \\
\closure{\Mem}_{\neq \fixedlab} \cap k_1(\command, \trace^{[:{i-1}]},\fixedlevel), & 
\text{persistent}
\end{cases}
\end{multline*}
}
\end{theorem}

\begin{proofsketch}
We conduct a proof of induction on the trace length $i$. Note that unlike a traditional proof for static noninterference, we need to show that under two indistinguishable initial memories, they both produce \emph{consistent}, rather than indistinguishable traces. 

\mypara{Transient} For a transient policy $\fixedlab$, proving $k_2(\command,\trace^{[:i]}, \fixedlevel, \fixedlab) \supseteq \closure{\Mem}_{\neq \fixedlab}$
boils down, after expanding the definitions of $k_1$ and $k_2$, to show that for any $\Mem'$ indistinguishable from $\Mem$ and $\configThree{\Mem'}{c}{\strace} \termout \trace'$, we can find some index $j$ on $\tau'$ such that $\trace'^{[:j]}$ and $\trace^{[:i]}$ are consistent (i.e., $\trace'^{[:j]} \equiv_{\fixedlab, \fixedlevel}  \trace^{[:i]}$). The induction step is illustrated in Figure~\ref{fig:theorem3_figure_main}, where the green part illustrates the new index $j'$ that we need to find in order to match index $i+1$ on trace $\tau$. The orange part illustrates the induction hypothesis: traces $\trace^{[:i]}$ and $\trace'^{[:j]}$ are consistent. 

There are two cases to consider at each induction step: $\auxfuncold{\fixedlab}{\strace_i}\not\LEQ \fixedlevel$ and $\auxfuncold{\fixedlab}{\strace_i}\LEQ \fixedlevel$. 
\begin{itemize}
    \item When $\auxfuncold{\fixedlab}{\strace_i}\not\LEQ \fixedlevel$, we can directly apply Lemma~\ref{lemma:unwinding_dynamicrelease} to show that $\proj{\trace'^{[j:j']}}_{B,\level}=\proj{\trace^{[i:i+1]}}_{B,\level}$. In addition to the induction hypothesis, we know that $\trace'^{[:j']} \equiv_{B,\level} \trace^{[:i+1]}$ stands.

    \item When $\auxfuncold{\fixedlab}{\strace_i}\LEQ \fixedlevel$, the event $v$ produced by $\configThree{\Mem_i}{\command_i}{\strace_i}$ (in the green part of Figure~\ref{fig:theorem3_figure_main}) must have an empty projection:
        \[\proj{\configThree{
		\Mem_i}{\command_i}{\strace_i}\xrightarrow{\configThree{\level'}{v}{\strace_i}}\configThree{
		\Mem_{i+1}}{\command_{i+1}}{\strace_{i+1}}}_{\fixedlab, \fixedlevel}~\triangleq 
        \emptyset\]
Hence, we simply set $j'=j$ and have $\trace'^{[:j']} \equiv_{B,\level} \trace^{[:i+1]}$ as $\trace'^{[:j']} =\trace'^{[:j]}$, $\trace^{[:i]}=\trace'^{[:i+1]}$ and, from the induction hypothesis, $\trace'^{[:j]} \equiv_{B,\level} \trace^{[:i]}$.
\end{itemize}

\mypara{Persistent} For persistent policy, the proof is mostly identical to the transient case, since, by definition, its security condition is a relaxation of its corresponding transient policy.
The exception is for 
the case of $$(\command_i=\outcmd{\level'}{\x}{\sevent{o}_\x,\sevent{s_0\dots s_k}})\land (\x \text{ is immutable})$$
because the unwinding lemma (Lemma~\ref{lemma:unwinding_dynamicrelease}) is not applicable. We need a new proof strategy for this case, which is illustrated by the green dotted box in Figure~\ref{fig:theorem3_figure_per_main}. 
To proceed, we utilize the fact that $\Mem'\in k_1(\command,\tau^{[:i-1]}, \fixedlevel)$ in the definition of persistent policy. By definition of $k_1$, we further know that $\proj{\trace^{[:i-1]}}_{\fixedlevel} \preceq  
\proj{\trace'^{[:j']}}_{\fixedlevel}$ where $j'$ is the index that corresponds to $i-1$ from the induction hypothesis, as illustrated in the orange part in Figure~\ref{fig:theorem3_figure_per_main}. Moreover, from the induction hypothesis that $\trace'^{[:j']} \equiv_{\fixedlab,\fixedlevel}  \trace^{[:i-1]}$, we know that all past outputs of $\x$ must produce the same value, say $n$. Note that since $\x$ is immutable, both $\trace^{[i+1]}$ and $\trace'^{[j'+1]}$ will output the same value $n$. Hence, we let $j=j'+1$ and have $\trace'^{[:j]} \equiv_{\fixedlab,\fixedlevel}  \trace^{[:i+1]}$. 
\end{proofsketch}

\section{Implementation}
\label{sec:case_studies}

Prior work such as Jif~\cite{myers2001} provides fine-grained information flow control for a Java-like language, but its implementation requires modifications to the programming language and relies on nonstandard compilation tools. To make our language and type system more accessible to end users, we follow the approach of recent systems such as Cocoon~\cite{lamba2024}, Carapace~\cite{carapace}, and Filament~\cite{filament}, which are implemented as Rust libraries without requiring language and compiler modifications. 

In particular, our prototype implementation is built on top of Filament~\cite{filament}, which provides Jif-like fine-grained information flow control for Rust. Consequently, our implementation is fully compatible with the standard Rust language and compiler.

\subsection{Dynamic Release on Rust} 

The recent work Filament~\cite{filament} presents a static information flow library in Rust that enforces noninterference with no runtime checks. To attach a value with a label, a Filament program wraps the value in a generic struct $\text{Label}\langle T, \level\rangle$, where $T$ is the data type and $\level$ is a (static) security level drawn from a predefined lattice. Filament uses traits to define the ordering of security levels (i.e., the lattice structure) based on these labels. However, Filament only supports static security levels.

To support dynamic release labels, we expand the generic struct to support dynamic labels and leverage Rust's phantom types for compile-time checking. More specifically, we define our label and value wrapper as a generic struct $\text{DRLabel}\langle T, S, LFrom, LTo \rangle$, representing data of type $T$ associated with a security label $S?LFrom\rightarrow LTo$ using the syntax in Figure~\ref{fig:while-syntax}, where $LFrom$ and $LTo$ are security levels\footnote{\revise{Our current implementation does not support first-order functions; we leave this feature for future work (Section~\ref{sec:language_features}).}}. Moreover, we maintain both a Boolean variable $\sevent{s}$ and a unique type $S$ to represent a security event. The difference is that $\sevent{s}$ can be set by library functions $\cod{eventon}$ and $\cod{eventoff}$, while $S$ is a static type indicating $\sevent{s}$ holds. The event trace $\strace$ is implemented as a vector that tracks eventon/eventoff operations at runtime. The library function $\cod{relabel}$ only statically checks information flow constraints with the assumption that $S$ holds when $\sevent{s}$ is set; otherwise, it does nothing (see~\ruleref{OS-Relabel-Succ} and~\ruleref{OS-Relabel-Fail} in Figure~\ref{fig:key_semantics}).
Similarly to prior works (\cite{lamba2024,carapace,filament}), our implementations of flows-to and release-to relations use Rust traits for compile-time checking as shown in Figure~\ref{fig:flows_to_imp}, covering the cases from the typing rules in Figures~\ref{fig:inference_rules} and~\ref{fig:inference_rules_no_extend}. The trait implementation models the security lattice, and by encoding these rules directly in the security functions (\cod{relabel} and \cod{Output}), we are able to implement and type-check all secure examples throughout the paper with our prototype.

\begin{figure}
\centering
\begin{subfigure}[t]{0.45\textwidth}
\centering
\begin{lstlisting}[style=ruststyle]
impl<T, S, P, Pp, L> FlowsTo<DRLabel<T, S, L, L>> for (DRLabel<T, S, P, Pp>, DynamicLabel<P, Pp>)
    where Pp: LEQ<L>, SEvent<S>: Holds, {}
\end{lstlisting}
\vspace{-2ex}
\caption{Implementing rule~\ruleref{R-From} with Rust's trait.}
\label{fig:flows_to_imp_R_From}
\end{subfigure}
\hfill
\begin{subfigure}[t]{0.45\textwidth}
\centering
\begin{lstlisting}[style=ruststyle]
impl<L, S, Q, Qq> FlowsTo<(DRLabel<(), S, Q, Qq>, DynamicLabel<Q, Qq>)> for DRLabel<(), S, L, L>
    where L: LEQ<Qq>, SEvent<S>: Holds {}
\end{lstlisting}
\vspace{-2ex}
\caption{Implementing rule~\ruleref{R-To} with Rust's trait.}
\Description{}
\vspace{-2ex}
\label{fig:flows_to_imp_R_To}
\end{subfigure}

\caption{Rule~\ruleref{R-From} and~\ruleref{R-To} from the implementation.}
\label{fig:flows_to_imp}
\end{figure}

Following Filament, our implementation also supports procedure calls, where labeled parameters and return values are statically checked. Consider the following example:

\begin{minipage}{0.43\textwidth}
\begin{lstlisting}[style=ruststyle]
fn pwd_reminder(..., mut password: DRLabel<String, B1, S, P>,)-> DRLabel<String, B1, P, P> {
  //set B1 if the user owns the password
  ...}
\end{lstlisting}
\end{minipage}
\hfill
\begin{minipage}{0.47\textwidth}
\begin{lstlisting}[style=ruststyle]
let password=DRLabel::<String, B1, S, P> {
    value: "secret123".to_string(),
    cond: false,...};
let reminder = pwd_reminder(..., password);
\end{lstlisting}
\end{minipage}

The parameter \texttt{password} is annotated with the label $\sevent{B1}?S\rightarrow P$, where the security event $\sevent{B1}$ is set when the user owns the password. Our implementation statically checks that the label of the actual argument \cod{password} does not exceed the corresponding annotation in \cod{pwd\_reminder}. Moreover, Rust's type system automatically infers the label of \cod{reminder} to match the function's declared return label.

For simplicity, our prototype currently lacks the following features (1) nested dynamic labels, (2) user-defined security events and (3) first-order functions. We plan to add them in future work.

\subsection{Case Studies}
\label{sec:cases}

To demonstrate the efficacy of dynamic release policy, we ported a conference management system written in Lifty~\cite{polikarpova2020} and
Civitas~\cite{clarkson2008civitas, juels2005coercion} (excluding its cryptography proof component) written in Jif in approximately 1,700 lines of Rust code on top of our prototype implementation. The evaluation demonstrates that, in contrast to prior systems, each built for a specific dynamic policy and its corresponding enforcement mechanism, our framework can accommodate a broader class of dynamic policies within the same type system. For example, Lifty supports only declassification policies and therefore cannot express erasure policies required by Civitas. The full implementation of the case studies is available in~\cite{implementation}.

\subsubsection{Conference Reviewing system}
\begin{figure}
\centering
\begin{subfigure}{0.43\textwidth}
\begin{lstlisting}[style=ruststyle]
eventoff(&mut password);
if is_requester_owner(current_user, password_owner) {
  eventon(&mut password);
  let relabel_result = relabel::< _, TrueB1, S, P, P, P > (&password, &conditions);
  relabel_result.assign_to(&mut passcode_out);
  output_to(&output_label, &passcode_out, &conditions);}
else {//reject, line 8-9 in pseudocode}
\end{lstlisting}
\subcaption{Implementation on the password example.}
\Description{}
\label{fig:relabel_call_a}
\end{subfigure}
\hspace{0.05\textwidth}
\begin{subfigure}{0.43\textwidth}
\begin{lstlisting}[style=mystyle]
//password: $\neg$checkuser?S^P, passcode:P
eventoff(checkuser);
if (isUser == true) {  //check user
  eventon(checkuser);
  passcode=relabel(password,"$\neg$checkuser?S^P" to "P" using checkuser);  
  output(P, passcode, using checkuser);}
else {  // administrator
  passcode = password; //reject
  output(P, passcode, using checkuser);}//nop
\end{lstlisting}
\subcaption{Pseudocode of the password example. }
\Description{Rust code showing a generic call to \texttt{relabel} with type arguments and guards.}
\label{fig:relabel_call_b}
\end{subfigure}
\caption{Password Disclosure Example from Conference Reviewing System.}

\label{fig:relabel_call}
\end{figure}

HotCRP~\cite{hotcrp} is a widely used conference reviewing system with known security flaws~\cite{yip2009}. Previous work Lifty~\cite{polikarpova2020} provides several code snippets that model the essence of those flaws in their Haskell-like domain specific language. Under the hood, Lifty detects such flaws by combining Labeled IO (LIO)~\cite{stefan2011} with Liquid Types~\cite{rondon2008}. We port the Lifty-based implementation with all core features such as score release, review release, and password reminder. Below, we discuss two representative examples.

\mypara{Password Disclosure} A representative security vulnerability in HotCRP involves password disclosure \cite{yip2009}. HotCRP has two features. First, when a user forgets the password, HotCRP sends a reminder email to the user. Second, HotCRP has an email preview mode that allows an administrator to display emails in the browser, including those reminder emails as shown in Figure~\ref{fig:relabel_call_b}. Consequently, in the email preview mode, an administrator can trigger a reminder email, which directly reveals a user’s password in the browser to the administrator~\cite{yip2009}.

To prevent the accidental leak of a password, we assign it the label $\neg\sevent{checkuser}?\High \rightarrow_t \Low$. This label ensures the password remains at a high security level until a specific security event, $\sevent{checkuser}$, is activated (\texttt{eventon} from line 3 in Figure~\ref{fig:relabel_call_a}). To achieve this, the implementation checks the functions \texttt{relabel} and \texttt{output\_to} at compile time. This check is enforced by the flows-to relations: if it compiles, then the execution reaches lines 4-6 in Figure~\ref{fig:relabel_call_b}; otherwise, it fails to compile according to the type system.

\mypara{Releasing papers' score}
The system checks other scenarios, including paper score release, which is a microbenchmark used in Lifty~\cite{polikarpova2020}. This scenario requires that the paper's score should be visible only when the review is finished. To address this, the scores of the papers are labeled with $\neg\sevent{done}?\High \rightarrow_t \Low$. The typing rule $\cod{relabel}$ checks the status indicator $\sevent{done}$, which determines when the result can be viewed.

\subsubsection{Secure voting system}

Civitas~\cite{clarkson2008civitas, juels2005coercion} is a secure voting system that enforces declassification and erasure policies across four phases: setup, registration, voting, and tabulation. Prior case studies~\cite{Chong2008} highlight two erasure scenarios involving credential shares by authorities and voters, written in Jif~\cite{myers2001}. We ported all Civitas components related to vote delivery and credential erasure.

For authorities, the credentials must be removed to ensure voters' anonymity after they are delivered to the voters. Hence, we label credentials with $\neg\sevent{delivered}? \cod{M} \rightarrow \top$, so that the credential must be erased after delivery. The output rule uses the fact that $\sevent{delivered}$ has never been $\true$ (i.e., $\absent{\sevent{delivered}}$) to check that the dynamic label can be released through the output channel. 

For voters, each share on the voter side is labeled with $\neg\sevent{postCombined}?\cod{M} \rightarrow \top$. After receiving all the shares from the authorities and combining the shares, the voter should delete any record provided by the authorities. The proof tree for output is the same as the previous example. The code is available in~\cite{implementation}.

\subsection{Performance and Correctness}
We implemented the semantics shown in Figure~\ref{fig:key_semantics} for the novel constructs introduced in our language. Since the unique semantics require tracking a trace of security events and evaluating event conditions at runtime, our approach may introduce additional runtime overhead. To assess this overhead, we measured the median execution time over 100 runs for both the original and the verified Rust implementations across our case studies. The results are shown below.

\begin{table}[h]
\centering

\begin{tabular}{l r r }
\hline
\textbf{} & \textbf{Original} & \textbf{IFC}  \\
\hline
conference & 0.029ms   & 0.033ms   \\
Civitas     & 5.694ms & 5.736ms  \\
\hline
\end{tabular}
\vspace{-2ex}
\end{table} 
The evaluation results suggest that the overhead is low. We attribute the low overhead to the fact that even though the event trace grows during program execution, the overall overhead remains low as (1) event trace only grows with eventon/eventoff, and output commands, and (2) the trace is only checked with a relabel command. Since both kinds of policy-specific commands occur infrequently in practice, the overall overhead remains low.

Moreover, because \cod{relabel} and \cod{output} commands may alter program semantics when the corresponding \cod{using} clauses are not satisfied, we verified that all case studies produce the same outputs as the original programs without information flow control.

\section{Related Work}
\mypara{Dynamic Policies} To address the limitations of noninterference, prior work introduces various mechanisms~\cite{askarov2007,askarov2009,Chong2008,clarkson2008civitas,Chong2005} to allow information sensitivity/integrity to change over time. Additional studies include extensions to gradual release~\cite{banerjee2008}, temporally verified logic~\cite{balliu2011}, and security-specific languages~\cite{swamy2006,hicks2005}, which support role-based information flow policies with dynamic updates. These works enable flexible label updates while accommodating both Denning-style and role-based lattices. Recent work~\cite{cecchetti2025} introduces nonmalleable progress leakage, adapting nonmalleable information flow to secure progress channels for downgrading policies.
More recently, Ahmadian~\cite{ahmadian2022} introduced an attacker knowledge-based information flow policy. Zegarelli et al.~\cite{zegarelli2024} propose a modular enforcement mechanism that supports multiple interpretations of facets within a single framework. In contrast, our work uses the dynamic release policy~\cite{Li2022}, which provides a uniform formalization of various dynamic policies.

\mypara{Enforcement of Dynamic Policies in Applications}
Building on Myers’ influential work~\cite{myers1998} in establishing a formal semantics for decentralized labels, various applications have been built on top of Jif~\cite{myers2001}, including some that involve dynamic policies such as declassification and erasure. 
Dynamic policies are also widely used in systems research. Guarnieri et al.~\cite{guarnieri2016,guarnieri2019} explore the “forgetful attacker” model~\cite{askarov2012} to implement declassification via a security monitor. In addition, dynamic policies have been applied to applications such as conference reviewing systems~\cite{polikarpova2020,stefan2011}, Battleship Game~\cite{myers2001,stoughton2014}, and social media platforms~\cite{bauereiss2018}.

In comparison, our dynamic release Rust library constitutes the first approach capable of expressing and verifying the full range of such policies. A detailed comparison with dynamic release policy and other variants of dynamic policies can be found in~\cite{Li2022}.

\mypara{Soundness Proof}
Most information flow policies enforce noninterference, and are typically enforced via static type systems, whose soundness has been proven using both big- and small-step semantics~\cite{flowcaml,volpano1996,li2017arxiv,cecchetti2021,rajani2020,gregersen2021}. As discussed in Section~\ref{sec:endtoendproof}, these proof techniques are not directly applicable to dynamic policies.

Prior work in dynamic policies, notably~\cite{Chong2008}, established the soundness of static enforcement for declassification and erasure. Unlike~\cite{Chong2008}, which uses two separate sets of rules, our dynamic release policy is a superset of those supported policies and provides a single set of rules for various dynamic policies. Furthermore, differences in policy formalization mean that the soundness proof of~\cite{Chong2008} is not applicable to our dynamic release enforcement.

\section{Discussion and Future work}

\subsection{Annotation Burden}
In this paper, we assume that \cod{using} clauses are provided by programmers. From the programmer's perspective, an appropriate annotation should specify, for each \cod{output} and \cod{relabel} command: (1) the conditions under which information may be released safely, thereby enabling the command to satisfy the typing rules (Section~\ref{sec:typesystem}); and (2) the conditions that hold at the corresponding program point, ensuring that the annotated program exhibits the same behavior as the original program whenever the original program is secure (Section~\ref{sec:commandsemantics}).
Across all of our case studies (Section~\ref{sec:case_studies}), these conditions are sufficiently simple to identify and annotate manually. However, constructing such annotations may become more challenging in large code bases. Hence, an important direction for future work is the automatic inference of \cod{using} clauses through static analysis. Our language design is well suited to this objective because security events are explicitly distinguished from ordinary program variables and can only be modified by the \cod{eventon} and \cod{eventoff} commands (Section~\ref{sec:syntax}). This separation enables a compositional inference strategy that we plan to explore: an intraprocedural dataflow analysis can propagate event-related facts along the control-flow graph to infer conditions that must hold at each program point, while the effect of each procedure on security events can be summarized independently by the set of events it may modify.

Moreover, our type system currently does not support label inference. Consequently, all label annotations, including the ``from'' and ``to'' labels in \cod{relabel} commands, must be provided explicitly by programmers. A promising direction for future work is to extend the language with automatic label inference by formulating the flows-to relation as a system of constraints to be solved by a constraint solver. The primary challenge lies in the presence of dynamic release labels, which introduce dependencies among these constraints and substantially complicate the inference process. Addressing such dependencies is beyond the capabilities of our current Rust-based implementation.

\subsection{Language Features}
\label{sec:language_features}
Our primary goal is to establish a theoretical foundation for enforcing dynamic release policies. In this context, the use of a minimal imperative language is deliberate and consistent with prior type systems that come with formal soundness proofs. Building on this foundation, an important direction for future work is to extend both the source language and the type system with more expressive language features, such as first-class functions and heap manipulation. In pursuing these extensions, we plan to draw on established designs such as Jif~\cite{myers2001} and Cocoon~\cite{lamba2024}.

For example, supporting first-class functions requires assigning security labels to those functions (including their parameters, return values, and side effects) just as for any other value in the program. These labels capture the security constraints governing the use of a function, such as the maximum sensitivity of its arguments and the minimum security level of the $\pc$ context in which it may be invoked securely. In the setting of dynamic release policies, an additional challenge arises because these constraints also depend on the state of security events at the call site. Consequently, the annotation may also require a \cod{using} clause, specifying the necessary event states to ensure that function invocations remain secure.

\section{Conclusion}

In this paper, we present the first sound enforcement of the dynamic release policy introduced by Li and Zhang~\cite{Li2022}. Compared to prior enforcements of dynamic policies, which are based on ad hoc, policy-specific enforcement mechanisms, our work offers simple, intuitive, and unified enforcement of various dynamic policies, such as declassification, endorsement, and erasure policies. We developed a metalanguage with carefully designed features, such as $\cod{relabel}$, $\cod{eventon}$, $\cod{eventoff}$ commands for writing code governed by embedded dynamic release policies. We further developed a type system that is formally proven to soundly enforce dynamic release policies, and implemented a prototype as an extension to Rust.

\section*{Data Availability Statement}
An initial artifact for this paper consists of the prototype and the case studies from Section~\ref{sec:case_studies}. A version of the source code is available at the link~\cite{implementation}. 
The full proof can be found in~\cite{full_proof}.

\begin{acks}
We would like to express our gratitude to the anonymous reviewers for their insightful feedback and suggestions. We also appreciate the insights from Peixuan Li and Quan Zhou. This research is sponsored
by National Science Foundation grants CNS 2401496 and CNS 2401182. 
\end{acks}

\bibliographystyle{ACM-Reference-Format}
\bibliography{ref}

\clearpage
\appendix
\setcounter{page}{1}
\renewcommand{\thepage}{\arabic{page}}

\begin{figure}
\begin{mathpar}
\mprset{flushleft} 
\inferrule[]
  { }
  {\configs{m,n} \Downarrow n}

\inferrule[]
  { }
  {\configs{m,x} \Downarrow m(x)}

\inferrule[]
  { \configs{m,e_1}\Downarrow n_1 \\\configs{m,e_2}\Downarrow n_2\\n_3=n_1\op\; n_2}
  {\configs{m,e_1 \;\op \;e_2} \Downarrow n_3}
\end{mathpar}
\caption{Big-Step Semantics of Expression}
\Description{}
\label{fig:arith_semantics}
\end{figure}

\begin{figure*}
\begin{mathpar}
\mprset{flushleft}

\inferrule[OS-Skip] 
  { }
  { \configThree{\Mem}{\Skip;c}{\strace}\rightarrow\configThree{\Mem}{c}{\strace} }
\and
\inferrule[OS-Sequence] 
  {\configThree{\Mem}{c_1}{\strace}\rightarrow\configThree{\Mem'}{c_1'}{\strace'}}
  { \configThree{\Mem}{c_1;c_2}{\strace}\rightarrow\configThree{\Mem'}{c_1';c_2}{\strace'} }
\and
\inferrule[OS-Assign] 
  { \configs{\Mem,\expr_1} \Downarrow n }
  { \configThree{\Mem}{x:=e_1}{\strace} \rightarrow \configThree{\Mem[x:=n]}{\Skip}{\strace}}
\and
\inferrule[OS-If-True] 
{ \configs{\Mem,\expr} \Downarrow \true}
{ \configThree{\Mem}{\ifcmd{\expr}{c_1}{c_2}}{\strace}\rightarrow \configThree{\Mem}{c_1}{\strace}} 
\and
\inferrule[OS-If-False] 
{ \configs{\Mem,\expr} \Downarrow \false}
{ \configThree{\Mem}{\ifcmd{\expr}{c_1}{c_2}}{\strace}\rightarrow \configThree{\Mem}{c_2}{\strace}} 
\and
\inferrule[OS-While] 
{ }
{ \configThree{\Mem}{\while{\expr}{c}}{\strace}\rightarrow \configThree{\Mem}{\ifcmd{b}{(c;\while{b}{c})}{\Skip}}{\strace} }
\and
\inferrule[OS-Output-Succ]
{\configTwo{\Mem}{e}\evalto v \quad
\inferrule{\forall i \in 0...k.~\\\\
\cod{eval}(\sevent{s_i},\sigma)=\true}{} \quad
\strace'=
\text{
$\begin{cases}
\strace\cdot\sevent{o_x@\level}, & \text{$\expr$ is a variable $x$} \\
\strace, & \text{otherwise}
\end{cases}$
}}
{\configThree{\Mem}{\outcmd{L}{e}{\sevent{s_0}\dots\sevent{s_k}}}{\strace } \xrightarrow[]{\configThree{\level}{v}{\strace}}
\configThree{\Mem}{\Skip}{\strace} }
\and
\inferrule[OS-Output-fail]
{\configTwo{\Mem}{e}\evalto v \\ 
\inferrule{\exists i \in 0...k.~\\\\
\cod{eval}(\sevent{s_i},\sigma) = \false}{}
}
{\configThree{\Mem}{\outcmd{L}{e}{\sevent{s_0}\dots\sevent{s_k}}}{\strace} \xrightarrow[]{}
\configThree{\Mem}{\Skip}{\strace}}
\and
\inferrule[OS-eventon]
{}
{\configThree{\Mem}{\cod{eventon}(\sevent{s})}{\strace} \To 
\configThree{\Mem}{\Skip}{\strace\cdot \sevent{s} }}
\and
\inferrule[OS-eventoff]
{}
{\configThree{\Mem}{\cod{eventoff}(\sevent{s})}{\strace} \To 
\configThree{\Mem}{\Skip}{\strace\cdot \neg\sevent{s} }}
\and
\inferrule[OS-Relabel-Succ] 
{\configTwo{\Mem}{e} \Downarrow n \\ \forall i \in 0...k.~\cod{eval}(\sevent{s_i},\sigma)=\true }
{ \configThree{m}{x:=\relabel{e}{\lab_f}{\level_t}{\sevent{s_0}\dots\sevent{s_k}}}{\strace}\rightarrow \configThree{\Mem[x:=n]}{\Skip}{\strace}}
\and
\inferrule[OS-Relabel-Fail] 
{\exists i \in 0...k.~\cod{eval}(\sevent{s_i},\sigma)=\false}
{ \configThree{\Mem}{x:=\relabel{e}{\lab_f}{\level_t}{\sevent{s_0}\dots\sevent{s_k}}}{\strace}\rightarrow \configThree{\Mem}{\Skip}{\strace}}
\end{mathpar}
\caption{Operational Semantics for Commands}
\Description{}
\label{fig:command_semantics}
\end{figure*}

\section{Full Operational Semantics of Expressions and Commands}
\label{appendix:os}
Figure~\ref{fig:arith_semantics} and Figure~\ref{fig:command_semantics} present the full operational semantics.

\section{Full Static Flows-to Rules}
\label{appendix:flows}
\begin{figure}
\framebox{
\textbf{General Lattice Rules}}
\begin{mathpar}
\mprset{flushleft} 
\inferrule[Lattice] 
  { \level\LEQ\level' }
  {S\vdash \level \flowsto\level' }

\inferrule[Trans]
  {\vdash \lab\flowsto \lab' \\ \vdash \lab'\flowsto \lab''}
  {S\vdash \lab\flowsto \lab''} 
\end{mathpar}

\framebox{
\textbf{Directional Inference Rules}}
\begin{mathpar}
\mprset{flushleft}
\inferrule[R-To] 
  { S\vDash \neg \sevent{s_2} \\  \vdash \level\flowsto \labtwo'}
  { S\vdash \level\flowsto \sevent{s_2}?\labtwo\rightarrow \labtwo'}
\and
\inferrule[D-To] 
  {S\vDash \neg\sevent{s_2} \\\neg\sevent{s_1} \Rightarrow \neg\sevent{s_2}  \\  
  \vdash \sevent{s_1}?\level \rightarrow \lab' \flowsto \labtwo'}
  { S\vdash \sevent{s_1}?\level\rightarrow \lab' \flowsto \sevent{s_2}?\labtwo\rightarrow \labtwo'}
\and
\inferrule[R-From] 
  { S\vDash \neg \sevent{s_1} \\  \vdash \lab'\flowsto \level}
  { S\vdash \sevent{s_1}?\lab\rightarrow\lab'\flowsto \level}
\and
\inferrule[D-From] 
  {S\vDash \neg\sevent{s_1} \\\neg\sevent{s_2} \Rightarrow \neg\sevent{s_1}  \\  
  \vdash \lab' \flowsto \sevent{s_2}?\level\rightarrow \labtwo'}
  { S\vdash \sevent{s_1}?\lab\rightarrow \lab' \flowsto \sevent{s_2}?\level\rightarrow \labtwo'}
\and
\inferrule[D-To Two] 
{ S\vdash \lab\flowsto \level \quad \vdash \lab'\flowsto \level}
  {S\vdash \sevent{s_1}?\lab\rightarrow \lab'\flowsto \level}
\and
\inferrule[R-To Two] 
{\neg\sevent{s_2} \Rightarrow \neg\sevent{s_1}\quad S \vdash \lab\flowsto \sevent{s_2}?\labtwo\rightarrow \labtwo' \quad \vdash \lab'\flowsto \sevent{s_2}?\labtwo\rightarrow \labtwo'}
  {S\vdash \sevent{s_1}?\lab\rightarrow \lab'\flowsto \sevent{s_2}?\labtwo\rightarrow \labtwo'}
  \and
\inferrule[D-From Two] 
{ S\vdash \level \flowsto \lab \quad \vdash \level\flowsto \lab'}
  {S\vdash \level \flowsto \sevent{s_2}?\labtwo\rightarrow \labtwo'}
\and
\inferrule[R-From Two] 
{\neg\sevent{s_1} \Rightarrow \neg\sevent{s_2}\quad S \vdash \sevent{s_1}?\lab\rightarrow \lab' \flowsto \labtwo \quad \vdash \sevent{s_1}?\lab\rightarrow \lab' \flowsto \labtwo'}
  {S\vdash \sevent{s_1}?\lab\rightarrow \lab'\flowsto \sevent{s_2}?\labtwo\rightarrow \labtwo'}
\end{mathpar}

\framebox{
\textbf{Bi-directional Inference Rules}}
\begin{mathpar}
\mprset{flushleft}
\inferrule[Bi-To] 
  { \vdash \level \flowsto \labtwo \\ \vdash\level \flowsto \labtwo'}
  {S\vdash  \level \flowsto \sevent{s_2}? \labtwo \leftrightarrow \labtwo'}
\and
\inferrule[Bi-From] 
  { \vdash \lab\flowsto \level  \\ \vdash \lab'\flowsto \level}
  {S\vdash \sevent{s_1}?\lab \leftrightarrow \lab'\flowsto  \level}
\and 
\inferrule[Bi-Both] 
  { \sevent{s_1} \Leftrightarrow \sevent{s_2} \\ \vdash \lab\flowsto \labtwo \\ \vdash \lab'\flowsto \labtwo'}
  {S\vdash \sevent{s_1}?\lab \leftrightarrow \lab'\flowsto  \sevent{s_2}? \labtwo \leftrightarrow \labtwo'}
\end{mathpar}

\caption{In top-down order, these are the general rules, the directional inference rules, and bidirectional rules.}
\Description{}
\label{fig:inference_rules_full}
\end{figure}
Figure~\ref{fig:inference_rules_full} presents the full static flows-to rules.

\section{Typing Rules for Expressions}
\label{appendix:arrows}
The typing rules for expressions are in Figure~\ref{fig:type_expression}.
\begin{figure*}
{\small
\begin{mathpar}
\mprset{flushleft} 
    
\inferrule[Var]
{ }
  {\Gamma\vdash x:\Gamma(x)}
\and  
\inferrule[Num]
  { }
  {\Gamma\vdash n : \bot}
\and
\inferrule[Op]
  {\Gamma\vdash \expr_1: \lab_1 \\ \Gamma\vdash \expr_2 :\lab_2 \\ \vdash \lab_1\flowsto \lab_2}
  {\Gamma\vdash \expr_1 \op\; \expr_2: \lab_2 } 
\and
\inferrule[SubType]
  {\Gamma\vdash \expr: \lab \\ \vdash  \lab\flowsto \lab' }
  {\Gamma\vdash \expr: \lab'}
\end{mathpar}
}
\caption{Typing Rules for Expressions.}
\Description{}
\label{fig:type_expression}
\end{figure*}

\section{Soundness of Flows-To and Release-To Relation}
\label{appendix:t1proof}
\begin{theorem*}[Local Soundness, Theorem~\ref{theorem:setsound}]

For all policies p and q, a trace $\strace$, a set of security events $S={\sevent{s_0}\dots\sevent{s_k}}$, if $S\vdash \lab\flowsto \labtwo$ and $S$ holds on $\strace$, then $\lab \LEQ_\strace \labtwo$.
\end{theorem*}

\begin{proof} 
By Definition~\ref{def:partial_dynamic}, it is sufficient to show that for any $\strace'$ such that $\strace\preceq\strace'$, we have $\auxfuncold{\lab}{\strace'}\LEQ\auxfuncold{\labtwo}{\strace'}$. The proof is by induction on the static inference rules:

\begin{itemize}
\item \ruleref{Lattice} ($S\vdash \level \flowsto\level'$): It is trivial.

\item \ruleref{Trans} ($S\vdash \lab\flowsto \lab''$): It is trivial.

\item \ruleref{R-To} ($S\vdash \level\flowsto \sevent{s_2}?\labtwo\xrightarrow{} \labtwo'$):
By the assumption, we have $\vdash \level\flowsto \labtwo'$ and $S\vDash \neg\sevent{s_2}$. Since $S$ holds on $\strace$ by assumption, $\strace \preceq \strace'$, and $S\vDash \neg\sevent{s_2}$, there must be some index $i$ on $\strace'$ at which $\sevent{s_2}$ first turns $\false$. According to the semantics of Figure~\ref{fig:label-semantics}, we know that $\auxfuncold{\sevent{s_2}?\labtwo\rightarrow \labtwo'}{\strace'}=\auxfuncold{\labtwo'}{\strace'^{[i:]}}$.  
    
    From the induction hypothesis and the assumption that $\vdash \level\flowsto \labtwo'$, we get $\auxfuncold{\level}{\strace'^{[i:]}}\LEQ \auxfuncold{\labtwo'}{\strace'^{[i:]}}$.
    Hence, $\auxfuncold{\level}{\strace'}= \level = \auxfuncold{\level}{\strace'^{[i:]}}\LEQ \auxfuncold{\labtwo'}{\strace'^{[i:]}}=\auxfuncold{\sevent{s_2}?\labtwo\rightarrow \labtwo'}{\strace'}$.


\item \ruleref{R-From} Symmetric to the case above.

\item \ruleref{D-To}
Similar to \ruleref{R-To}, we first derive $\auxfuncold{\sevent{s_2}?\labtwo\rightarrow \labtwo'}{\strace'}=\auxfuncold{\labtwo'}{\strace'^{[i:]}}$ from the assumptions, where $i$ is the index when $\sevent{s_2}$ first turns $\false$. 

If $\sevent{s_1}$ has been $\false$ on $\strace'$, let $j$ be the index when it turns false the first time. Since $\neg\sevent{s_1} \Rightarrow \neg\sevent{s_2}$, it is guaranteed that $j\geq i$. According to the semantics of Figure~\ref{fig:label-semantics}, we know that $\auxfuncold{\sevent{s_1}?\level\rightarrow \lab'}{\strace'}=\auxfuncold{\lab'}{\strace'^{[j:]}}=\auxfuncold{\sevent{s_1}?\level\rightarrow \lab'}{\strace'^{[i:]}}$ where the second equation holds as $\sevent{s_1}$ must remain true between $i$ and $j$. Moreover,
from the induction hypothesis and the assumption $\vdash \sevent{s_1}?\level\rightarrow\lab' \flowsto \labtwo'$, we know that 
$\auxfuncold{\sevent{s_1}?\level\rightarrow\lab'}{\strace'^{[i:]}} \LEQ \auxfuncold{\labtwo'}{\strace'^{[i:]}}$. Hence $\auxfuncold{\sevent{s_1}?\level \rightarrow \lab'}{\strace'}=\auxfuncold{\sevent{s_1}?\level\rightarrow\lab'}{\strace'^{[i:]}}\LEQ \auxfuncold{\labtwo'}{\strace'^{[i:]}}=\auxfuncold{\sevent{s_2}?\labtwo\rightarrow \labtwo'}{\strace'}$. 

If $\sevent{s_1}$ has never been $\false$ on $\strace'$, we know that $\auxfuncold{\sevent{s_1}?\level \rightarrow \lab'}{\strace'}=\level=\auxfuncold{\sevent{s_1}?\level \rightarrow \lab'}{\strace'^{[i:]}}$. 
From the induction hypothesis and the assumption $\vdash \sevent{s_1}?\level\rightarrow\lab' \flowsto \labtwo'$, we know that 
$\auxfuncold{\sevent{s_1}?\level\rightarrow\lab'}{\strace'^{[i:]}} \LEQ \auxfuncold{\labtwo'}{\strace'^{[i:]}}$. Hence $\auxfuncold{\sevent{s_1}?\level \rightarrow \lab'}{\strace'}=\auxfuncold{\sevent{s_1}?\level\rightarrow\lab'}{\strace'^{[i:]}}\LEQ \auxfuncold{\labtwo'}{\strace'^{[i:]}}=\auxfuncold{\sevent{s_2}?\labtwo\rightarrow \labtwo'}{\strace'}$.

\item \ruleref{D-From} Symmetric to the case above.

\item \ruleref{D-To Two} ($S\vdash \sevent{s_1}?\lab\rightarrow \lab'\flowsto \level$): By the assumption, we have $S\vdash \lab\flowsto \level$ and $\vdash \lab'\flowsto\level$. 

If $\sevent{s_1}$ has been $\false$ on $\strace'$, let $i$ be the first such index on $\strace'$. By the semantics of Figure~\ref{fig:label-semantics}, we have $\auxfuncold{\sevent{s_1}?\lab\rightarrow \lab'}{\strace'} = \auxfuncold{\lab'}{\strace'^{[i:]}}$. From the induction hypothesis we can derive that $\auxfuncold{\lab'}{\strace'^{[i:]}}\LEQ \auxfuncold{\level}{\strace'^{[i:]}} = \level = \auxfuncold{\level}{\strace'}$. So $\auxfuncold{\sevent{s_1}?\lab\rightarrow \lab'}{\strace'} \LEQ \auxfuncold{\level}{\strace'}$.

If $\sevent{s_1}$ has never been $\false$ on $\strace'$, by the semantics of Figure~\ref{fig:label-semantics}, we have $\auxfuncold{\sevent{s_1}?\lab\rightarrow \lab'}{\strace'} = \auxfuncold{\lab}{\strace'}$. From the induction hypothesis we can derive that $\auxfuncold{\lab}{\strace'}\LEQ \auxfuncold{\level}{\strace'}$. So $\auxfuncold{\sevent{s_1}?\lab\rightarrow \lab'}{\strace'}\LEQ \auxfuncold{\level}{\strace'}$.

\item\ruleref{D-From Two} Symmetric to the case above.

\item \ruleref{R-To Two} ($S\vdash \sevent{s_1}?\lab\rightarrow \lab'\flowsto \sevent{s_2}?\labtwo \rightarrow \labtwo'$): By the assumption, we have $S\vdash \lab\flowsto \sevent{s_2}?\labtwo \rightarrow \labtwo'$ and $\vdash \lab'\flowsto \sevent{s_2}?\labtwo \rightarrow \labtwo'$. 

If $\sevent{s_2}$ has been $\false$ on $\strace'$ and let $i$ be the first such index on $\strace'$. Since $\neg\sevent{s_2} \Rightarrow \neg\sevent{s_1}$, it is guaranteed that $\sevent{s_1}$ has been $\false$ on $\strace'$ with the first index $j$ such that $j\leq i$. According to the semantics of Figure~\ref{fig:label-semantics}, we know that $\auxfuncold{\sevent{s_2}?\labtwo\rightarrow \labtwo'}{\strace'}=\auxfuncold{\labtwo'}{\strace'^{[i:]}}=\auxfuncold{\sevent{s_2}?\labtwo\rightarrow \labtwo'}{\strace'^{[j:]}}$ where the second equation holds as $\sevent{s_2}$ must remain true between $j$ and $i$. Moreover,
from the induction hypothesis and the assumption $\vdash \lab'\flowsto \sevent{s_2}?\labtwo \rightarrow \labtwo'$, we know that 
$\auxfuncold{\lab'}{\strace'^{[j:]}} \LEQ \auxfuncold{\sevent{s_2}?\labtwo \rightarrow \labtwo'}{\strace'^{[j:]}}$. Hence $\auxfuncold{\sevent{s_1}?\lab \rightarrow \lab'}{\strace'}=\auxfuncold{\lab'}{\strace'^{[j:]}}\LEQ \auxfuncold{\sevent{s_2}?\labtwo\rightarrow\labtwo'}{\strace'^{[j:]}} =\auxfuncold{\sevent{s_2}?\labtwo\rightarrow \labtwo'}{\strace'}$.

If $\sevent{s_1}$ has never been $\false$ on $\strace'$, we have $\auxfuncold{\sevent{s_1}?\lab \rightarrow \lab'}{\strace'}=\auxfuncold{\lab}{\strace'}$ according to the semantics of Figure~\ref{fig:label-semantics}. Moreover,
from the induction hypothesis and the assumption $S \vdash \lab\flowsto \sevent{s_2}?\labtwo \rightarrow \labtwo'$, we know that $\auxfuncold{\lab}{\strace'} \LEQ \auxfuncold{\sevent{s_2}?\labtwo \rightarrow \labtwo'}{\strace'}$. Hence $\auxfuncold{\sevent{s_1}?\lab \rightarrow \lab'}{\strace'} \LEQ \auxfuncold{\sevent{s_2}?\labtwo \rightarrow \labtwo'}{\strace'}$.

\item\ruleref{R-From Two} Symmetric to the case above.

\item\ruleref{Bi-Both}:
By the assumption, we have $\sevent{s_1} \Leftrightarrow \sevent{s_2}$, $\vdash \lab\flowsto \labtwo$ and $\vdash \lab'\flowsto \labtwo'$. 

With $\sevent{s_1} \iff \sevent{s_2}$, it is guaranteed that conditions $\sevent{s_2}$ and $\sevent{s_1}$ have the same value. If they both turn from $\true$ to $\false$ at index $i$ towards the end of $\strace'$, we have $\auxfuncold{\sevent{s_1}?\lab\leftrightarrow\lab'}{\strace'}=\auxfuncold{\lab'}{\strace'^{[i+1:]}}$ and 
$\auxfuncold{\sevent{s_2}?\labtwo\leftrightarrow\labtwo'}{\strace'}=\auxfuncold{\labtwo'}{\strace'^{[i+1:]}}$ according to the semantics of Figure~\ref{fig:label-semantics}. Moreover, by the induction hypothesis and the assumption that $\vdash \lab'\flowsto \labtwo'$, we have $\auxfuncold{\lab'}{\strace'^{[i+1:]}}\LEQ \auxfuncold{\labtwo'}{\strace'^{[i+1:]}}$. Hence, $\auxfuncold{\sevent{s_1}?\lab\leftrightarrow\lab'}{\strace'}\LEQ \auxfuncold{\sevent{s_2}?\labtwo\leftrightarrow\labtwo'}{\strace'}$.

If they both turn from $\false$ to $\true$ at index $i$ towards the end of $\strace'$, we have $\auxfuncold{\sevent{s_1}?\lab\leftrightarrow\lab'}{\strace'}=\auxfuncold{\lab}{\strace'^{[i+1:]}}$ and 
$\auxfuncold{\sevent{s_2}?\labtwo\leftrightarrow\labtwo'}{\strace'}=\auxfuncold{\labtwo}{\strace'^{[i+1:]}}$ according to the semantics of Figure~\ref{fig:label-semantics}. Moreover, by the induction hypothesis and the assumption that $\vdash \lab\flowsto \labtwo$, we have $\auxfuncold{\lab}{\strace'^{[i+1:]}}\LEQ \auxfuncold{\labtwo}{\strace'^{[i+1:]}}$. Hence, $\auxfuncold{\sevent{s_1}?\lab\leftrightarrow\lab'}{\strace'}\LEQ \auxfuncold{\sevent{s_2}?\labtwo\leftrightarrow\labtwo'}{\strace'}$.

\item\ruleref{Bi-To} ($S\vdash\level \flowsto \sevent{s_2}? \labtwo \leftrightarrow\labtwo'$):
By the assumption we have $\vdash\level \flowsto\labtwo$ and $\vdash\level \flowsto\labtwo'$. If $\sevent{s_2}$ turns from $\true$ to $\false$ at index $i$ towards the end of $\strace'$, we have $\auxfuncold{\sevent{s_2}?\labtwo\leftrightarrow\labtwo'}{\strace'}=\auxfuncold{\labtwo'}{\strace'^{[i+1:]}}$ according to the semantics of Figure~\ref{fig:label-semantics}. By the induction hypothesis and the assumption that $\vdash \level\flowsto \labtwo'$, we have $\auxfuncold{\level}{\strace'^{[i+1:]}}\LEQ \auxfuncold{\labtwo'}{\strace'^{[i+1:]}}$. Hence, $\auxfuncold{\level}{\strace'}=\level=\auxfuncold{\level}{\strace'^{[i+1:]}}\LEQ \auxfuncold{\sevent{s_2}?\labtwo\leftrightarrow\labtwo'}{\strace'}$.

If $\sevent{s_2}$ turns from $\false$ to $\true$ at index $i$ towards the end of $\strace'$, we have $\auxfuncold{\sevent{s_2}?\labtwo\leftrightarrow\labtwo'}{\strace'}=\auxfuncold{\labtwo}{\strace'^{[i+1:]}}$ according to the semantics of Figure~\ref{fig:label-semantics}. By the induction hypothesis and the assumption that $\vdash \level\flowsto \labtwo$, we have $\auxfuncold{\level}{\strace'^{[i+1:]}}\LEQ \auxfuncold{\labtwo}{\strace'^{[i+1:]}}$. Hence, $\auxfuncold{\level}{\strace'}=\level=\auxfuncold{\level}{\strace'^{[i+1:]}}\LEQ \auxfuncold{\sevent{s_2}?\labtwo\leftrightarrow\labtwo'}{\strace'}$.

\item\ruleref{Bi-From} Symmetric to the case above.

\end{itemize}
\end{proof}

\begin{theorem*}[Soundness of Release-To Relation, Theorem~\ref{theorem:setsound_no_extend}]
For all policies p, a trace $\strace$, a set of security events $S={\sevent{s_0}\dots\sevent{s_k}}$, if $S\vdash\lab\releaseto\level$ and $S$ holds on $\strace$, then $\auxfuncold{p}{\strace}\LEQ\level$.
\end{theorem*}
\begin{proof} 
The proof is by induction on the static inference rules. The only interesting case is Rule~\ruleref{R-From-F}; the rest are the same as their corresponding rules in Theorem~\ref{theorem:setsound}.
\begin{itemize}


\item \ruleref{R-From-F} ($S\vdash \sevent{s_1}?\lab\xrightarrow{} \lab'\releaseto \level$):
By the assumption, we have $S \vdash \lab\releaseto \level$ and $S\vDash\absent{\neg\sevent{s_1}}$. Since $S$ holds on $\strace$ and $S\not\vDash \neg\sevent{s_1}$, according to the semantics of Figure~\ref{transient_policy_label_semantics}, we know that
$\auxfuncold{c?\lab\rightarrow \lab'}{\traceout{\strace}}$ evaluates to $\auxfuncold{\lab}{\strace}$. From the induction hypothesis and the assumption that $S \vdash \lab\releaseto \level$, we get $\auxfuncold{\lab}{\strace}\LEQ\level$. Hence, $\auxfuncold{\cnd{cnd}?\lab\rightarrow \lab'}{\strace} = \auxfuncold{\lab}{\strace}\LEQ\level$ which the relation stands.


\end{itemize}
\end{proof}

\section{Soundness of Type System}
\label{appendix:type_proof}
The structure of the proof follows a previous small-step noninterference proof~\cite{li2017arxiv}, formed with supporting lemmas. In the rest of the proof, we introduce a distinguished \emph{dynamic} label $\fixedlab$ to denote the source label that is considered ``sensitive'' per security reasoning. In addition, we introduce a distinguished level $\fixedlevel$ to denote the output that can be viewed by an attacker at the security level $\fixedlevel$. Since the proof below is applicable to any $\fixedlab$ and $\fixedlevel$, we show that dynamic release policy is enforced by the proposed type system.

\subsection{Typing Rules for Expressions}

The typing rules for expressions (Figure~\ref{fig:type_expression}) are straightforward. Rules~\ruleref{Var} and~\ruleref{Num} compute the policy on a variable and a value, respectively. Rule~\ruleref{Op} ensures that the policy $\lab_2$ serves as an upper-bound on the policies on subexpressions $\expr_1$ and $\expr_2$. Since our type system supports dynamic labels, we include a subtyping rule~\ruleref{SubType} to compute the upper bound of an expression and enable comparisons among different labels.

\subsection{Augmented Language}

Following~\cite{li2017arxiv}, we first introduce an augmented semantics that ``protects'' all sensitive values with brackets. Hence, a memory is said to be ``well-formed'' if all sensitive values with brackets are only stored in variables that are considered as sensitive at the end of $\strace$: $\fixedlab \LEQ_\strace \Gamma(x)$.

\subsubsection{Augmented Syntax}
We augment memories to allow mapping variables $x$ to a bracketed value. Intuitively, a bracketed value is information to be protected: given any two executions on initial memory states that agree on all labels except $\fixedlab$, only bracketed values might differ throughout program execution. In this sense, we augment the language syntax to include bracketed commands and values. 
\begin{align*}
 e::=\dots\mid [n]\\
 c::=\dots\mid [c]
\end{align*}

\subsubsection{Equivalence on Commands and Memories}
Based on the intuition that bracketed values are the ``sensitive'' data to be protected, we define an equivalence relation $\approx$ on values and commands in Figure~\ref{fig:equivalence_commands}. Two values are indistinguishable if either (1) they both have the same value without brackets, or (2) both have bracketed values. Accordingly, two memories $\Mem_1$ and $\Mem_2$ are indistinguishable if for each variable $x$ that is considered as sensitive, $\Mem_1(x)$ and $\Mem_2(x)$ are indistinguishable: 
\[
  {m_1 \approx_\strace m_2} \iff {\forall x.~ \fixedlab\not\LEQ_\strace \Gamma(x) \Rightarrow m_1(x)\approx m_2(x) }
\]

Bracketed commands represent the commands being executed under ``high'' $\pc$, in order to rule out implicit flows. Hence, two equivalent commands are either structurally the same with equivalent expressions,  or both of them are bracketed commands.
\begin{figure*}
\begin{mathpar}
\mprset{flushleft} 


\inferrule[] 
{ }
  {n\approx n}
\and
\inferrule[] 
{ }
  {[n_1]\approx [n_2]}
\and
\inferrule[] 
 { }
  {[c]\approx [c']}
\and
\inferrule[]
  { }
  {\Skip\approx \Skip}
\and
\inferrule[]
  {c_1 \approx c_1' \\ c_2\approx c_2'}
  {c_1;c_2 \approx c_1';c_2'}
\and
\inferrule[]
  { c_1 \approx c_1' \\ c_2 \approx c_2'}
  {\ifcmd{e}{c_1}{c_2}\approx \ifcmd{e}{c_1'}{c_2'}}
\and
\inferrule[]
  {c\approx c'}
  {\while{e}{c}\approx \while{e}{c'}}
\and
\inferrule[] 
  { \sevent{s_0}\dots\sevent{s_k} = \sevent{s_0'}\dots\sevent{s_k'} }
  { x:=\relabel{e}{\lab_f}{L_t}{\sevent{s_0}\dots\sevent{s_k}}\approx x:=\relabel{e}{\lab_f}{L_t}{\sevent{s_0'}\dots\sevent{s_k'}} } 
\and
\inferrule[]
  { \sevent{s} = \sevent{s'}}
  {\cod{EventOn}(\sevent{s})\approx \cod{EventOn}(\sevent{s'})}
\and
\inferrule[]
  { \sevent{s} = \sevent{s'}}
  {\cod{EventOff}(\sevent{s})\approx \cod{EventOff}(\sevent{s'})}
\end{mathpar}
\caption{Two commands are equivalent if they are equal modulo bracketed commands.}
\Description{}
\label{fig:equivalence_commands}
\end{figure*}

\subsubsection{Memory Well-Formedness}
Intuitively, labels that are more restrictive than $\fixedlab$ under the current event trace $\strace$ are considered to be sensitive. Hence, only those variables can hold bracketed values. This property is formalized below as Memory Well-Formedness:

\begin{definition}[Memory Well-Formedness]
\label{def:mem_well}
\[\wellformmem{\strace}{\Mem} \triangleq \forall x.~(\exists n.~\Mem(x)=[n])\Rightarrow \fixedlab\LEQ_\strace \Gamma(x)\]
\end{definition}

\begin{figure}
\begin{mathpar}
\mprset{flushleft} 
    
\inferrule[]
  { }
  {\configs{\Mem,[n]}\Downarrow_\strace [n]}
\and
\inferrule[]
  { B \not\LEQ_\strace \Gamma(x) 
  \lor \Mem(x)=[n]}
  {\configs{\Mem,x}\Downarrow_\strace \Mem(x)}
\and  
\inferrule[]
  { B \LEQ_\strace \Gamma(x) \land \Mem(x)=n  }
  {\configs{\Mem,x}\Downarrow_\strace [n]}
\and  
\inferrule[]
{\configs{\Mem,\expr_1}\Downarrow_\strace [n_1]\\\configs{\Mem,\expr_2}\Downarrow_\strace n_2\\n=n_1 \;\op\; n_2}
{\configs{\Mem,\expr_1\;op\;\expr_2}\Downarrow_\strace [n]}
\and  
\inferrule[]
  {\configs{\Mem,\expr_1}\Downarrow_\strace n_1\\\configs{\Mem,\expr_2}\Downarrow_\strace [n_2]\\n=n_1 \;\op \;n_2}
  {\configs{\Mem,e_1\;op\;\expr_2}\Downarrow_\strace [n]}
\and
\inferrule[] {\configs{n,\expr_1}\Downarrow_\strace [n_1]\\\configs{\Mem,\expr_2}\Downarrow_\strace [n_2]\\ n=n_1\; \op \; n_2}
  {\configs{\Mem,e_1\;op\;\expr_2}\Downarrow_\strace [n]}
\end{mathpar}
\caption{Augmented operational semantics for expressions.}
\Description{}
\label{fig:bracket_semantics_arith}
\end{figure}
\begin{figure*}
\begin{mathpar}
\mprset{flushleft} 
\inferrule[]
  { }
  {\configThree{\Mem}{[\Skip]}{\strace}\rightarrow \configThree{\Mem}{\Skip}{\strace}}
\and
\inferrule[]
  {\configThree{\Mem}{c}{\strace}\xrightarrow{} \configThree{\Mem'}{c'}{\strace'}}
  {\configThree{\Mem}{[c]}{\strace}\xrightarrow{} \configThree{\Mem'}{[c']}{\strace'}}
\and
\inferrule[]
  {\configs{\Mem,\expr}\Downarrow [\true]}
  {\configThree{\Mem}{\ifcmd{\expr}{c_1}{c_2}}{\strace} \rightarrow \configThree{\Mem}{[c_1]}{\strace}}
\and
 \inferrule[]
  {\configs{\Mem,\expr}\Downarrow [\false]}
  {\configThree{\Mem}{\ifcmd{\expr}{c_1}{c_2}}{\strace}\rightarrow \configThree{\Mem}{[c_2]}{\strace}}
\inferrule[OS-Assign-b1] 
  { \configs{\Mem,\expr} \Downarrow n \\ \fixedlab \LEQ_{\strace}\Gamma (x)}
  { \configThree{\Mem}{x:=\expr}{\strace} \rightarrow \configThree{\Mem[x:=[n]]}{\Skip}{\strace} }
\and
\inferrule[OS-Assign-b2] 
  { \configs{\Mem, \expr} \Downarrow [n] }
  {  \configThree{\Mem}{x:=\expr}{\strace} \rightarrow \configThree{m[x:=[n]]}{\Skip}{\strace} }
\and
\inferrule[OS-Assign-b3] 
  { \configs{\Mem, \expr} \Downarrow n \\ \fixedlab \not\LEQ_{\strace}\Gamma (x)}
  {  \configThree{\Mem}{x:=e}{\strace} \rightarrow \configThree{\Mem[x:=n]}{\Skip}{\strace} }
\and
\inferrule[OS-Relabel-r1] 
    {\configs{\Mem,\expr} \Downarrow n \\ \fixedlab \LEQ_{\strace}\Gamma (x)\\ \forall i \in 0...k.~\cod{eval}(\sevent{s_i},\sigma)=\true}
    { \configThree{\Mem}{x:=\relabel{\expr}{\lab_f}{\level_t}{\sevent{s_0}\dots\sevent{s_k}}}{\strace}\rightarrow \configThree{\Mem[x:=[n]]}{\Skip}{\strace}}
\and
\inferrule[OS-Relabel-r2] 
    {\configs{\Mem,\expr} \Downarrow [n] \\ \forall i \in 0...k.~\cod{eval}(\sevent{s_i},\sigma)=\true}
    { \configThree{\Mem}{x:=\relabel{\expr}{\lab_f}{\level_t}{\sevent{s_0}\dots\sevent{s_k}}}{\strace}\rightarrow \configThree{\Mem[x:=[n]]}{\Skip}{\strace}}
\and
\inferrule[OS-Relabel-r3] 
    {\configs{\Mem,\expr} \Downarrow n \\ \fixedlab\not\LEQ_{\strace}\Gamma (x)\\ \forall i \in 0...k.~\cod{eval}(\sevent{s_i},\sigma)=\true}
    { \configThree{\Mem}{x:=\relabel{\expr}{\lab_f}{\level_t}{\sevent{s_0}\dots\sevent{s_k}}}{\strace}\rightarrow \configThree{\Mem[x:=n]}{\Skip}{\strace}}
\and
\inferrule[OS-Output]
{\configTwo{\Mem}{e}\evalto v \quad
\inferrule{\forall i \in 0...k.~\\\\
\cod{eval}(\sevent{s_i},\sigma)=\true}{} \quad
\strace'=
\text{
$\begin{cases}
\strace \cdot \sevent{o_x@\level}, & \text{$\expr$ is a variable $x$} \\
\strace, & \text{otherwise}
\end{cases}$
}}
{\configThree{\Mem}{\outcmd{\level}{\expr}{\sevent{s_0}\dots\sevent{s_k}}}{\strace} \xrightarrow[]{\configThree{\level}{v}{\strace}}
\configThree{\Mem}{\Skip}{\strace'} }

\end{mathpar}
\caption{Augmented operational semantics for commands.}
\Description{}
\label{fig:bracket_assign_semantics}
\end{figure*}

\subsubsection{Augmented operational semantics}


The operational semantics is extended to model the bracketed syntax, as shown in Figures~\ref{fig:bracket_semantics_arith} and~\ref{fig:bracket_assign_semantics}. The rules here are extended for operational semantics with brackets, nothing special on the computational perspective. The interesting part only happens in assignment and relabel commands, where three cases are distinguished depending on whether $x$ and $\expr$ has a bracket or not. For assignments, there are three rules. \ruleref{OS-Assign-B1} makes sure when a variable $x$ satisfies $\fixedlab \LEQ_{\strace} \Gamma(x)$, it is assigned to a bracketed value. \ruleref{OS-Assign-B2} is trivial. \ruleref{OS-Assign-B3} states that when $B \not\LEQ_{\strace} \Gamma(x)$, the value $n$ is not bracketed. The augmented rules for relabel commands are similar.

\subsubsection{Extra typing rules}
Since the static semantics are augmented, under the definition of memory well-formedness, we also need to provide a type (i.e., label) for bracketed values and commands. We define extra typing rules for~\ruleref{AugNum} and~\ruleref{AugComm} in Figure~\ref{fig:bracketassign_type} for that purpose. These rules ensure that the type system treats the label of bracketed expressions as sensitive (i.e., higher than the distinguished $\fixedlab$). Bracketed commands are treated as commands that can be type-checked under a sensitive $\pc$ label. 

\begin{figure}
\begin{mathpar}
\mprset{flushleft} 

\inferrule[AugNum]
  {}
  {\Gamma\vdash [n]: \fixedlab}
\and
\inferrule[AugComm]
  { \pc, \Gamma\vdash c \\ \fixedlab, \Gamma\vdash c}
  {\pc, \Gamma\vdash [c]} 
\end{mathpar}
\caption{Augmented typing rules for bracketed syntax.}
\Description{}
\label{fig:bracketassign_type}
\end{figure}

\subsubsection{Completeness}

Completeness means that every step in the augmented semantics (with brackets) can be performed in the non-augmented semantics. Therefore, any full evaluation of a program in the original language can be simulated in the augmented language. With an augmented program $c$ and memory $m$, $\lfloor c\rfloor$ denotes removing all the brackets from $c$ and $\lfloor m\rfloor$ as the converted memory. Completeness can be expressed with the following lemma.

    \begin{lemma}
Completeness of the augmented language
	\label{lem:completeness}
	\[
	\Gamma, \pc \vdash c \AND \configThree{\lfloor c \rfloor}{\lfloor \Mem
		\rfloor}{\strace} \rightarrow^*
	\configThree{\Skip}{\Mem'}{\strace'} \rightarrow \exists \Mem''.\
        \configThree{c}{\Mem}{\strace}\rightarrow^* 
	\configThree{\Skip}{\Mem''}{\strace}~\AND~ \Mem'=\lfloor \Mem'' \rfloor
	\]
\end{lemma}
\begin{proof}
    By induction on each evaluation step.
\end{proof}

\subsection{Useful Lemmas}
\label{sec:lemmas}
We first show that for any well-formed memory, the label of any expression $e$ that evaluates to a bracketed value $[n]$ must always be sensitive (i.e., $\fixedlab \LEQ_\strace \lab$ where $\lab$ is the label of $\expr$).
\newtheorem*{lemma*}{Lemma}
\begin{lemma}
\label{lem:highexpr}
Expression evaluation
    $$ \forall \Mem, \strace, e, \lab,\Gamma.~ \wellformmem{\strace}{\Mem} \land \Gamma\vdash e: \lab \land \exists n. \configs{\Mem,\expr}\Downarrow [n] \Rightarrow \ \fixedlab \LEQ_\strace \lab 
    $$
\end{lemma}

\begin{proof}
By structural induction on $e$:
\begin{itemize}
    \item case $\expr=n$: vacuously true since $n$ evaluates to a value without brackets.

    \item case $\expr=[n]$: by the typing rules~\ruleref{AugNum} and \ruleref{SubType}, we have $\vdash \fixedlab \flowsto \lab$. Hence, we have $\fixedlab \LEQ_\strace \lab$ by Theorem~\ref{theorem:setsound}.
    
    \item case $\expr=x$: Due to the semantics, we have $\Mem(x)=[n]$.   
    Since $\Mem$ is well-formed, we have $\fixedlab \LEQ_\strace \Gamma(x)$. From Rules~\ruleref{Var}, \ruleref{SubType} and the assumption that $\Gamma\vdash x:\lab$, we have $\vdash \Gamma(x)\flowsto \lab$, which implies $\Gamma(x) \LEQ_\strace \lab$ by Theorem~\ref{theorem:setsound}. Hence, $B \LEQ_\strace \Gamma(x) \LEQ_\strace \lab$.
    
    \item case $\expr=\expr_1\:\op\:\expr_2$: From the assumption of Rule~\ruleref{Op}, we have $\vdash \expr_1:\lab_1$, $\vdash \expr_2:\lab_2$ and $ \vdash \lab_1\flowsto \lab_2$. By evaluation rules, at least one of $\expr_1$ and $\expr_2$ is evaluated to some $[\Mem]$. By the induction hypothesis, at least one of $\fixedlab \LEQ_\strace p_1$ and $\fixedlab \LEQ_\strace p_2$ is correct. The harder case is when $\fixedlab \LEQ_\strace p_1$. Since $ \vdash \lab_1\flowsto \lab_2$, we have $\lab_1 \LEQ_\strace \lab_2$ by Theorem~\ref{theorem:setsound}. Hence, $\fixedlab\LEQ_\strace p_2$ 
\end{itemize}
\end{proof}

We next show that if a program is well-typed under a program counter label $\pc$, then it is also well-typed under any $\pc'$ such that $\vdash \pc'\flowsto \pc$.
\begin{lemma}
PC subsumption
\label{lemma:pcsump}
$$\forall \pc, \pc', \Gamma, c.~\pc, \Gamma\vdash c \land \vdash \pc' \flowsto \pc \Rightarrow \pc',\Gamma\vdash c$$
\end{lemma}
\begin{proof}

By induction on the typing derivation $\pc, \Gamma \vdash c$:
\begin{itemize}
  \item case $pc,\Gamma\vdash \Skip$: Trivial.
  
  \item case $pc,\Gamma\vdash x:=e$: 
  From the typing rule~\ruleref{Assign}, we have $\vdash pc \flowsto \Gamma(x)$ and $\Gamma\vdash \expr: \Gamma(x)$. Since $\vdash \pc'\flowsto \pc$, we have $\vdash \pc'\flowsto \Gamma(x)$ by transitivity. So $\pc',\Gamma\vdash x:=e$.
  
  \item case $\pc,\Gamma\vdash c_1;c_2$:
  From typing rules, $\pc,\Gamma \vdash c_1$ and $\pc,\Gamma \vdash c_2$. With the induction hypothesis, we have $\pc',\Gamma \vdash c_1$ and $\pc',\Gamma \vdash c_2$. Hence, $\pc', \Gamma\vdash c_1;c_2$.
  
   \item case $\pc,\Gamma\vdash{\ifcmd{e}{c_1}{c_2}}$:
   From typing rules, we have $\Gamma\vdash e:\pc_1$, $\pc_1,\Gamma\vdash c_1$, $\pc_1, \Gamma\vdash c_2$, and $\vdash \pc\flowsto \pc_1$ for some $\pc_1$. With the relation $\vdash \pc'\flowsto \pc$, we know that $\vdash  \pc'\flowsto \pc_1$ still holds due to transitivity. Therefore, $\pc',\Gamma\vdash \ifcmd{e}{c_1}{c_2}$ stands under the same $\pc_1$.
   
    \item case $\pc,\Gamma\vdash \while{e}{c}$: From typing rules, we have $\Gamma\vdash e: \pc_1$, $pc_1, \Gamma\vdash c$ and $\vdash \pc \flowsto \pc_1$ for some $\pc_1$. With the relation $\vdash \pc'\flowsto \pc$, the condition $\vdash \pc'\flowsto \pc_1$ still holds due to transitivity. Therefore, $\pc',\Gamma\vdash\while{e}{c}$ stands under $\pc_1$.
    
    \item case $\pc, \Gamma\vdash x:=\relabelusing{e}{\lab_f}{L_t}{\sevent{s_0}\dots\sevent{s_k}}$:
   From the typing rule~\ruleref{Relabel} and the assumption, we have $\vdash \pc\flowsto\Gamma(x)$, $\Gamma\vdash e:\lab_f$, $\sevent{s_0}\dots\sevent{s_k} \vdash \level_t\flowsto \Gamma(x)$, and $\sevent{s_0}\dots\sevent{s_k}\vdash \lab_f\flowsto L_t$. Moreover, by $\vdash \pc'\flowsto \pc$, we have $\vdash \pc'\flowsto \Gamma(x)$. The other conditions do not change under $\pc'$. Therefore, $\pc', \Gamma\vdash x:=\relabelusing{e}{\lab_f}{L_t}{\sevent{s_0}\dots\sevent{s_k}}$.
    
  \item case $\pc, \Gamma \vdash \outcmd{\level}{\expr}{\sevent{s_0}\dots\sevent{s_k}}$: From typing rules, we have $\Gamma\vdash e: \lab$, $\sevent{s_0}\dots\sevent{s_k}\vdash \lab \releaseto \level$ and $\vdash \pc\flowsto \level$. By $\vdash \pc'\flowsto \pc$, we have $\vdash \pc'\flowsto \level$. Therefore, $\pc',\Gamma\vdash\outcmd{\level}{e}{\sevent{s_0}\dots\sevent{s_k}}$ stands.
    
  \item case $\pc, \Gamma \vdash \outcmd{\level}{x}{\sevent{o}_\x, \sevent{s_1}\dots\sevent{s_k}}$: same as the case above.

  \item case $\pc, \Gamma \vdash \cod{eventon}(\sevent{s})$: From the typing rules, we have $\vdash \pc\flowsto\bot$. By $\vdash\pc'\flowsto\pc$, we have $\vdash \pc'\flowsto \bot$. Therefore, the rule stands.
  
  \item case $\pc, \Gamma \vdash \cod{eventoff}(\sevent{s})$: Same as the case above.
    
  \item case $\pc, \Gamma\vdash [c]$: From the typing rule, we have $\pc, \Gamma\vdash c$ and $\fixedlab, \Gamma\vdash c$. By the induction hypothesis and the assumption that $\vdash \pc'\flowsto \pc$, we have $\pc', \Gamma\vdash c$. Therefore, $\pc', \Gamma\vdash [c]$ stands.

\end{itemize}
\end{proof}


Next, we show that for any expression that evaluates to a value under some memory $m_1$, it must evaluate to an equivalent value under any $m_2$ that is equivalent to $m_1$, under the equivalence relation in Figure~\ref{fig:equivalence_commands}.

\begin{lemma}
Expression preserves equivalence
\label{lemma:eppe}
$$\forall\Mem_1,\Mem_2,\strace, \expr, v_1, v_2.~ \Mem_1\approx_\strace \Mem_2\land \configs{\Mem_1,\expr}\Downarrow v_1\Rightarrow \exists v_2. \configs{\Mem_2,e}\Downarrow v_2 \land v_1\approx v_2$$
\end{lemma}

\begin{proof} By induction on the structure of expression $e$.
\begin{itemize}
    \item case $\expr=n$. We have $\configs{\Mem,e}\Downarrow n$ for both $\Mem_1$ and $\Mem_2$. So $v_1=v_2=n$.

    \item case $\expr=[n]$. Since we have $\configs{\Mem,\expr}\Downarrow [n]$ for both $\Mem_1$ and $\Mem_2$, so $v_1=v_2=[n]$.
    
    \item case $\expr=x$. When $\fixedlab \not\LEQ_\strace \Gamma(x)$, we have $v_1=\Mem_1(x)\approx \Mem_2(x)=v_2$ from the assumption that $\Mem_1\approx_\strace \Mem_2$. Otherwise, we have $v_1=[n_1]\approx [n_2]=v_2$ for some $n_1$ and $n_2$ from the semantics.

    \item case $\expr = \expr_1$ op $\expr_2$. By the semantics, $\configs{\Mem_1,\expr_1}\Downarrow v_3$, $\configs{\Mem_2,\expr_1}\Downarrow v_3'$,
    $\configs{\Mem_1,\expr_2}\Downarrow v_4$, $\configs{\Mem_2,\expr_2}\Downarrow v_4'$. 
    \begin{itemize}
        \item If $v_3$, $v_4$ have no brackets, $v_3= v_3'$ and $v_4= v_4'$ by the induction hypothesis. Further, $v_1= v_2$.
        
        \item When one of $v_3$ or $v_4$ is a bracketed value, $v_3'$ or $v_4'$ are bracketed value by the induction hypothesis. From the semantics, we know $e$ has to be evaluated to bracketed values under both memories $\Mem_1$ and $\Mem_2$. So $v_1\approx v_2$.
\end{itemize}
\end{itemize}
\end{proof}

Next, we show that the partial ordering on dynamic release labels is preserved over any extension of the current event. Intuitively, this property is important to show that after any changes to the event trace, with $\cod{eventon}$ and $\cod{eventoff}$ commands, the ``protected'' values with brackets are still stored in ``protected'' memory locations (Lemma~\ref{lemma:preservation}, Preservation).

\begin{lemma*} Trace Preservation, Lemma~\ref{lemma:trace_preservation}
$$ \forall \lab,\lab',\strace.~\lab \LEQ_\strace\lab' \Rightarrow (\forall \strace'.~\strace\preceq\strace'\Rightarrow  \lab\LEQ_{\strace'}\lab')$$
\end{lemma*}
\begin{proof}
Consider any extension of $\strace$, denoted as $\strace'$. By Definition~\ref{def:partial_dynamic}, it is sufficient to show that $\forall \strace''.~\strace'\preceq\strace''\Rightarrow \auxfuncold{\lab}{\strace''} \LEQ \auxfuncold{\lab'}{\strace''}$. Note that any such $\strace''$ is also an extension of $\strace$. Hence, $\auxfuncold{\lab}{\strace''} \LEQ \auxfuncold{\lab'}{\strace''}$ by the definition of $\lab \LEQ_\strace\lab'$.

\end{proof}

\subsection{Preservation}

We prove that program evaluation under the augmented semantics preserves typing and well-formedness during evaluation.

\begin{lemma}
Preservation
$$\forall c, c', \strace, \strace', \Mem, \Mem', \pc, \Gamma.~\wellformmem{\strace}{\Mem} \land \pc, \Gamma\vdash c \land \configs{\Mem,c,\strace} \rightarrow \configs{\Mem',c',\strace'}\Rightarrow \wellformmem{\strace'}{\Mem'} \land \pc,\Gamma\vdash c'$$
\label{lemma:preservation}
\end{lemma}

\begin{proof}
By induction on evaluation rules of $\configs{\Mem,c,\strace} \rightarrow \configs{\Mem',c', \strace'}$.
\begin{itemize}
    \item case $\configThree{\Mem}{[\Skip]}{\strace}\rightarrow\configThree{\Mem}{\Skip}{\strace}$: Trivial, since $\Mem'$ and $\strace'$ remain the same and $\Skip$ type checks under any $\Gamma, \pc$.
    
    \item case $\configThree{\Mem}{[c]}{\strace}\rightarrow \configThree{\Mem'}{[c']}{\strace'}$:
    From the evaluation rule, we have $\configThree{\Mem}{c}{\strace}\rightarrow\configThree{\Mem'}{c'}{\strace'}$. 
    From the typing rule~\ruleref{AugComm}, we know $\pc',\Gamma\vdash c$ and $\fixedlab,\Gamma\vdash c$. By $\pc',\Gamma\vdash c$ and the assumption that $\wellformmem{\strace}{\Mem}$, we have $\pc',\Gamma\vdash c'$ and $\wellformmem{\strace'}{\Mem'}$
    by the induction hypothesis. Similarly, from $\fixedlab,\Gamma\vdash c$, we also have $\fixedlab,\Gamma\vdash c'$. Hence, we can derive $\pc,\Gamma\vdash [c']$. 
  
    \item case $\configThree{\Mem}{x:=e}{\strace} \rightarrow \configThree{\Mem[x:=n]}{\Skip}{\strace} $:
    We have $\pc,\Gamma \vdash \Skip$ by typing rule~\ruleref{$\Skip$}. From the evaluation rule of $\configThree{\Mem}{x:=e}{\strace} \rightarrow \configThree{\Mem[x:=n]}{\Skip}{\strace}$, we know that $\fixedlab \not\LEQ_\strace \Gamma(x)$, and $\configs{\Mem,\expr}\Downarrow n$.  
    With $\wellformmem{\strace}{\Mem}$, we know $\Mem'$ is still well-formed as the command sets $x$ to $n$ without brackets.
  
    \item case $\configThree{\Mem}{x:=e}{\strace} \rightarrow \configThree{\Mem[x:=[n]]}{\Skip}{\strace} $:  We have $pc,\Gamma \vdash \Skip$ by typing rule~\ruleref{Skip}. Next, we show that $\Mem[x:=[n]]$ is well-formed. From typing rule~\ruleref{Assign}, we have $\Gamma \vdash e:\Gamma(x)$. From the evaluation rules with $\configThree{\Mem} {x:=e}{\strace} \rightarrow \configThree{\Mem[x:=[n]]}{\Skip}{\strace} $, there are two possibilities:
    \begin{itemize}
      \item Rule~\ruleref{OS-Assign-b1}: $\Mem[x:=[n]]$ is still well-formed under trace $\strace$ since we have $\fixedlab \LEQ_\strace \Gamma(x)$ in the rule assumption. Hence, $x$ can hold bracketed values. For other variables, we still have $\wellformmem{\strace}{\Mem'}$ as the assignment does not change any values other than $x$'s.

      \item Rule~\ruleref{OS-Assign-b2}: by rule assumption, $\configs{\Mem, \expr} \Downarrow [n]$ for some $n$.
      From the typing rule, we have $\Gamma\vdash e: \Gamma(x)$.
      So by Lemma~\ref{lem:highexpr}, we have 
      $\fixedlab \LEQ_\strace \Gamma(x)$.
      Hence, $x$ can hold bracketed values. For other variables, we still have $\wellformmem{\strace}{\Mem'}$ as the assignment does not change any values other than $x$'s.

  \end{itemize}
  
  \item case $\configThree{\Mem}{c_1;c_2}{\strace}\rightarrow\configThree{\Mem'}{c_1';c_2}{\strace'}$:
  From the evaluation rule we have $\configThree{\Mem}{c_1}{\strace}\rightarrow\configThree{\Mem'}{c_1'}{\strace'}$. From the typing rule~\ruleref{Seq}, we have $pc, \Gamma\vdash c_1$ and $\pc, \Gamma\vdash c_2$. By the induction hypothesis, we get $\wellformmem{\Mem'}{\strace'}$ and $\pc,\Gamma\vdash c_1'$. Therefore we can derive $\pc,\Gamma\vdash c_1';c_2$.
  
  \item case $\configThree{\Mem}{\Skip;c}{\strace}\rightarrow\configThree{\Mem}{c}{\strace}$: From typing rule~\ruleref{Seq}, we have $\pc, \Gamma \vdash c$. We can also derive that $\wellformmem{\Mem}{\strace}$ as both $\Mem$ and $\strace$ do not change.
  
  \item case $\configThree{\Mem}{{\ifcmd{\expr}{c_1}{c_2}}}{\strace}\rightarrow \configThree{\Mem}{c_1}{\strace}$:
  From the evaluation rule, we know that $\configs{\Mem,\expr}\Downarrow \true$.
  From the typing rule~\ruleref{If}, we have $\Gamma\vdash e:\lab_e$, $\lab_e,\Gamma\vdash c_1$, and $\vdash pc\flowsto \lab_e$ for some $\lab_e$. With $\pc\flowsto \lab_e$, $\lab_e,\Gamma\vdash c_1$ and Lemma~\ref{lemma:pcsump}, we have $\pc,\Gamma\vdash c_1$. Therefore, the type of the true branch is preserved. Moreover, $\Mem$ is still well-formed as both $\Mem$ and $\strace$ do not change.
  
  \item case $\configThree{\Mem}{{\ifcmd{\expr}{c_1}{c_2}}}{\strace}\rightarrow \configThree{\Mem}{c_2}{\strace}$: similar to the case above.
  
  \item case $\configThree{\Mem}{{\ifcmd{\expr}{c_1}{c_2}}}{\strace}\rightarrow \configThree{\Mem}{[c_1]}{\strace}$:
  From the evaluation rule, we know that $\configs{\Mem,\expr}\Downarrow [\true]$. 
  From the typing rule~\ruleref{If}, we have $\Gamma\vdash e:\lab_e$, $\lab_e,\Gamma\vdash c_1$, and $\vdash pc\flowsto \lab_e$ for some $\lab_e$.

  Since $\expr$ evaluates to a value with brackets, we have $\vdash \fixedlab \flowsto \lab_e$ according to~Lemma~\ref{lem:highexpr}. By Lemma~\ref{lemma:pcsump}, we have $\pc, \Gamma \vdash c_1$ and $\fixedlab, \Gamma \vdash c_1$.  
  Hence, we can type check $[c_1]$ under $\pc$ and $\Gamma$ as follows:
\[
\inferrule[]
  { \pc, \Gamma\vdash c_1 \\ \fixedlab, \Gamma\vdash c_1}
  {\pc, \Gamma\vdash [c_1]} 
\]
  
  $\Mem$ is still well-formed as both $\Mem$ and $\strace$ do not change.
  
  \item case $\configThree{\Mem}{{\ifcmd{\expr}{\command_1}{\command_2}}}{\strace}\rightarrow \configThree{\Mem}{[\command_2]}{\strace}$: similar to the case above.

  \item case $\configThree{\Mem}{\while{\expr}{\command}}{\strace}\rightarrow \configThree{\Mem}{{\ifcmd{e}{\command;\;(\while{\expr}{\command)}}{\Skip}}}{\strace}$:\\
  From the typing rule~\ruleref{While}, we have three assumptions, $A=\Gamma\vdash \expr:\lab_\expr$, {$B=\pc\flowsto\lab_\expr$} and $C=\lab_\expr, \Gamma\vdash \command$ for some $\lab_\expr$. 
  Therefore, the program after evaluation on the right hand side must look like this:
\[
\inferrule* 
  {\inferrule* {A}
    {\Gamma\vdash \expr: \lab_\expr}
    \\{\inferrule* {B}
    {\vdash \pc\flowsto \lab_\expr}}\\   {\inferrule* {{\inferrule* {C}
    {\lab_\expr, \Gamma\vdash c}}\\{\inferrule* {{\inferrule* {A}
    {\Gamma\vdash \expr: \lab_\expr}} \\{\inferrule* {B}
    {\vdash \pc \flowsto \lab_\expr}} \\{\inferrule* {C}
    {\lab_\expr, \Gamma\vdash c}}}
    {\lab_\expr,\Gamma\vdash \while{\expr}{c}}
    }}
    {\lab_e,\Gamma\vdash c; \while{\expr}{c}}
    }\\{\inferrule* { }
    {\lab_e,\Gamma\vdash \Skip}
    }
    }
    {\pc, \Gamma\vdash\ifcmd{\expr}{\command;\;(\while{\expr}{\command})}{\Skip} }
\]

 $\Mem$ is still well-formed as both $\Mem$ and $\strace$ do not change.

 \item case: $\configThree{\Mem}{\cod{eventon}(\sevent{s})}{\strace}\rightarrow \configThree{\Mem} {\Skip}{\strace\cdot\sevent{s}}$:  We have $\pc,\Gamma \vdash \Skip$ by typing rule~\ruleref{$\Skip$}.
 
 For each variable $x$, we have $\Mem'(x)=\Mem(x)$ as the command only changes security event status. Hence, 
 \begin{itemize}
     \item If $\exists~n.~\Mem'(x)=[n]$, we know that $\Mem(x)=[n]$ as well. Since $\wellformmem{\strace}{\Mem}$, we know that $\fixedlab \LEQ_{\strace}\Gamma (x)$.   By Lemma~\ref{lemma:trace_preservation} (Trace Preservation), we have $\fixedlab \LEQ_{\strace\cdot\sevent{s}}\Gamma (x)$. Hence, $\wellformmem{\strace \cdot \sevent{s}}{\Mem'}$.

    \item If $\exists~n.~\Mem'(x)=n$, we know that it can be hold in any variable according to Definition~\ref{def:mem_well} .
 \end{itemize}

 \item case $\configThree{\Mem}{\cod{eventoff}(\sevent{s})}{\strace}\rightarrow \configThree{\Mem} {\Skip}{\strace\cdot\neg\sevent{s}}$: Same as the \cod{eventon} case.
 
 \item case $\configThree{\Mem}{ x:=\relabelusing{e}{\lab_f}{\level_t}{\sevent{s_0}\dots\sevent{s_k}}}{\strace}\rightarrow \configThree{\Mem[x:=n]}{\Skip}{\strace}$:
  We have $pc,\Gamma \vdash \Skip$ by typing rule~\ruleref{$\Skip$}. Since $\strace$ does not change and only the value of $x$ changes in $\Mem$, $\Mem$ is still well-formed under $\strace$ for variables other than $x$. For $x$, we know that $n$ can be hold in any variable according to Definition~\ref{def:mem_well} .
  
  \item case $\configThree{\Mem}{x:=\relabelusing{e}{\lab_f}{\level_t}{\sevent{s_0}\dots\sevent{s_k}}}{\strace}\rightarrow \configThree{\Mem[x:=[n]]}{\Skip}{\strace}$: We have $pc,\Gamma \vdash \Skip$ by typing rule~\ruleref{$\Skip$}. Next, we show that $\wellformmem{\strace}{\Mem[x:=[n]]}$.
    {\begin{itemize}
      \item Rule~\ruleref{OS-Relabel-r1}: From the evaluation rule, we have that $\fixedlab \LEQ_{\strace}\Gamma (x)$. Hence, the variable $x$ can hold a bracketed value: $\Mem[x:=[n]]$ is still well-formed.
      
      \item Rule~\ruleref{OS-Relabel-r2}: From the evaluation rule, we have that $\configs{\Mem,e} \Downarrow [n]$ and $\sevent{s_0}\dots\sevent{s_k}$ holds on $\strace$. From Lemma~\ref{lem:highexpr}, we have $\vdash \fixedlab \flowsto \lab_f$, where $\Gamma \vdash e: \lab_f$, which implies $\fixedlab \LEQ_\strace \lab_f$ by Theorem~\ref{theorem:setsound}. 
      From the typing rule, we know that $\sevent{s_0}\dots\sevent{s_k} \vdash \level_t\flowsto \Gamma(x)$ and $\sevent{s_0}\dots\sevent{s_k}\vdash \lab_f \flowsto \level_t$. By Theorem~\ref{theorem:setsound}, we have $\level_t \LEQ_\strace \Gamma(x)$ and $\lab_f \LEQ_\strace \level_t$. Hence, we have $\fixedlab \LEQ_\strace \lab_f \LEQ_\strace \level_t \LEQ_\strace \Gamma(x)$ due to transitivity. So memory $\Mem'$ stays well-formed after setting $x$, whose label is at least as restrictive as $\fixedlab$, to $[n]$.
  \end{itemize}}
  
  \item case $\configThree{\Mem}{x:=\relabelusing{\expr}{\lab_f}{\level_t}{\sevent{s_0}\dots\sevent{s_k}}}{\strace}\rightarrow \configThree{\Mem}{\Skip}{\strace}$: This is a trivial case as $\Mem$ and $\strace$ remain the same and $\Skip$ can be type-checked under any environment.
  
  \item case ${\configThree{\Mem}{\outcmd{\level'}{e}{\sevent{s_0}\dots\sevent{s_k}}}{\strace} \xrightarrow[]{\configThree{\level'}{v_1}{\strace}}
\configThree{\Mem}{\Skip}{\strace'} }$: same to the case above when $\expr$ is not a variable $x$. Otherwise, we have $\strace'=\strace \cdot \sevent{o_x@\level}$. By Lemma~\ref{lemma:trace_preservation} (Trace Preservation), we know for any $x$ such that $\fixedlab\LEQ_\strace \Gamma(x)$, it must be true that $\fixedlab \LEQ_{\strace\cdot \sevent{o_x@\level}}\Gamma (x)$. Hence, they can still hold bracketed values. For other variables, they previously hold non-bracketed values due to the well-formed-memory assumption. Those values can be stored in any memory location. Therefore, $\wellformmem{\strace \cdot \sevent{s}}{\Mem'}$.

  \item case ${\configThree{\Mem}{\outcmd{\level'}{\expr}{\sevent{s_0}\dots\sevent{s_k}}}{\strace} \xrightarrow[]{}
\configThree{\Mem}{\Skip}{\strace} }$: same to the case above.
\end{itemize}

\end{proof}

\subsection{High Step}
Next, we show that whenever the sources is considered as secret ($\auxfuncold{\fixedlab}{\strace}\not\LEQ\fixedlevel$), a step of well-typed bracketed command (1) cannot leak any sensitive information to the output channel with attacker's level $\fixedlevel$, and (2) cannot emit new security events.

\begin{lemma}[High Step]
\label{lemma:high_main}
\begin{multline*}
 \forall\Mem, c, x, \pc, \strace, \Gamma.~\pc, \Gamma \vdash [c] \land 
\auxfuncold{\fixedlab}{\strace}\not\LEQ\fixedlevel \land
 \configThree{\Mem}{[c]}{\strace}\xrightarrow{\tau}\configThree{\Mem'}{c'}{\strace'} 
\Rightarrow
\Mem\approx_{\strace'} \Mem'\land |\proj{\trace}_\fixedlevel|=0 \land \strace=\strace'
\end{multline*}
\end{lemma}
\begin{proof}

From the typing rule~\ruleref{AugComm}, we have $\fixedlab, \Gamma \vdash c$. We proceed by induction on the structure of $c$.
\begin{itemize}
  \item cases where $c$ is $(\Skip)$, $({\ifcmd{\expr}{c_1}{c_2}})$, $(\while{\expr}{c})$,~\ruleref{Output-Fail} and~\ruleref{Relabel-Fail} are trivial
  since those commands do no modify memory and execution trace.
  
  \item case where $c$ is $\assign{x}{\expr}$: by augmented bracket semantics for commands, we have that $\configThree{\Mem}{\assign{x}{\expr}}{\strace} \xrightarrow{\emptyset} \configThree{\Mem'}{\Skip}{\strace}$. Hence, it is easy to check that $|\proj{\trace}_\fixedlevel|=0 \land \strace=\strace'$.
  
  By the typing rule, we have that $\fixedlab, \Gamma \vdash \assign{x}{\expr}$. Via typing rule~\ruleref{Assign}, we have $\fixedlab \flowsto \Gamma(x)$. There are two cases for $\Mem'$:
   \begin{itemize}
       \item $\Mem[x:=[n]]$ via the evaluation rule~\ruleref{OS-Assign-B1} and~\ruleref{OS-Assign-B2}. Since $\fixedlab \flowsto \Gamma(x)$, we clearly have $\Mem\approx_\strace \Mem[x:=[n]]$ by definition. 
       
       \item $\Mem[x:=n]$ via the evaluation rule~\ruleref{OS-Assign-b3}. By that fact that $\fixedlab \flowsto \Gamma(x)$ and Theorem~\ref{theorem:setsound}, we can derive that $\fixedlab \LEQ_\strace \Gamma(x)$ which contradicts the assumption $\fixedlab \not\LEQ_\strace \Gamma(x)$ in Rule~\ruleref{OS-Assign-b3}. Hence, this case is not applicable.
   \end{itemize}
  
  \item case where $c$ is $c_1;c_2$: from evaluation rule, we have $\configThree{\Mem}{c_1}{\strace} \xrightarrow{\tau} \configThree{\Mem'}{ c_1'}{\strace'}$. By typing rule, we have that $\fixedlab, \Gamma\vdash c_1$ and $\fixedlab, \Gamma\vdash c_2$. So by typing rule~\ruleref{AugComm}, we have $\pc, \Gamma \vdash [c_1]$ and $\pc, \Gamma \vdash [c_2]$. Hence, the conclusion is correct by the induction hypothesis.

 \item case where $c$ is $\outcmd{\level'}{\expr}{\sevent{s_0}\dots\sevent{s_k}}$: there are two possible rules to type-check an output command: rule~\ruleref{Output} and rule~\ruleref{Output-Per}. In either case, we know that $\vdash\fixedlab \flowsto\level'$. By Theorem~\ref{theorem:setsound}, we have $\fixedlab \LEQ_\strace \level'$ which implies $\auxfuncold{\fixedlab}{\strace} \LEQ \level'$ by definition. 

 If $\level' \LEQ \fixedlevel$, where $\fixedlevel$ is the fixed attacker's level, we can derive $\auxfuncold{\fixedlab}{\strace} \LEQ \fixedlevel$, which contradicts the assumption that $\auxfuncold{\fixedlab}{\strace}\not\LEQ \fixedlevel$.
 
 Otherwise, $|\proj{\trace}_\fixedlevel|=0$ as the output is not visible to level $\fixedlevel$ (Definition~\ref{def:proj}). The other conclusions are trivial as the output command does not modify memory and execution trace.

  \item case where $c$ is $x:=\relabelusing{e}{\lab_f}{\level_t}{\sevent{s_0}\dots\sevent{s_k}}$: By typing rule~\ruleref{Relabel}, we know that  $\fixedlab \flowsto\Gamma(x)$.

  There are two cases for $m'$:
  \begin{itemize}
      \item $\Mem[x:=[n]]$ from evaluation rules~\ruleref{OS-Relabel-R1} and~\ruleref{OS-Relabel-R2}. Since $\fixedlab \flowsto \Gamma(x)$, we clearly have $\Mem\approx_\strace \Mem[x:=[n]]$ by definition. Moreover, we know that $\trace$ is empty since the command does not change outputs and security events.

      \item $\Mem[x:=n]$ from evaluation rule~\ruleref{OS-Relabel-R3}. By the fact that $\fixedlab \flowsto \Gamma(x)$ and Theorem~\ref{theorem:setsound}, we have $\fixedlab\LEQ_\strace\Gamma(x)$ which contradicts the assumption $\fixedlab \not\LEQ_\strace \Gamma(x)$ in Rule ~\ruleref{OS-Relabel-R3}. Hence, this case is not applicable.
  \end{itemize}
  
  \item case where $c$ is $\cod{eventon}(\sevent{s})$: by typing rule, we have $\fixedlab,\Gamma \vdash \cod{eventon}(\sevent{s})$ and therefore, $\fixedlab\flowsto \bot$. By Theorem~\ref{theorem:setsound}, we can derive that $\fixedlab \LEQ_\strace \bot$, which implies $\auxfuncold{\fixedlab}{\strace} \LEQ \bot\LEQ \fixedlevel$ by definition. This contradicts the assumption that $\auxfuncold{\fixedlab}{\strace}\not\LEQ \fixedlevel$. Hence, this rule is not applicable.
  
  \item case where $c$ is $\cod{eventoff}(\sevent{s})$: Similar as the case above.

\end{itemize}
\end{proof}

\subsection{Unwinding}
Finally, we need a lemma to show that equivalence and well-formedness is preserved in each step. That is, for each step taken by one of the two equivalent programs, either the other program takes a number of steps to reach the same program state, or the program diverges inside a bracket. Besides equivalence and well-formedness the trace of output channels and security events generated by the two memories is the same.

\begin{lemma*}
Unwinding, Lemma~\ref{lemma:unwinding_dynamicrelease}
\begin{multline*}
\forall \Mem_1, \Mem_2, c_1, c_2, \tau_1, \tau_2, \strace, x, \pc, \Gamma.~\wellformmem{\strace}{\Mem_1} \land \wellformmem{\strace}{\Mem_2} \land \Mem_1\approx_\strace \Mem_2 \land \\
(\command_1\neq \outcmd{\level'}{\x}{\sevent{o}_\x, \sevent{s_0\dots s_k}})
\land 
 \auxfuncold{\fixedlab}{\strace}\not\LEQ \fixedlevel \land
 c_1\approx c_2 \land 
\configThree{\Mem_1}{c_1}{\strace}\xrightarrow{\tau_1}\configThree{\Mem_1'}{c_1'}{\strace'}\Rightarrow \\
(\exists \Mem_2',c_2',\tau_2. \configThree{\Mem_2} {c_2}{\strace} \xrightarrow{\tau_2}^{*}\configThree{\Mem_2'}{c_2'}{\strace'}\land \Mem_1'\approx_{\strace'} \Mem_2'\land c_1'\approx c_2'\land \proj{\trace_1}_\fixedlevel=\proj{\trace_2}_\fixedlevel) 
\lor(\exists c.~c_2=[c] \text{ and } c \text{ diverges })
\end{multline*}
\end{lemma*}
\begin{proof}
By rule induction on the step $\configs{\Mem_1,c_1,\strace} \xrightarrow{\trace_1} \configs{\Mem_1',c_1',\strace'}$.

\begin{itemize}
  \item case $\configThree{\Mem_1}{[c_3]}{\strace}\xrightarrow{\tau_1} \configThree{\Mem_1'}{[c_3']}{\strace'}$: \\
  From $c_1\approx c_2$, we know that $c_2=[c_4]$ for some command $c_4$. By Lemma~\ref{lemma:high_main} (High Step), we have that $\Mem_1'\approx_\strace \Mem_1 \approx_\strace \Mem_2$, $|\proj{\traceout{\tau_1}}_\fixedlevel|=0$ and $\strace'=\strace$. Hence, we select $c_2'$ to be the same as $c_2$ and make 0 steps on it. It is easy to check that all desired conditions hold.

  \item case $\configThree{\Mem_1}{[\Skip]}{\strace}\rightarrow\configThree{\Mem_1}{\Skip}{\strace'}$: \\
  From $c_1\approx c_2$, we know that $c_2=[c_4]$ for some command $c_4$. If $c_4$ diverges under $m_2$, then let $c=c_4$, we have $(\exists c.~c_2=[c] \text{ and $c$ diverges})$.

  If $c_4$ terminates, we know that there must be an execution trace like 
  \[\configThree{\Mem_2}{[c_4]}{\strace}\xrightarrow{\trace_2} \configThree{\Mem_2'}{[c_4']}{\strace_2'}\xrightarrow{} \cdots \configThree{\Mem_2''}{\Skip}{\strace_2''}\]
  
  By Lemma~\ref{lemma:high_main} (High Step), we have that $\Mem_2'\approx \Mem_2 \approx \Mem_1$, $|\proj{\traceout{\tau_2}}_\fixedlevel|=0$ and $\strace_2'=\strace$. We can repeat the same on the rest of the execution trace and derive that $\strace_2''=\strace \land \Mem_2''\approx_{\strace_2''} \Mem_2 \approx_{\strace_2''} \Mem_1$ and the whole trace produces no output events on the projection of $\fixedlevel$. So the first condition in the conclusion holds.

  \item case $\configThree{\Mem_1}{x:=e}{\strace}\xrightarrow{\emptyset}\configThree{\Mem_1[x:=n]}{\Skip}{\strace}$:\\
  From $c_1\approx c_2$, we know $c_2$ is $x:=\expr$. 
  From the evaluation rule, we know $\configs{\Mem_1,\expr}\Downarrow n$. From Lemma~\ref{lemma:eppe}, we have that $\configs{\Mem_2,e}\Downarrow n$. Since $\Mem_1\approx_{\strace} \Mem_2$, we have $\Mem_1'=\Mem_1[x:=n]\approx_{\strace} \Mem_2[x:=n]=\Mem_2'$ regardless the label of $x$. $\Skip\approx \Skip$ is trivial. The results on $\trace_1$ and $\trace_2$ are also trivial as assignments does not modify the trace.
  
  \item case $\configThree{\Mem_1}{x:=e}{\strace}\xrightarrow{\emptyset}\configThree{\Mem_1[x:=[n]]}{\Skip}{\strace}$:\\  
    From $c_1\approx c_2$, we know that $c_2$ is $x:=\expr$. There are two cases:
    \begin{itemize}
        \item (Evaluation rule~\ruleref{OS-Assign-b1})  We have $\fixedlab \LEQ_{\strace}\Gamma(x)$ from the assumption. Hence, $\configThree{\Mem_2}{x:=e}{\strace}\xrightarrow{\emptyset}\configThree{\Mem_2[x:=[n']]}{\Skip}{\strace}$ for some $n'$. Since $\Mem_1\approx_\strace \Mem_2$, we have that $\Mem_2[x:=[n']]\approx_\strace \Mem_1[x:=[n]]$. $\Skip\approx \Skip$ is trivial. The results on $\trace_1$ and $\trace_2$ are also trivial as they remain unchanged.
        
        \item (Evaluation rule~\ruleref{OS-Assign-b2}) We have $\configs{\Mem_1,e}\Downarrow [n]$ in this case. From Lemma~\ref{lemma:eppe}, we have that $\configs{\Mem_2,e}\Downarrow [n']$ for some $n'$. Hence, $\configThree{\Mem_2}{x:=e}{\strace}\xrightarrow{\emptyset}\configThree{\Mem_2[x:=[n']]}{\Skip}{\strace}$ for some $n'$. Since $\Mem_1\approx_\strace \Mem_2$, we have $\Mem_1[x:=[n]]\approx_\strace \Mem_2[x:=[n']]$. $\Skip\approx \Skip$ is trivial. The results on $\trace_1$ and $\trace_2$ are also trivial as they remain unchanged.
    \end{itemize}
  
  \item case $\configThree{\Mem_1}{c_3;c_4}{\strace}\xrightarrow{\tau_1}\configThree{\Mem_1'}{c_3';c_4}{\strace_1}$: \\
    Command $c_1$ has the form ($c_3;c_4$), so command $c_2$ must have the same form ($c_5;c_6$) where $c_3\approx c_5$ and $c_4\approx c_6$. According to the evaluation rule, we know that $\configThree{\Mem_1}{c_3}{\strace}\xrightarrow{\tau_1}\configThree{\Mem_1'}{c_3'}{\strace_1}$. By the induction hypothesis, we know from $c_3\approx c_5$ that either $\configThree{\Mem_2}{c_5}{\strace}$ diverges, or $\configThree{\Mem_2}{c_5}{\strace}\xrightarrow{\tau_2}^{*}\configThree{\Mem_2'}{c_5'}{\strace_2}$ where $\strace_2=\strace_1$, $\Mem_1'\approx_{\strace_1} \Mem_2'$ and $\proj{\trace_1}_\fixedlevel=\proj{\trace_2}_\fixedlevel$. In either case, we select $c_2'$ as $c_5';c_6$. In the first case, $c_2'$ diverges. In the second case,
    from the evaluation rule, we have $\configThree{\Mem_2}{c_5;c_6}{\strace_1}\xrightarrow{\tau_2}^{*}\configThree{\Mem_2'}{c_5';c_6}{\strace_2'}$ where $\Mem_1'\approx_{\strace_1} \Mem_2'$ and  $\proj{\trace_1}_\fixedlevel=\proj{\trace_2}_\fixedlevel$, and $\strace_1 = \strace_2$ from the induction hypothesis.
    
  \item case $\configThree{\Mem_1}{{\ifcmd{e}{c_3}{c_4}}}{\strace}\xrightarrow{\emptyset} \configThree{\Mem_1}{c_3}{\strace}$:\\
  As $c_1\approx c_2$, we know that $c_2$ has the form of $\ifcmd{e}{c_5}{c_6}$ where $c_3\approx c_5$, and $c_4\approx c_6$. From the evaluation rule, we know that $\configs{\Mem_1,e}\Downarrow \true$. By the assumption that $m_1\approx_\strace m_2$ and Lemma~\ref{lemma:eppe} (expression preserves), we have $\configs{\Mem_2,e}\Downarrow \true$. 
  Hence, we have that $\configThree{\Mem_2}{\ifcmd{e}{c_5}{c_6}}{\strace}\xrightarrow{\emptyset} \configThree{\Mem_2}{c_5}{\strace}$.
  Since $c_3\approx c_5$ and both traces under $\Mem_1$ and $\Mem_2$ are empty, we know that by selecting $c_2'$ to be $c_5$, all conditions in the conclusion holds.
  
  \item case $\configThree{\Mem_1}{\ifcmd{e}{c_3}{c_4}}{\strace}\xrightarrow{\emptyset} \configThree{\Mem_1}{c_4}{\strace}$: Similar to the case above.
  
  \item case $\configThree{\Mem_1}{{\ifcmd{e}{c_3}{c_4}}}{\strace}\xrightarrow{\emptyset} \configThree{\Mem_1}{[c_3]}{\strace}$:\\
    As $c_1\approx c_2$, we know that $c_2$ has the form of $\ifcmd{e}{c_5}{c_6}$ where $c_3\approx c_5$, and $c_4\approx c_6$. From the evaluation rule, we know that $\configs{\Mem_1,e}\Downarrow [\true]$. By the assumption that $m_1\approx_\strace m_2$ and Lemma~\ref{lemma:eppe} (expression equivalence), it must be true that $\configs{\Mem_2,e}\Downarrow [b]$ for some boolean value $b$.
    Without losing generality, we consider the case where $\configs{\Mem_2,e}\Downarrow [\false]$ so that $\configThree{\Mem_2}{\ifcmd{e}{c_5}{c_6}}{\strace}\xrightarrow{\emptyset} \configThree{\Mem_2} {[c_6]}{\strace}$ by the evaluation rule. Let $c_2'=[c_6]\approx [c_3]$, we have $\Mem_2'=\Mem_2\approx_\strace \Mem_1=\Mem_1'$, and the conditions on $\trace_1$ and $\trace_2$ are trivially true as both $c_1$ and $c_2$ produce empty traces.
    
  \item case $\configThree{\Mem_1}{{\ifcmd{\expr}{c_3}{c_4}}}{\strace}\xrightarrow{\emptyset} \configThree{\Mem_1}{[c_4]}{{\strace}}$: Similar to the case above.
  
  \item case $\configThree{\Mem_1}{\while{e}{c_3}}{\strace}\xrightarrow{\emptyset} \configThree{\Mem_1}{{\ifcmd{(e}{c_3;\;\while{e}{c_3)}}{\Skip}}}{\strace}$:\\
   Since $c_1\approx c_2$, we know that $c_2=\while{\expr}{c_4}$ where $c_3\approx c_4$. By the evaluation rule, we have $\configThree{\Mem_2}{\while{\expr}{c_4}}{\strace}\xrightarrow{\emptyset} \configThree{\Mem_2}{\ifcmd{\expr}{c_4;\;\while{\expr}{c_4}}{\Skip}}{\strace}$.
   Hence, the conclusion is trivial by selecting $c_2'$ as $\ifcmd{\expr}{c_4;\;\while{\expr}{c_4}}{\Skip}$.

  \item case $\configThree{\Mem_1}{x:=\relabelusing{\expr}{\lab_f}{\level_t}{\sevent{s_0}\dots\sevent{s_k}}}{\strace}\xrightarrow{\emptyset} \configThree{\Mem_1[x:=n]}{\Skip}{\strace}$:\\
  From $c_1\approx c_2$, we know $c_2$ has the same form as $c_1$. From the evaluation rule, we have $\configs{\Mem_1,\expr}\Downarrow n$. From Lemma~\ref{lemma:eppe} (expression equivalence) and the assumption that $\Mem_1\approx_\strace \Mem_2$, we know $\configs{\Mem_2,\expr}\Downarrow n$. Further, we have $\configThree{\Mem_2}{x:=\relabelusing{\expr}{\lab_f}{\level_t}{\sevent{s_0}\dots\sevent{s_k}}}{\strace}\xrightarrow{\emptyset} \configThree{\Mem_2[x:=n]}{\Skip}{\strace}$.  
  Hence, $\Mem_1[x:=n]\approx_\strace \Mem_2[x:=n]$ regardless of the label of $x$. $\Skip\approx \Skip$ is trivial. The results on $\trace_1$ and $\trace_2$ are also trivial as the relabel command does not generate any events.
  
  \item case $\configThree{\Mem_1}{x:=\relabelusing{\expr}{\lab_f}{\level_t}{\sevent{s_0}\dots\sevent{s_k}}}{\strace}\xrightarrow{\emptyset} \configThree{\Mem_1[x:=[n]]}{\Skip}{\strace}$:\\
  From $c_1\approx c_2$, we know $c_2$ has the same form as $c_1$. There are two corresponding evaluation rules~\ruleref{OS-Relabel-r1} and ~\ruleref{OS-Relabel-r2}. The proofs are the same as the cases of ~\ruleref{OS-Assign-b1} and ~\ruleref{OS-Assign-b2} respectively.
  
   \item case $\configThree{\Mem_1}{x:=\relabelusing{\expr}{\lab_f}{\level_t}{\sevent{s_0}\dots\sevent{s_k}}}{\strace}\xrightarrow{\emptyset} \configThree{\Mem_1}{\Skip}{\strace}$:\\ 
   From $c_1\approx c_2$, we know $c_2$ has the same form as $c_1$. Since both $c_1$ and $c_2$ are evaluated under the same $\strace$, we know that evaluating $c_2$ under $\strace$ also follows Rule~\ruleref{RelabelFail}, and hence,  $\configThree{\Mem_2}{ x:=\relabelusing{\expr}{\lab_f}{\level_t}{\sevent{s_0}\dots\sevent{s_k}}}{\strace} \xrightarrow{\emptyset} \configThree{\Mem_2}{\Skip}{\strace}$. $\Skip\approx\Skip$ is trivial. The results on $\trace_1$ and $\trace_2$ are also trivial as relabel does not modify the trace.
  
  \item case: ${\configThree{\Mem_1}{\outcmd{\level'}{\expr}{\sevent{s_0}\dots\sevent{s_k}}}{\strace} \xrightarrow[]{\configThree{\level'}{v_1}{\strace}} \configThree{\Mem_1}{\Skip}{\strace'}}$:\\
  
  Based on the assumption that  $c\neq \outcmd{\level'}{x}{\sevent{o}_\x, \sevent{s_0}\dots \sevent{s_k}}$, we know that only the typing rule~\ruleref{Output} is applicable in this case. 
  
  As $c_1\approx c_2$, we have $c_2=\outcmd{\level'}{\expr}{\sevent{s_0}\dots\sevent{s_k}}$ and $${\configThree{\Mem_2}{\outcmd{\level'}{\expr}{\sevent{s_0}\dots\sevent{s_k}}}{\strace} \xrightarrow[]{\configThree{\level'}{v_2}{\strace}} \configThree{\Mem_2}{\Skip}{\strace'} }$$ 
  It is easy to show that $\Skip\approx\Skip$ and $\Mem_1'=\Mem_1\approx_\strace' \Mem_2=\Mem_2'$. To prove $\proj{\trace_1}_\fixedlevel=\proj{\trace_2}_\fixedlevel$, there are two cases, $\level'\LEQ\fixedlevel$ and $\level'\not\LEQ\fixedlevel$.
  \begin{itemize}
      \item  $\level'\LEQ\fixedlevel$: We show that $\configs{\Mem,\expr}\Downarrow n$ for some $n$ without brackets in this case. We proceed by assuming that $\configs{\Mem,\expr}\Downarrow[n]$ and show contradiction.
      
      From typing rule~\ruleref{Output}, we know that $\Gamma\vdash \expr:\lab$, $\sevent{s_0}\dots\sevent{s_k}\vdash\lab\releaseto\level'$. From Lemma~\ref{lem:highexpr} and the assumption that $\configs{\Mem,\expr}\Downarrow[n]$, we have $\fixedlab \LEQ_\strace \lab$, which implies $\auxfuncold{B}{\strace}\LEQ\auxfuncold{\lab}{\strace}$ by definition. By Theorem~\ref{theorem:setsound_no_extend} and $\sevent{s_0}\dots\sevent{s_k}\vdash\lab\releaseto\level'$, we know $\auxfuncold{\lab}{\strace}\LEQ\level'$, which when combined with $\auxfuncold{B}{\strace}\LEQ\auxfuncold{\lab}{\strace}$ can derive $\auxfuncold{B}{\strace}\LEQ L'$. By the assumption that $\level'\LEQ \fixedlevel$, we get $\auxfuncold{B}{\strace}\LEQ \fixedlevel$, which contradicts the assumption that $\auxfuncold{B}{\strace}\not\LEQ \fixedlevel$ in the theorem statement.
      
      Hence, it must be true that $\configs{\Mem_1,e}\Downarrow n$ for some $n$. From Lemma~\ref{lemma:eppe}, we have $\configs{\Mem_2,\expr}\Downarrow n$ too. Hence, both output traces under $\Mem_1$ and $\Mem_2$ produce the same output event $\configs{L',n,\strace}=\configs{L',n,\strace}$. 
      
      \item  $\level'\not\LEQ\fixedlevel$: The projection traces on $\fixedlevel$ are empty $\proj{\traceout{\tau_1}}_{\fixedlevel}=\emptyset=\proj{\traceout{\tau_2}}_{\fixedlevel}$ by definition of $\proj{}_\fixedlevel$ (Definition~\ref{def:proj}). Also, the security event traces are still identical has they remain the same.
  \end{itemize}
  
   \item case ${\configThree{\Mem_1}{\outcmd{\level'}{\expr}{\sevent{s_0}\dots\sevent{s_k}}}{\strace} \xrightarrow[]{\emptyset} \configThree{\Mem_1}{\Skip}{\strace}}$:\\ 
   From $c_1\approx c_2$, we know $c_2$ has the same form as $c_1$. Since $c_2$ is executed under the same event trace $\strace$, we know that ${\configThree{\Mem_2}{\outcmd{\level'}{\expr}{\sevent{s_0}\dots\sevent{s_k}}}{\strace_2} \xrightarrow[]{\emptyset} \configThree{\Mem_2}{\Skip}{\strace_2}}$ as well. Hence, we can easily check $\Mem_2'=\Mem_2\approx_\strace \Mem_1 = \Mem_1'$ and $\Skip\approx\Skip$. The results on $\trace_1$ and $\trace_2$ are also trivial as they are both empty.
  
  \item case $\configThree{\Mem_1}{\cod{eventon}(\sevent{s})}{\strace}\rightarrow \configThree{\Mem_1} {\Skip}{\strace\cdot\sevent{s}}$:\\
     From $c_1\approx c_2$, $c_2$ has the same form of $c_1$. From evaluation rule, we have $\configThree{\Mem_2}{\cod{eventon}(\sevent{s})}{\strace}\rightarrow \configThree{\Mem_2} {\Skip}{\strace\cdot\sevent{s}}$. Hence, it is easy to check that $\Mem_2'=\Mem_2\approx \Mem_1=\Mem_1'$ and $\Skip\approx\Skip$. Moreover, both traces produce the same security event $\sevent{s}$, and produce $\strace\cdot\sevent{s}$.
     Since there are no outputs, $\proj{\traceout{\tau_1}}_\fixedlevel=\emptyset=\proj{\traceout{\tau_2}}_\fixedlevel$.
     
  \item case $\configThree{\Mem}{\cod{eventoff}(\sevent{s})}{\strace_1}\rightarrow \configThree{\Mem} {\Skip}{\strace_1\cdot\neg\sevent{s}}$: Similar to the previous case.
\end{itemize}

\end{proof}

\begin{theorem*}[End-to-end Soundness, Theorem~\ref{theorem:typesound}]
For all commands c and memory $\Mem$, an arbitrary source label $\fixedlab$, an arbitrary attacker's level $\fixedlevel$,
if $pc, \Gamma \vdash c$ for some program counter $pc$ and typing environment $\Gamma$, then the program c satisfies dynamic release, which means that

	\begin{multline*}
	\forall \Mem, 
	\tau.~\configThree{\Mem}{c}{\strace} \termout \trace  
	\Rightarrow 
	\forall 1\leq i\leq \len{\trace }.\quad \\
	k_2(c,\trace^{[:i]}, \fixedlevel, \fixedlab) 
	\supseteq 
	\begin{cases} 
	\closure{\Mem}_{\neq \fixedlab}, & \text{transient} \\
	\closure{\Mem}_{\neq \fixedlab} \cap k_1(c, \trace^{[:{i-1}]},\fixedlevel), & 
	\text{persistent}
	\end{cases}
	\end{multline*}
\end{theorem*}

\begin{proof}
    We conduct a proof by induction on the trace length for both transient and persistent policies. Note that unlike a traditional proof for static noninterference, we need to show that under two indistinguishable initial memories, they both produce \emph{consistent}, rather than indistinguishable traces. According to its definition (Definition~\ref{def:consistency_appendix}), proving consistency boils down to filtering out sub-traces where $\auxfuncold{\fixedlab}{\strace}\LEQ \fixedlevel$ (i.e., when information with label $\fixedlab$ can be released), and showing the remaining sub-traces are indistinguishable.

\mypara{Transient policy.} We approach the proof by showing that a slightly stronger result holds:
\begin{multline}
	\forall \Mem,\Mem', \tau, c, \pc, \Gamma, \strace.~\wellformmem{\strace}{\Mem} \land \wellformmem{\strace}{\Mem'}\land \pc, \Gamma \vdash c \land \configThree{
		\Mem}{c}{\strace} \termout \trace \land \Mem\approx \Mem'  
	\Rightarrow \\
	\forall 0\leq i\leq \len{\trace }.~\Mem'\in k_2(c,\trace ^{[:i]}, \fixedlevel, \fixedlab) 
\end{multline}
The reason is that $\forall \Mem'\in \closure{\Mem}_{\neq \fixedlab}$ (i.e., only the values of variables with label $\fixedlab$ differ), $\Mem'\approx \Mem$ by definition of $\approx$.

Further, by expending the definition of $k_2$, we argue that it is sufficient to prove the following condition.
\begin{multline}
\label{eqn:transcondition}
	\forall \Mem,\Mem', \tau, \tau', c, \pc,\Gamma,\strace.~\wellformmem{\strace}{\Mem} \land \wellformmem{\strace}{\Mem'} \land \pc, \Gamma\vdash c\land  \configThree{
		\Mem}{c}{\strace} \termout \trace \land \configThree{
		\Mem'}{c}{\strace} \termout \trace' \land \Mem\approx \Mem'  
	\Rightarrow \\
	\forall 1\leq i\leq \len{\trace }.~\exists j.~\trace'^{[:j]} \equiv_{\fixedlab,\fixedlevel}  \trace^{[:i]}.
\end{multline}
The reason is that if the claim is correct, then 
\begin{align}
\Mem' &\in k_1(c, \tau'^{[:j]}, \fixedlevel) \\
      &\subseteq \bigcup_{\exists \Mem'',~k.~\configs{\Mem'', c} \termout 
\trace'' \AND \trace''^{[:k]} \equiv_{\fixedlab,\fixedlevel}  \trace^{[:i]} } k_1(c, \trace''^{[:k]}, \fixedlevel )\\
      & \triangleq k_2(c,\trace ^{[:i]}, \fixedlevel, \fixedlab)
\end{align}
where $(3)$ is true since $\configThree{\Mem'}{c}{\strace} \termout \trace'$, and $(4)$ is true since we know that $\configThree{\Mem'} {c}{\strace} \termout 
\trace' \AND \trace'^{[:j]} \equiv_{\fixedlab,\fixedlevel}  \trace^{[:i]}$.

Next, we perform induction on $i$ to prove that Equation~(\ref{eqn:transcondition}) holds.

\begin{figure}
    \centering
    \includegraphics[width=0.5\columnwidth]{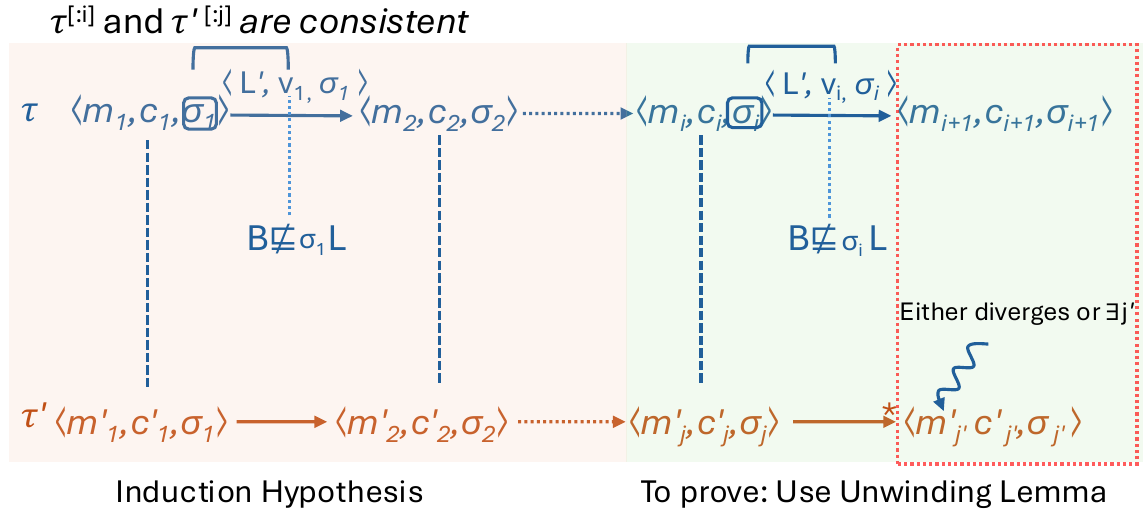}
    \caption{The orange part is the induction hypothesis for the transient policy statement, and the green part is the step to be proved with the Unwinding Lemma.}
    \label{fig:theorem3_figure}
    \Description{}
\end{figure}
 
\begin{itemize}
    \item Base case: This case is trivial, as any trace is consistent with an empty trace by definition. 

    \item Inductive case: Based on the induction hypothesis, we know that at step $i$, $\exists j.~\trace'^{[:j]} \equiv_{\fixedlab,\fixedlevel}  \trace^{[:i]}$. We need to prove that at step $i+1$, $\exists j'.~\trace'^{[:j']} \equiv_{\fixedlab,\fixedlevel} \trace^{[:i+1]}$. Next, we use $\strace_{i}$ to denote $\trace^{[:i]}$ and $\strace_{i+1}$ to denote $\trace^{[:i+1]}$.
    The proof is separated into two cases:
    
    \begin{itemize}
        \item Case where $\auxfuncold{\fixedlab}{\strace_i}\not\LEQ \fixedlevel$:
        Following Figure~\ref{fig:theorem3_figure}, the green part illustrate the new index $j'$ that we need to find in order to match index $i+1$ on the first trace, while the orange part is given from the induction hypothesis: the two traces $\trace^{[:i]}$ and $\trace'^{[:j]}$ are consistent. The main result we use to find $j'$ to finish the green part is the unwinding lemma (Lemma~\ref{lemma:unwinding_dynamicrelease}).
        
        From the induction hypothesis, we also have that $\wellformmem{\strace_i}{\Mem_i}$, $\wellformmem{\strace_i}{\Mem_i'}$, and $\Mem_i\approx\Mem_i'$ at step $i$. Moreover, by Lemma~\ref{lemma:preservation} (Preservation), we have $\pc,\Gamma\vdash c_i$ and $\pc,\Gamma\vdash c'_j$ (i.e., the programs $c_i$ and $c'_j$ remain well-typed during evaluation). For the extra event $v_i$ produced by $\configThree{\Mem_i}{c_i}{\strace_i}$ (in the green part of Figure~\ref{fig:theorem3_figure}), we know that by Definition~\ref{def:secretproj_modified_appendix},
        \[\proj{\configThree{
		\Mem_i}{c_i}{\strace_i}\xrightarrow{\configThree{\level'}{v}{\strace_i}}\configThree{
		\Mem_{i+1}}{c_{i+1}}{\strace_{i+1}}}_{\fixedlab, \fixedlevel}~\triangleq 
        \begin{cases}
            \{v\} \text{, when $L'\LEQ \fixedlevel$} \\
            \emptyset \text{, when $L'\not\LEQ \fixedlevel$}
        \end{cases}\]

    \begin{itemize}
        \item In the first case where $L'\LEQ \fixedlevel$, from Lemma~\ref{lemma:unwinding_dynamicrelease}, either the program $\configThree{
		\Mem'_{j'}}{c'_{j'}}{\strace_{j'}}$ diverges or, there exists $j'$ such that  
        \[ \configThree{\Mem_j'} {c'_j}{\strace_j} \xrightarrow{\tau'}^{*}\configThree{\Mem_{j'}'}{c'_{j'}}{\strace_{j'}}\land \Mem_{i+1}\approx \Mem_{j'}'\land c_{i+1}\approx c'_{j'}\land \proj{\configThree{\level'}{v}{\strace_i}}_\fixedlevel=\proj{\trace'^{[j:j']}}_\fixedlevel\]

        Since $L'\LEQ \fixedlevel$, we have $\proj{\trace'^{[j:j']}}_\fixedlevel=\configThree{\level'}{v}{\strace_i}$, and hence, $\proj{\trace'^{[j:j']}}_{\fixedlab,\fixedlevel}=\configThree{\level'}{v}{\strace_i}=\proj{\trace^{[i:i+1]}}_{\fixedlab,\fixedlevel}$ as $\auxfuncold{\fixedlab}{\strace_i}\not\LEQ \fixedlevel$. By induction hypothesis and the fact that the segment $\trace^{[i:i+1]}$ and $\trace^{[j:j']}$ are also consistent, we know that $\trace'^{[:j']} \equiv_{\fixedlab,\fixedlevel} \trace^{[:i+1]}$ stands. 
        
          In the case where the program $\configThree{
		\Mem'_{j'}}{c'_{j'}}{\strace_{j'}}$ diverges, we note that the dynamic release policy is defined as a termination insensitive policy. Hence, the policy incurs no restrictions in this case.
        
        \item In the second case where $L'\not\LEQ \fixedlevel$, no output observable to the attacker at level $\fixedlevel$ is produced. Hence, we simply set $j'=j$ and have $\trace'^{[:j']} \equiv_{\fixedlab,\fixedlevel} \trace^{[:i+1]}$ as
        $\trace'^{[:j']} =\trace'^{[:j]}$, $\trace^{[:i]}=\trace'^{[:i+1]}$ and from the induction hypothesis, $\trace'^{[:j]} \equiv_{\fixedlab,\fixedlevel} \trace^{[:i]}$.
        
    \end{itemize}

        \item Case where $\auxfuncold{\fixedlab}{\strace_i}\LEQ \fixedlevel$:
        By the induction hypothesis, the two traces $\trace^{[:i]}$ and $\trace'^{[:j]}$ are consistent.
        For the extra event $v_i$ produced by $\configThree{\Mem_i}{c_i}{\strace_i}$ (in the green part of Figure~\ref{fig:theorem3_figure}), we know that by Definition~\ref{def:secretproj_modified_appendix},
        \[\proj{\configThree{
		\Mem_i}{c_i}{\strace_i}\xrightarrow{\configThree{\level'}{v}{\strace_i}}\configThree{
		\Mem_{i+1}}{c_{i+1}}{\strace_{i+1}}}_{\fixedlab, \fixedlevel}~\triangleq 
        \emptyset\]
        since by definition, the projection only includes events produced when $\auxfuncold{\fixedlab}{\strace_i}\not\LEQ\fixedlevel$.

        Hence, we simply set $j'=j$ and have $\trace'^{[:j']} \equiv_{B,\level} \trace^{[:i+1]}$ as $\trace'^{[:j']} =\trace'^{[:j]}$, $\trace^{[:i]}=\trace'^{[:i+1]}$ and from the induction hypothesis, $\trace'^{[:j]} \equiv_{\fixedlab,\fixedlevel} \trace^{[:i]}$.
    \end{itemize}
\end{itemize}

\mypara{Persistent policy}
We approach the proof by showing that a slightly stronger result holds:
\begin{multline*}
	\forall \Mem,\Mem', \trace, \trace', c, \pc,\Gamma,\strace.~\wellformmem{\strace}{\Mem} \land \wellformmem{\strace}{\Mem'}\land \pc, \Gamma\vdash c \land  \configThree{
		\Mem}{c}{\strace} \termout \trace \land 
        \Mem\approx \Mem'  \Rightarrow \\
	\forall 0\leq i\leq \len{\trace }.~\Mem'\in k_1(c, \tau^{[:i-1]}, \fixedlevel)\Rightarrow ~\Mem'\in k_2(c,\trace ^{[:i]}, \fixedlevel, \fixedlab) 
\end{multline*}

The reason is similar to the transient policy case, except that a persistent policy allows previous outputs to be released once again by having $\Mem'\in k_1(c,\tau^{[:i-1]}, \fixedlevel)$ as an assumption before the conclusion $\Mem'\in k_2(c,\trace ^{[:i]}, \fixedlevel, \fixedlab)$.

Further, by expending the definition of $k_1$ and $k_2$, we argue that it is sufficient to show the following 
\begin{multline}
\label{eqn:perscondition}
	\forall \Mem,\Mem', \tau, \tau', c, \pc,\Gamma,\strace.~\wellformmem{\strace}{\Mem} \land \wellformmem{\strace}{\Mem'} \land \pc, \Gamma\vdash c \land  \configThree{
		\Mem}{c}{\strace} \termout \trace \land \configThree{
		\Mem'}{c}{\strace} \termout \trace' \land \Mem\approx \Mem'  \Rightarrow \\
	\forall 1\leq i\leq \len{\trace }.~\proj{\tau^{[:i-1]}}_\fixedlevel \preceq
\proj{\tau'}_\fixedlevel \Rightarrow \exists j.~\trace'^{[:j]} \equiv_{\fixedlab,\fixedlevel}  \trace^{[:i]}.
\end{multline}

The reason is that 
\begin{enumerate}
    \item We can rewrite $\Mem'\in k_1(c,\tau^{[:i-1]}, \fixedlevel)$ to $\proj{\tau^{[:i-1]}}_\fixedlevel \preceq
\proj{\tau'}_\fixedlevel$ by the definition of $k_1$. 

    \item Similar to Equations (3)-(5), we can replace the requirement that $\Mem'\in k_2(c,\trace ^{[:i]}, \fixedlevel, \fixedlab)$ with $\trace'^{[:j]} \equiv_{\fixedlab,\fixedlevel}  \trace^{[:i]}$.
\end{enumerate}

\begin{figure}
    \centering
    \includegraphics[width=0.5\columnwidth]{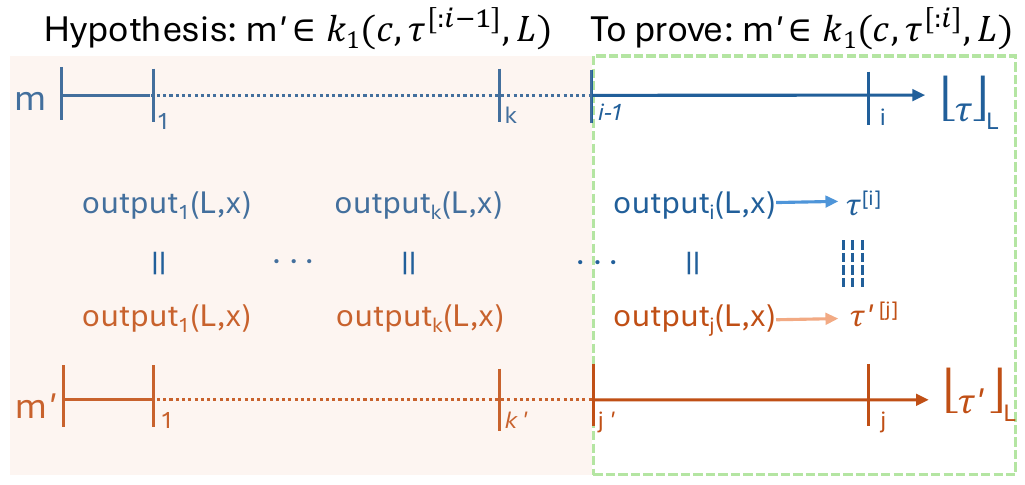}
    \caption{This Figure shows that $\Mem'\in k_1(c,\trace^{[:i],\level})$, which indicates the variable $x$ has previously been output. This difference causes the proof not to be concluded with Unwinding Lemma.}
    \Description{}
    \label{fig:theorem3_figure_per}
\end{figure}

Next, we perform induction on $i$ to prove that Equation (\ref{eqn:perscondition}) holds.

The proof is mostly identical to the transient case as Equation~\ref{eqn:perscondition} is a relaxed version of Equation~\ref{eqn:transcondition}. The exception is for 
the case of $$(c=\outcmd{\level}{x}{\sevent{o}_\x,\sevent{s_0\dots s_k}})\land \x \text{ is immutable})$$
because we cannot use the unwinding lemma (Lemma~\ref{lemma:unwinding_dynamicrelease}) to derive low equivalence on traces when the $i$-th step on $\tau$, from index $i$ to $i+1$. We need a new proof strategy for this case, which is illustrated by the green dotted box in Figure~\ref{fig:theorem3_figure_per}. 

When evaluation rule~\ruleref{Output-Fail} is applied at index $i$, the command has the same semantics as a $\Skip$ command. Hence, we simply set $j'=j$ and have $\trace'^{[:j']} \equiv_{\fixedlab,\fixedlevel} \trace^{[:i+1]}$ as
        $\trace'^{[:j']} =\trace'^{[:j]}$, $\trace^{[:i]}=\trace'^{[:i+1]}$ and from the induction hypothesis, $\trace'^{[:j]} \equiv_{\fixedlab,\fixedlevel} \trace^{[:i]}$.

When evaluation rule~\ruleref{Output-Succ} is applied at index $i$, we know that $x$ has been released at some level $\level'\LEQ \fixedlevel$ in the past. That is, an output $\configThree{\level'}{v}{\strace} \in \proj{\tau^{[:i-1]}}_\fixedlevel$ for some value $v$ and $\strace$ was produced by an output command $\outcmd{\level'}{x}{\sevent{s_0\dots s_k}}$ executed on trace $\tau^{[:i-1]}$, where $\level'\LEQ \fixedlevel$.  Hence, on trace $\tau$, step $i$ produces an output event $\configThree{\level'}{v}{\strace'}$ for some $\strace'$. Note that since $x$ is immutable, the output value must be the same as the previous output $v$. 

On trace $\trace'$, by induction hypothesis, there must be some $j'$ such that $\trace'^{[:j']} \equiv_{\fixedlab,\fixedlevel}  \trace^{[:i-1]}$. Due to the assumption that $\proj{\tau^{[:i-1]}}_\fixedlevel \preceq \proj{\tau'}_\fixedlevel$, we know that there must be an output $\configThree{\level'}{v}{\strace''} \in \proj{\trace'^{[:j']}}_\fixedlevel$ for some $\strace''$ produced by an output command $\outcmd{\level'}{x}{\sevent{s_0\dots s_k}}$ executed on trace $\trace'^{[:j']}$. This is illustrated in Figure~\ref{fig:theorem3_figure_per}: the orange part shows that before index $i-1$, the projection on $\fixedlevel$ has $k$ output events, which are identical to the first $k$ output events on trace $\proj{\tau'}_\fixedlevel$ due to the induction hypothesis.
Hence, we let $j=j'+1$. On trace $\trace'$, step $j$
produces an output event $\configThree{\level'}{v}{\strace'}$. Again, since $x$ is immutable, the output value must be the same as the previous output $v$. Since the outputs are the same at $\trace^{[i]}$ and $\trace'^{[j]}$, the $\trace^{[:i]}$ and $\trace'^{[:j]}$ are still consistent $\trace'^{[:j]} \equiv_{\fixedlab,\fixedlevel}  \trace^{[:i]}$, as illustrated in the green part of Figure~\ref{fig:theorem3_figure_per}.
\end{proof}

\end{document}
\endinput